\documentclass[
10pt,
aps,
prb,
twocolumn,
amsmath,amssymb,
superscriptaddress,
longbibliography,
nofootinbib,
notitlepage]{revtex4-2}
\usepackage[colorlinks=true,linkcolor=blue,anchorcolor=red,citecolor=blue, urlcolor=blue]{hyperref}
\usepackage{bm} 
\usepackage{graphicx}
\usepackage{adjustbox}

\usepackage{calc}
\usepackage{accents}
\newcommand{\dbtilde}[1]{\tilde{\raisebox{0pt}[0.85\height]{$\tilde{#1}$}}}

\usepackage{physics}
\usepackage{varwidth}
\usepackage{multirow}
\usepackage{stmaryrd}

\usepackage{diagbox}

\usepackage{color}
\usepackage{xcolor}
\definecolor{Lblue}{RGB}{56,108,176}      % deep blue
\definecolor{Lorange}{RGB}{214,126,44}    % burnt orange
\definecolor{Lteal}{RGB}{0,150,30}       % muted teal
\definecolor{Lgray}{RGB}{80,80,80}        % dark gray

\definecolor{myRed}{HTML}{F3BE90}   % red
\definecolor{myBlue}{HTML}{D6D6F5}

\definecolor{myGreen}{HTML}{009E73} % green
\definecolor{myGray}{HTML}{B3B3B3}

\usepackage{soul}
\usepackage{enumerate}
\usepackage{mathrsfs}

\usepackage{enumitem}

\hypersetup{
    colorlinks=true,
    linkcolor=blue,
    citecolor=magenta,
    filecolor=magenta,      
    urlcolor=blue,
    pdftitle={},
    pdfpagemode=FullScreen,
    }

\usepackage{tikz, pgfplots,tikz-cd}
\usetikzlibrary{calc, decorations.markings, decorations.pathmorphing,shapes}
\usetikzlibrary{arrows.meta,decorations.pathmorphing,decorations.markings,backgrounds,positioning,fit,petri}
\usepackage{amsthm}
\theoremstyle{definition}
\newtheorem{theorem}{Theorem}
\newtheorem{definition}{Definition}
\newtheorem{lemma}{Lemma}

\newcommand{\beq}{\begin{equation}\begin{aligned}}
\newcommand{\eeq}{\end{aligned}\end{equation}}

\begin{document}
\makeatletter
\def\l@subsubsection#1#2{}
\makeatother

\title{State preparation via measurement and feedback: \\
pushing relations, state structures, and non-invertible symmetries}
\author{Yabo Li}
\affiliation{
Center for Quantum Phenomena, Department of Physics,
New York University, 726 Broadway, New York, New York, 10003, USA
}
\author{Aditi Mitra}
\affiliation{
Center for Quantum Phenomena, Department of Physics,
New York University, 726 Broadway, New York, New York, 10003, USA
}
\author{Tsung-Cheng Lu}
\affiliation{
Joint Center for Quantum Information and Computer Science, University of Maryland, College Park, Maryland 20742, USA
}
\affiliation{Department of Physics and Astronomy, Center for Quantum Research and Technology,
University of Oklahoma, Norman, OK 73069, USA}

\begin{abstract}

Quantum circuits with measurements and unitary feedback (MF) can prepare long-range entangled states in constant depth, but a systematic construction of the MF preparation circuit for a given target state remains underexplored. We develop such a scheme for one-dimensional states, based on the notion of pushable defects: virtual-bond operators of a matrix product state that can be pushed through the tensor at the price of a physical feedback unitary. We show that the set of pushable defects, together with their pushing relations classifies finite-depth MF-preparable states and dictates their preparation circuits. To each class of the target state $\ket{A}$, we associate a state $\ket{B}$ from which $\ket{A}$ can be prepared using a 1-round MF circuit; in particular, $\ket{A}$ is preparable from a product state using a circuit with 1 round of MF whenever $\ket{B}$ is preparable by a finite-depth local unitary (FDLU) circuit. For a general target state, the scheme is obtained by iterating this procedure until the associated state is FDLU-preparable. For open-boundary matrix product states, the scheme is complete: it constructs a preparation circuit whenever finite-depth MF preparation with left-conditioned feedback corrections is possible. Pushable defects and pushing relations thus emerge as a unifying principle for quantum state preparation via measurements and feedback. This characterization further reveals an intrinsic connection between MF circuits and non-invertible symmetries: states with certain classes of pushing relations are related to a product state by Tambara-Yamagami duality operators, or by continuous cosine symmetry operators with fusion rules $L_\alpha L_{\alpha'} = L_{\alpha+\alpha'} + L_{\alpha-\alpha'}$, together with their generalizations up to (not necessarily transversal) gates.

\end{abstract}

\maketitle
{
  \hypersetup{linkcolor=black}
  \tableofcontents
}

\section{Introduction}
Preparing quantum states with desired properties is a central task across multiple subfields of quantum science. In quantum computation, state preparation is necessarily the first step in any information processing task, whereas in quantum many-body physics, the ability to prepare ground states, thermal states, or other states of interest is essential for probing phases of matter, characterizing quantum correlations, and exploring non-equilibrium dynamics.

A central challenge in state preparation is how to reduce the temporal overhead so that a target state can be prepared efficiently, ideally in a finite time that does not scale with the system size. Time efficiency is particularly important for ensuring the reliability and fidelity of state preparation, given the limited coherence time of near-term quantum devices. Unfortunately, the states prepared by finite-depth local unitary (FDLU) circuits from product states are quite limited due to the famous Lieb-Robinson theorem~\cite{lieb1972finite,bravyi2006lieb}. Specifically, they cannot exhibit long-range entanglement, and the spatial correlation must vanish at long distances.

Remarkably, the fundamental constraint of the Lieb-Robinson theorem can be bypassed with the inclusion of local measurements and unitary feedback (MF) conditioned on global measurement outcomes~\cite{Piroli2021quantum,bravyi2022adaptive,lu2022measurement,neumann2025adaptive,cao2025measurement,Buhrman2024statepreparation}\footnote{The MF circuits are also known as circuits assisted by LOCC (Local Operations and Classical Communication)~\cite{Piroli2021quantum}, adaptive circuits~\cite{bravyi2022adaptive}, or LAQCC (local alternating quantum classical computation) circuits~\cite{Buhrman2024statepreparation,neumann2025adaptive}.}. With MF circuits, Greenberger-Horne-Zeilinger (GHZ) state, toric-code topological order, and more generally, all Pauli stabilizer states can be prepared by finite-depth~\cite{Raussendorf_2001,cluster_mbqc_2005,Foliated_2016_css,jozsa2006introduction,Buhrman2024statepreparation}. More remarkably, certain non-abelian topological orders~\cite{bravyi2022adaptive,verresen2021efficiently,tantivasadakarn2024long,lu2022measurement,Shortest_non_abelian_ruben_2023,li2023symmetry,Hierarchy}, non-fixed-point wavefunctions~\cite{malz2024preparation,smith2024constant,sahay2025classifying,lu2025spacetime}, and intriguing mixed-state quantum matter~\cite{Nishimori_2023,Lu_mixedstate_preparation_2023} can also be prepared using MF circuits. Along with recent experimental capabilities of performing mid-circuit measurements, exploring what many-body quantum states can be prepared by finite-depth MF circuits has attracted significant interest in the field of quantum science~\cite{foss2023experimental,exp_abelian_2024_ruben,non_abelian_ruben_exp_2024,baumer2024efficient,baumer2024efficient}. In particular, such efforts significantly enhance the experimental realization of complex quantum systems on current quantum devices, and also provide a novel resource perspective characterization of quantum many-body states based on their preparability.

\begin{table*}[t]
\begin{tabular}{cccccc}
\hline
\multicolumn{1}{|c|}{Pushable defects for $A$}                 & \multicolumn{1}{c|}{Pushable defects for $B$}   & \multicolumn{1}{c|}{Pushing relations} & \multicolumn{1}{c|}{Rounds} & \multicolumn{1}{c|}{Unitary circuits}                & \multicolumn{1}{c|}{Sections} \\ \hline
\multicolumn{1}{|c|}{\multirow{4}{*}{All Pauli}}               & \multicolumn{1}{c|}{\multirow{4}{*}{\diagbox[dir=NE,width=3.4cm,height=1.4cm]{}{}}}    & \multicolumn{1}{c|}{$RIRI$, $RILI$}           & \multicolumn{1}{c|}{0}         & \multicolumn{1}{c|}{FDLU}                               & \multicolumn{1}{c|}{\ref{subsec:RIRI}, \ref{subsec:RILI}}        \\ \cline{3-6} 
\multicolumn{1}{|c|}{}                                         & \multicolumn{1}{c|}{}                           & \multicolumn{1}{c|}{$RIRX$}           & \multicolumn{1}{c|}{1}         & \multicolumn{1}{c|}{KW}                                 & \multicolumn{1}{c|}{\ref{subsec:RIRX}}        \\ \cline{3-6} 
\multicolumn{1}{|c|}{}                                         & \multicolumn{1}{c|}{}                           & \multicolumn{1}{c|}{$RZRX$}           & \multicolumn{1}{c|}{1}         & \multicolumn{1}{c|}{KT}                                 & \multicolumn{1}{c|}{\ref{subsec:RZRX}, \ref{subsec:example_fusion_measurement}}        \\ \cline{3-6}  
\multicolumn{1}{|c|}{}                                         & \multicolumn{1}{c|}{}                           & \multicolumn{1}{c|}{General}          & \multicolumn{1}{c|}{1}         & \multicolumn{1}{c|}{Tambara-Yamagami, ...}              & \multicolumn{1}{c|}{\ref{subsec:general}}        \\ \hline
\multicolumn{1}{|c|}{\multirow{4}{*}{Only $Z$}}                & \multicolumn{1}{c|}{\multirow{4}{*}{All Pauli}} & \multicolumn{1}{c|}{$RI$--FDLU}     & \multicolumn{1}{c|}{0}         & \multicolumn{1}{c|}{FDLU}                               & \multicolumn{1}{c|}{\ref{subsec:RI-FDLU}}        \\ \cline{3-6} 
\multicolumn{1}{|c|}{}                                         & \multicolumn{1}{c|}{}                           & \multicolumn{1}{c|}{$RI$--$RIRX$}     & \multicolumn{1}{c|}{1}         & \multicolumn{1}{c|}{Brick-wall$\cdot \text{KW}$ }               & \multicolumn{1}{c|}{\ref{subsec:RI-RIRX}, \ref{subsec:example_RI-RIRX}}        \\ \cline{3-6} 
\multicolumn{1}{|c|}{}                                         & \multicolumn{1}{c|}{}                           & \multicolumn{1}{c|}{$RZ$--FDLU}     & \multicolumn{1}{c|}{1}         & \multicolumn{1}{c|}{Brick-wall$\cdot \text{KW}$}               & \multicolumn{1}{c|}{\ref{subsec:RZ-FDLU}, \ref{subsec:example_RI-RIRX}}        \\ \cline{3-6} 
\multicolumn{1}{|c|}{}                                         & \multicolumn{1}{c|}{}                           & \multicolumn{1}{c|}{$RZ$--$RIRX$}     & \multicolumn{1}{c|}{2}         & \multicolumn{1}{c|}{Cosine symmetries, ...}             & \multicolumn{1}{c|}{\ref{subsec:RZ-RIRX}, \ref{subsec:example_RZ-RIRX}}        \\ \hline
\multicolumn{1}{|c|}{\multirow{3}{*}{All Pauli in subspace}} & \multicolumn{1}{c|}{\multirow{3}{*}{All Pauli}} & \multicolumn{1}{c|}{$RIRX$--FDLU}   & \multicolumn{1}{c|}{1}         & \multicolumn{1}{c|}{$\text{KW}\cdot$FDLU, ...}                 & \multicolumn{1}{c|}{\ref{subsec:RIRIX-FDLU}}        \\ \cline{3-6} 
\multicolumn{1}{|c|}{}                                         & \multicolumn{1}{c|}{}                           & \multicolumn{1}{c|}{$RIRI$--$RIRX$}   & \multicolumn{1}{c|}{1}         & \multicolumn{1}{c|}{FDLU$\cdot \text{KW}$, ...}                 & \multicolumn{1}{c|}{\ref{subsec:RIRI-RIRX}}        \\ \cline{3-6} 
\multicolumn{1}{|c|}{}                                         & \multicolumn{1}{c|}{}                           & \multicolumn{1}{c|}{$RIRX$--$RZRX$}   & \multicolumn{1}{c|}{2}         & \multicolumn{1}{c|}{$\text{KW}\cdot$transversal$\cdot \text{KW}$, ...} & \multicolumn{1}{c|}{\ref{subsec:RIRX-RZRX}} \\   \hline

\end{tabular}
\caption{Summary of results. We classify the preparability of a target MPS $\ket{A}$ based on the pushing relations of its defects and the associated MPS $\ket{B}$, from which $\ket{A}$ can be prepared via measurement and feedback. The notation $RaRb$ means the Pauli $Z$ defect is right-pushed to $a$, whereas the Pauli $X$ defect is right-pushed to $b$. Similarly, $RaLb$ means $Z$ is right-pushed to $a$ and $X$ is left-pushed to $b$. The MPS $\ket{A}$ can also be prepared using unitary circuits, a subset of which are sequential circuits that implement non-invertible symmetries. In particular, when all Pauli defects of $\ket{A}$ are pushable, $\ket{A}$ is fusion-measurement preparable~\cite{smith2024constant,sahay2025classifying,stephen2024preparing,zhang2024characterizing}, and we show that such states can also be obtained from product states using various classes of non-invertible symmetry transformations based on different pushing relations. The ``general'' pushing relation in the fourth row denotes the class of states for which the virtual space is decomposed into multiple subspaces, and all the Pauli defects in all subspaces are pushable to Pauli defects. The ``all Pauli in subspace'' class in the last row denotes the class of states for which all Pauli defects in one subspace are pushable, whereas defects in other subspaces are not necessarily pushable.} 
\label{tab:1}
\end{table*}

The tensor-network formalism provides a natural perspective for characterizing many-body states that can be prepared by finite-depth MF circuits~\cite{lu2022measurement,sahay2025classifying,stephen2024preparing,zhang2024characterizing,gunn2025phases,gunn2026phases}. In one spatial dimension, starting from a product state, the output of any finite-depth MF circuit is necessarily a matrix product state (MPS) with finite bond dimension. As such, the essential question is: how to characterize those MPSs preparable by finite-depth MF circuits? What are their defining properties? Refs.~\cite{smith2024constant,sahay2025classifying,stephen2024preparing,zhang2024characterizing} analyze a specific class of MF circuits known as the fusion-measurement protocol. In this approach, one starts with extensive decoupled clusters (which are FDLU preparable), performs local Bell-basis measurements to contract the virtual bonds of neighboring MPS tensors, and applies conditional feedback to compensate for unwanted measurement outcomes. The MPSs preparable by fusion measurement have been characterized: they either have flat entanglement spectra or zero correlation lengths for some operators~\cite{sahay2025classifying}. Moreover, the local tensor of the MPS must exhibit a Clifford structure~\cite{zhang2024characterizing}. Equivalently, these MPSs are related to product states via a sequential Clifford circuit followed by a transversal unitary circuit, up to disentangled ancillary degrees of freedom. The Clifford structure is particularly important; under Bell-basis measurements, various Pauli defects ($X, Y, Z$) will appear in the virtual bonds of the target MPS, and the Clifford structure of the local tensor allows those defects to be pushed to physical legs where they can be annihilated by operations acting on the physical qubits. This ensures the deterministic preparation of the MPS under the fusion measurement protocol.\footnote{It is deterministic when the target MPS is defined with open boundary conditions. For MPSs with closed boundary conditions, the preparation is achieved with a non-zero probability that is constant in the system size. More details will be discussed in Sec. \ref{sec:MF}.}

While the fusion-measurement protocol provides a promising framework for constructing MF circuits, it is still too restrictive because it demands all possible Pauli defects to be pushable in the MPS tensor. In this work, we propose a general scheme for preparing translation-invariant MPSs using MF circuits that go beyond fusion measurement. As in Ref.~\cite{smith2024constant,sahay2025classifying,stephen2024preparing,zhang2024characterizing}, the unwanted measurement outcomes in our scheme correspond to pushable defects in the target MPS, but importantly, we do not demand all Pauli defects to be pushable. For instance, even when Pauli-$Z$ is the only pushable defect, one may still construct an MF circuit that prepares a corresponding MPS. By systematically classifying the possible pushable defects and the associated pushing relations, we will explore the general landscape of MPSs preparable by MF circuits. In particular, our construction provides a natural framework to classify  MPSs based on the number of measurement-feedback rounds required; see Table \ref{tab:1}.

When all Pauli defects are pushable in the target MPS, our scheme leads to a single-round MF preparation circuit, which is equivalent to fusion measurements~\cite{smith2024constant,sahay2025classifying,stephen2024preparing,zhang2024characterizing}. In this case, by systematically classifying the pushing relations, we provide a fine-grained classification of fusion measurement preparable MPSs (see Table \ref{tab:1}). Notably, our classification reveals an intriguing connection between fusion measurement and non-invertible symmetries—a class of generalized symmetries that has recently attracted attention in quantum many-body physics for understanding exotic quantum phases and for the practical design of logical gates in fault-tolerant quantum computation~\cite{thorngren2024fusion,thorngren2024fusion2,aasen2020topological,shao2023s,lootens2023dualities,lootens2025low,tantivasadakarn2024long,Hierarchy}. In particular, states classified by different pushing relations can be prepared from product states via distinct classes of non-invertible duality transformations, including Kramers–Wannier (KW)~\cite{kramers1941statistics,onsager1944crystal,Kaufman1949,kogut1979introduction,seiberg2023majorana,chen2023sequential,tantivasadakarn2024long}, Kennedy–Tasaki (KT)~\cite{kennedy1992hidden,kennedy1992hidden2,doherty2009identifying,li2023non,seifnashri2024cluster}, and the more general Tambara–Yamagami duality symmetries~\cite{tambara1998tensor,aasen2020topological,seifnashri2024cluster}. In addition to the fine-grained classification of the fusion-measurement preparable states and the connection to non-invertible symmetries, we consider a more general setting in which defects can be pushed either to the left or to the right, extending the framework of Ref.~\cite{smith2024constant,sahay2025classifying,stephen2024preparing,zhang2024characterizing,sahay2024finite}. It turns out that this generalization does not lead to additional capability in state preparation, addressing one of the open questions raised in Ref.~\cite{sahay2025classifying}.

By classifying the pushable defect of a target MPS, our framework naturally gives rise to a preparation scheme that involves multiple rounds of MF circuits and FDLU circuits (Fig.~\ref{fig:circuit}). As an immediate application, we discuss the preparation of an MPS that requires two rounds of MF circuits, and notably, reveal a connection between the preparation protocol and the implementation of the cosine symmetry, a more general class of non-invertible symmetries that include KW and KT as special cases~\cite{chang2021lorentzian,thorngren2024fusion2,aasen2020topological,hsin2024fractionalization,seifnashri2025gauging,cao2025global}.

This paper is organized as follows: in Sec.~\ref{sec:FDLU} we discuss the MPSs preparable by FDLU and its transfer matrix structure. In Sec.~\ref{sec:MF}, we discuss the general measurement-feedback circuit architecture. In Sec.~\ref{sec:all_pauli}, we consider the case where the MPS tensor of the target state possesses a complete set of pushable Pauli defects; we classify their structures and establish a connection to non-invertible symmetries.  In Sec.~\ref{sec:pauli_z} and Sec.~\ref{sec:Pauli_subspace}, we relax the conditions imposed on pushable defects so that they do not need to form a complete set, allowing us to prepare more general classes of states, many of which are beyond the fusion-measurement framework. In particular, in Sec.~\ref{sec:pauli_z} we consider the case where only the Pauli $Z$ defect is pushable, and in Sec.~\ref{sec:Pauli_subspace} we consider the case where only the Pauli defects in a subspace of the entire virtual space of the MPS tensor are pushable. In Sec.~\ref{sec:examples} we present several examples; these include certain imaginary-time-deformed states, states that have spontaneously broken non-invertible symmetries, and states with non-onsite symmetries. In Sec.~\ref{sec:outlook} we summarize our work and highlight future questions.

\section{State preparation with a finite-depth local unitary (FDLU)}
\label{sec:FDLU}

Before introducing our measurement and feedback preparation scheme, let us first discuss the preparation of quantum states via FDLU circuits to familiarize the readers with notations and techniques used throughout this work. In this section, we will show that a necessary and sufficient condition for an MPS to be FDLU-preparable can be inferred from the spectrum of its transfer matrix (Lemma \ref{lemma:FDLU}).

Suppose that state $\ket{\Psi}$ is prepared from the product state $\ket{0}^{\otimes N}$ using an FDLU, due to the Lieb-Robinson bound, $\ket{\Psi}$ must be short-range entangled and therefore can be represented as a matrix product state (MPS) $\ket{A}$ with a finite bond dimension. In the following discussion, when referring to a translation-invariant MPS $\ket{A}$, we implicitly refer to a family of states $\{\ket{A}_N\}$ with system size $N\in \mathbb{N}$. We require that the properties of $\ket{A}_N$ are insensitive to $N$ in order to take the thermodynamic limit. The Hilbert space of a translation-invariant one-dimensional (1D) system is given by $\mathcal{H}=\bigotimes_{i=1}^N\mathbb{C}^{d}$, where each site hosts a $d$-dimensional degree of freedom. A translation-invariant MPS is given by a local tensor $A$ as follows,
\begin{equation}
\begin{aligned}
\ket{A}=\sum_{\{i_n\}}&\Tr\bigg(
\begin{tikzpicture}[baseline=(current bounding box),scale=1]        
    \draw (.3,0) --  (.7,0);
    \foreach \i in {1,2} {
        \node at (\i,.75) {$i_{\i}$};
        \draw (\i,.3) -- (\i,.5);
        \filldraw[fill = white] (\i-.3,-.3) rectangle node {$A$} (\i+.3,.3);
        \draw (\i+.3,0) --  (\i+.7,0);
    }
    \draw[opacity=0] (0,-.5) -- (1,-.5);
    \node at (-0.05,0) {$...$};
    \node at (3.05,0) {$...$};
\end{tikzpicture}\bigg) \ket{\{i_n\}}. 
\end{aligned} 
\end{equation}
The transfer matrix $T_A$ for this local tensor is a matrix in the doubled virtual space,
\beq
    T_A &= \ \begin{tikzpicture}[baseline=(current bounding box),scale=1]
        \draw (0,.3) -- (0,.5);
        \filldraw[fill = white] (-.3,-.3) rectangle node {$A$} (.3,.3);
        \draw (0.7,0) --  (0.3,0);
        \draw (-0.3,0) --  (-0.7,0);
        \filldraw[fill = white] (-.3,.5) rectangle node {$\bar{A}$} (.3,1.1);
        \draw (0.7,.8) --  (0.3,.8);
        \draw (-0.3,.8) --  (-0.7,.8);
    \end{tikzpicture}\ .
    \label{eq:transfer_matrix}
\eeq
The tensor $\bar{A}$ denotes the complex conjugate of $A$, with the physical bond directed downward. For any local tensor $A$, its transfer matrix $T_A$ always defines a completely positive map~\cite{cirac2021matrix}:\footnote{The matrices in the diagrams act from top to bottom and from left to right throughout this work.}
\beq
    & \begin{tikzpicture}[baseline=(current bounding box),scale=1]
        \draw (0.7,0) --  (0.3,0);
        \draw (0.3,.8) --  (0.7,.8);
        \filldraw[fill = white] (.7,-.3) rectangle node {$\rho^R$} (1.4,1.1);
    \end{tikzpicture}
    \ \rightarrow\ \begin{tikzpicture}[baseline=(current bounding box),scale=1]
        \draw (0,.3) -- (0,.5);
        \filldraw[fill = white] (-.3,-.3) rectangle node {$A$} (.3,.3);
        \draw (0.7,0) --  (0.3,0);
        \draw (-0.3,0) --  (-0.7,0);
        \filldraw[fill = white] (-.3,.5) rectangle node {$\bar{A}$} (.3,1.1);
        \draw (0.7,.8) --  (0.3,.8);
        \draw (-0.3,.8) --  (-0.7,.8);
        \filldraw[fill = white] (.7,-.3) rectangle node {$\rho^R$} (1.4,1.1);
    \end{tikzpicture}\ ,\\[2ex]
    & \begin{tikzpicture}[baseline=(current bounding box),scale=1]
        \draw (-0.3,0) --  (-0.7,0);
        \draw (-0.3,.8) --  (-0.7,.8);
        \filldraw[fill = white] (-1.4,-.3) rectangle node {$\rho^L$} (-.7,1.1);
    \end{tikzpicture}\ \rightarrow\ \begin{tikzpicture}[baseline=(current bounding box),scale=1]
        \draw (0,.3) -- (0,.5);
        \filldraw[fill = white] (-.3,-.3) rectangle node {$A$} (.3,.3);
        \draw (0.7,0) --  (0.3,0);
        \draw (-0.3,0) --  (-0.7,0);
        \filldraw[fill = white] (-.3,.5) rectangle node {$\bar{A}$} (.3,1.1);
        \draw (0.7,.8) --  (0.3,.8);
        \draw (-0.3,.8) --  (-0.7,.8);
        \filldraw[fill = white] (-1.4,-.3) rectangle node {$\rho^L$} (-.7,1.1);
    \end{tikzpicture}\ .
    \label{eq:CP_map}
\eeq

Positivity of this map ensures that for any positive $\rho$, i.e. $\bra{\psi}\rho\ket{\psi}\geq 0$ for any $\ket{\psi}$, defined in the virtual space, $T_A(\rho)$ is also positive.

\begin{definition}
    An MPS $\ket{A}$ is {\it FDLU-preparable} if it can be prepared from a product state using an FDLU circuit.
\end{definition}

In the remainder of this section, we discuss the necessary and sufficient condition for a state $\ket{A}$ to be FDLU-preparable. 

For the ``necessary'' direction, it is shown in Ref.~\cite{malz2024preparation} that $\ket{A}$ must have a vanishing correlation length. Here, we instead give a heuristic argument for states whose preparations are insensitive to the system size $N$. In other words, we expect that there exists a matrix product operator (MPO) generated by a tensor $M$ that prepares $\ket{A}$ for any $N$, and we expect this MPO to be an FDLU for any $N$. To be explicit, there is a local tensor $M$ satisfying
\begin{equation}
\begin{aligned}
\begin{tikzpicture}[baseline=(current bounding box),scale=1.2]
        \draw (0,.3) -- (0,.5);
        \filldraw[fill = white] (-.3,-.3) rectangle node {$A$} (.3,.3);
        \draw (0.7,0) --  (0.3,0);
        \draw (-0.3,0) --  (-0.7,0);
        \node at (1,0) {$\equiv$};
        \draw (2,.3) -- (2,.5);
        \filldraw[fill = white] (1.7,-.3) rectangle node {$M$} (2.3,.3);
        \draw (2.7,0) --  (2.3,0);
        \draw (1.7,0) --  (1.3,0);
        \draw (2,-.3) -- (2,-.5);
        \node at (2,-.7) {$0$};
\end{tikzpicture}\ ,
\label{eq:M_tensor}
\end{aligned}  
\end{equation}
that generates an FDLU $D_M$ for any system size $N$ as follows:
\begin{equation}
\begin{aligned}
D_{M}=\sum_{\{i_n, j_n\}}\operatorname{Tr}\bigg(
\begin{tikzpicture}[baseline=(current bounding box),scale=1]        
    \draw (.3,0) --  (.7,0);
    \foreach \i in {1,2} {
        \node at (\i,.75) {$i_{\i}$};
        \draw (\i,.3) -- (\i,.5);
        \node at (\i,-.75) {$j_{\i}$};
        \draw (\i,-.3) -- (\i,-.5);
        \filldraw[fill = white] (\i-.3,-.3) rectangle node {$M$} (\i+.3,.3);
        \draw (\i+.3,0) --  (\i+.7,0);
    }
    \node at (0,0) {$...$};
    \node at (3,0) {$...$};
\end{tikzpicture}\bigg) \ket{\{i_n\}}\bra{\{j_n\}}
\label{eq:D_M}
\end{aligned} 
\end{equation} 
It turns out that this requirement is fulfilled as long as $D_M$ is unitary for any $N$~\cite{shukla2025simple}. From the unitarity of $D_M$,\footnote{For convenience, all the equalities in this work regarding MPS tensors and transfer matrices are up to a normalization factor.}
\beq
    D_M^{\dagger} D_M =\operatorname{Tr}\bigg(
\begin{tikzpicture}[baseline=(current bounding box),scale=1]        
    \draw (.3,0) --  (.7,0);
    \foreach \i in {1,2} {
        \draw (\i,1.3) -- (\i,-.5);
        \filldraw[fill = white] (\i-.3,-.3) rectangle node {$M$} (\i+.3,.3);
        \draw (\i+.3,0) --  (\i+.7,0);
        \filldraw[fill = white] (\i-.3,.5) rectangle node {$\bar{M}$} (\i+.3,1.1);
        \draw (\i+0.7,.8) --  (\i+0.3,.8);
        \draw (\i-0.3,.8) --  (\i-0.7,.8);
    }
    \node at (0,0.4) {$...$};
    \node at (3,0.4) {$...$};
\end{tikzpicture}\bigg)= I .
\eeq 
Taking the dangling physical bonds as zero, the above equation yields $\Tr(T_A^N) = 1$. When this equation holds for any $N$, the spectrum of $T_A$ must be $(1,0,\cdots,0)$. 

In general, a matrix $T_A$ with only one non-zero eigenvalue does not necessarily mean that it is of rank 1 as $T_A$ could contain a non-trivial Jordan block. Since $T_A$ is a $D^2\times D^2$ matrix in the doubled virtual space, the non-trivial Jordan block in the zero eigenvalue subspace is at most $D^2-1$ dimensional when $T_A$ has only one non-zero eigenvalue. As a result, the $n$-th power of $T_A$ must be of rank 1 as long as $n\geq D^2$. 

Given an MPS $\ket{A}$ with a $d$-dimensional local physical Hilbert space on each site, there is a standard procedure, dubbed {\it site blocking}, where one blocks $n$ consecutive sites into a large site with a $d^n$-dimensional physical Hilbert space. After site blocking, the corresponding local MPS tensor $A'$ is given by
\begin{equation}
\begin{aligned}
\begin{tikzpicture}[baseline=(current bounding box),scale=1]
        \draw (-.8,.3) -- (-.8,.7);
        \draw (-.5,.3) -- (-.5,.7);
        \node at (-.2,.5) {...};
        \draw (.1,.3) -- (.1,.7);
        \filldraw[fill = white] (-1,-.3) rectangle node {$A'$} (.3,.3);
        \draw (.7,0) --  (.3,0);
        \draw (-1,0) --  (-1.4,0);
        \node at (1,0) {$\equiv$};
        \filldraw[fill = white] (1.7,-.3) rectangle node {$A$} (2.3,.3);
        \draw (1.7,0) --  (1.3,0);
        \draw (2,.3) -- (2,.7);
        \draw (2.7,0) --  (2.3,0);
        \filldraw[fill = white] (2.7,-.3) rectangle node {$A$} (3.3,.3);
        \draw (3,.3) -- (3,.7);
        \draw (3.5,0) --  (3.3,0);
        \node at (3.8,0) {...};
        \draw (4.1,0) --  (4.3,0);
        \filldraw[fill = white] (4.3,-.3) rectangle node {$A$} (4.9,.3);
        \draw (4.6,.3) -- (4.6,.7);
        \draw (4.9,0) -- (5.3,0);
\end{tikzpicture}
\end{aligned}  
\end{equation} 
Correspondingly, the transfer matrices before and after the blocking satisfy a simple relation $T_{A'} = T_A^n$. Therefore, after blocking $n\geq D^2$ sites, the transfer matrix must be of rank one. 

For the ``sufficient'' direction, let us start with a target state $\ket{A}$ with a rank-1 transfer matrix. The general form of a rank-1 transfer matrix $T_A$ is as follows: 
\beq
    T_A =\ \begin{tikzpicture}[baseline=(current bounding box),scale=1.5]
        \draw[rounded corners=8pt]
        (-.8,.8) -- (-.3,.8) -- (-.3,0) -- (-.8,0);
        \draw[rounded corners=8pt]
        (.8,.8) -- (.3,.8) -- (.3,0) -- (.8,0);
        \filldraw[draw, circle, fill=white, minimum size=6pt](-.3,0.4) circle (6pt);
        \node at (-0.3,0.4) {$R$};
        \filldraw[draw, circle, fill=white, minimum size=6pt](.3,0.4) circle (6pt);
        \node at (0.3,0.4) {$L$};
        \end{tikzpicture}\ .
        \label{eq:T_rank_1}
\eeq

\begin{figure*}[!t]
    \centering
    \includegraphics[width=0.85\linewidth]{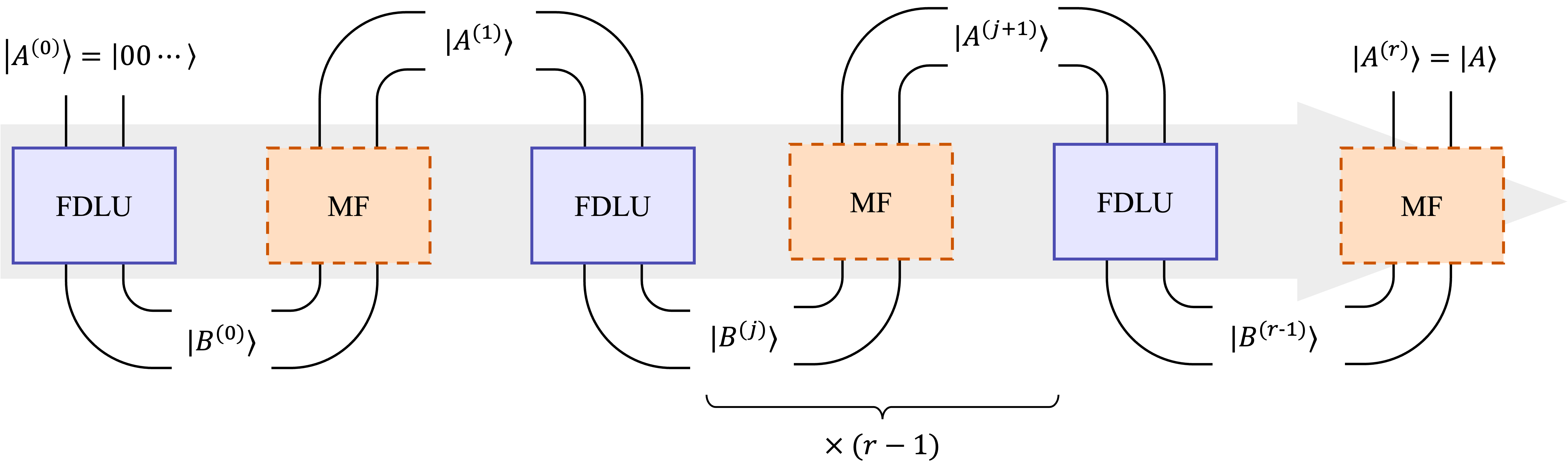}
    \caption{Sequential preparation of the target state $\ket{A}$ from $\ket{00\cdots}$ using circuits with local measurements and feedback corrections. In the sequence, $|B^{(j)}\rangle $ is prepared from $|A^{(j)}\rangle $ using FDLU, and $|A^{(j+1)}\rangle$ is prepared from $|B^{(j)}\rangle$ using a single round of local measurements and feedback corrections.}
    \label{fig:circuit}
\end{figure*}

By the definition of the transfer matrix in Eq.~\eqref{eq:transfer_matrix}, $R$ and $L$ are non-negative Hermitian matrices. This is because if, say, $L$ has a negative eigenvalue associated with eigenvector $\ket{\psi}$, then the above $T_A$ does not preserve the positivity of $\rho^R = \ket{\psi}\bra{\psi}$. This contradicts the fact that a transfer matrix must correspond to a completely positive map. We diagonalize both matrices and write this transfer matrix as
\beq
    T_A =\ \begin{tikzpicture}[baseline=(current bounding box),scale=1.5]
        \draw[rounded corners=8pt]
        (-1.1,.8) -- (-.3,.8) -- (-.3,0) -- (-1.1,0);
        \draw[rounded corners=8pt]
        (1.1,.8) -- (.3,.8) -- (.3,0) -- (1.1,0);
        \filldraw[draw, circle, fill=white, minimum size=6pt](-.7,0.8) circle (6pt);
        \node at (-0.7,0.8) {$\bar{u}^R$};
        \filldraw[draw, circle, fill=white, minimum size=6pt](-.7,0) circle (6pt);
        \node at (-0.7,0) {$u^R$};
        \filldraw[draw, circle, fill=white, minimum size=6pt](-.3,0.4) circle (6pt);
        \node at (-0.3,0.4) {$\Lambda^R$};
        \filldraw[draw, circle, fill=white, minimum size=6pt](.3,0.4) circle (6pt);
        \node at (0.3,0.4) {$\Lambda^L$};
        \filldraw[draw, circle, fill=white, minimum size=6pt](.7,0.8) circle (6pt);
        \node at (0.7,0.8) {$\bar{u}^L$};
        \filldraw[draw, circle, fill=white, minimum size=6pt](.7,0) circle (6pt);
        \node at (0.7,0) {$u^L$};
        \end{tikzpicture}\ ,
\eeq
where both $\Lambda^R$ and $\Lambda^L$ are diagonal matrices with non-negative entries. A representative MPS tensor for this transfer matrix can be given by
\beq
    A' =\ \begin{tikzpicture}[baseline=(current bounding box),scale=1.5]
        \draw[rounded corners=8pt]
        (-.3,.8) -- (-.3,0) -- (-1.1,0);
        \draw[rounded corners=8pt]
        (.3,.8) -- (.3,0) -- (1.1,0);
        \filldraw[draw, circle, fill=white, minimum size=6pt](-.3,0.4) circle (6pt);
        \node at (-0.3,0.4) {$\Gamma^{R}$};
        \filldraw[draw, circle, fill=white, minimum size=6pt](.3,0.4) circle (6pt);
        \node at (0.3,0.4) {$\Gamma^{L}$};
        \filldraw[draw, circle, fill=white, minimum size=6pt](-.7,0) circle (6pt);
        \node at (-0.7,0) {$u^R$};
        \filldraw[draw, circle, fill=white, minimum size=6pt](.7,0) circle (6pt);
        \node at (0.7,0) {$u^L$};
        \end{tikzpicture}\ ,
\eeq
where $\Gamma$'s are the square root of $\Lambda$'s. This tensor generates the following product state 
\beq
    \ket{A'} =\ \begin{tikzpicture}[baseline=(current bounding box),scale=1.5]
        \draw[rounded corners=8pt]
        (-.3,.3) -- (-.3,0) -- (-1.1,0) -- (-1.1,.3);
        \draw[rounded corners=8pt]
        (0,.3) -- (0,0) -- (.8,0) -- (.8,0.3);
        \filldraw[draw, circle, fill=white, minimum size=6pt](-.7,0) circle (6pt);
        \node at (-.7,0) {$S$};
        \filldraw[draw, circle, fill=white, minimum size=6pt](.4,0) circle (6pt);
        \node at (.4,0) {$S$};
        \node at (1.3,0) {...};
        \node at (-1.6,0) {...};
        \end{tikzpicture}
        \label{eq:FDLU_rep_state}
\eeq
where $S = \Gamma^L u^L u^R (\Gamma^R)^T$. Since the MPSs $\ket{A}$ and $\ket{A'}$ have the same transfer matrix, they differ only by a transversal isometry, i.e., a transversal unitary circuit up to a disentangled ancilla~\cite{QChannelLecture,cirac2021matrix}. Therefore, we can define an FDLU as a composition of a circuit that transforms $\ket{0}^{\otimes N}$ to the product state $\ket{A'}$, and another circuit that transforms $\ket{A'}$ to $\ket{A}$. This FDLU is translation invariant and prepares the target state $\ket{A}$ from $\ket{0}^{\otimes N}$ for any $N$.

In the remainder of this work, we will assume that the site blocking has been done, such that we can equate the FDLU-preparable condition for state $\ket{A}$ with the rank-1 condition for its transfer matrix $T_A$. We conclude the discussion with Lemma 1, the main result of this section: 

\begin{lemma}
    The necessary and sufficient condition for an MPS $\ket{A}$ to be FDLU-preparable is that $T_A$ is of rank 1, after a finite site blocking.
    \label{lemma:FDLU}
\end{lemma}

\section{State preparation with measurement and feedback}
\label{sec:MF}

In this section, we discuss the general circuit architecture for state preparation, including the use of single-site measurement and transversal feedback correction. Specifically, to prepare a target state $\ket{A}$, we start from the initial state $\ket{0}^{\otimes N}\otimes \ket{0}^{\otimes \text{ancilla}}$, then apply FDLU and MF circuits in an alternating manner. Throughout the process, a sequence of states $\{|A^{(j)}\rangle \}$ and $\{|B^{(j)}\rangle \}$ for $j=0,\cdots,r$ are prepared: $|B^{(j)}\rangle $ is prepared from $|A^{(j)}\rangle $ using FDLU, and $|A^{(j+1)}\rangle$ is prepared from $|B^{(j)}\rangle$ using a single round of MF circuit; see Fig.~\ref{fig:circuit}.

In Sec.~\ref{sec:general_strategy}, we introduce a general strategy for preparing a target state by analyzing its pushable defects, which applies to  MPSs with both periodic boundary condition (PBC) and open boundary condition (OBC). In Sec.~\ref{sec:open_bdry}, we specialize to the preparation of MPSs with OBC. Under certain mild assumptions, we show that all MF preparable states can be prepared via our defect approach, therefore providing a powerful framework towards classifying the preparability of target states.

\subsection{General approach from defects}\label{sec:general_strategy}
In this subsection, we will introduce the notion of pushable defects in the virtual space of a given state $\ket{A}$, and use it to construct an associated state $\ket{B}$ that can be related to $\ket{A}$ by a single round of MF. This construction enables a general state preparation scheme beyond one round of MF.

Let us consider the last step of the entire preparation circuit in Fig.~\ref{fig:circuit}, which is to prepare the target state $\ket{A}$ from another state $\ket{B}$ using a single round of MF. To begin, we assume that there is an ancillary bond at site $n$ with a $d_a$-dimensional Hilbert space spanned by the computational basis $\ket{a_n}$ with $a_n=0,1,2,..., d_a-1$.  We write $\ket{B}$ as the following MPS:\footnote{The MPS shown here is represented with PBC, but our general preparation strategy applies to MPSs represented with both PBC and OBC.} 
\begin{equation}
\begin{aligned}
\ket{B}=\sum_{\{i_n,a_n\}}\Tr\bigg(
\begin{tikzpicture}[baseline=(current bounding box),scale=1.2]        
    \draw (.3,0) --  (.7,0);
    \foreach \i in {1,2} {
        \node at (\i-.25,.75) {$i_{\i}$};
        \node at (\i+.25,.75) {$a_{\i}$};
        \draw (\i-.2,.3) -- (\i-.2,.5);
        \draw (\i+.2,.3) -- (\i+.2,.5);
        \filldraw[fill = white] (\i-.3,-.3) rectangle node {$B$} (\i+.3,.3);
        \draw (\i+.3,0) --  (\i+.7,0);
    }
    \draw[opacity=0] (0,-.5) -- (1,-.5);
    \node at (-0,0) {...};
    \node at (3,0) {...};
\end{tikzpicture}\bigg) \ket{\{i_n,a_n\}}.
\end{aligned} 
\end{equation}
We construct the local tensor $B$ such that it satisfies 
\begin{equation}
\begin{aligned}
\begin{tikzpicture}[baseline=(current bounding box),scale=1]
        \draw (2,.3) -- (2,.7);
        \filldraw[fill = white] (-.3,-.3) rectangle node {$B$} (.3,.3);
        \draw (0.7,0) --  (0.3,0);
        \draw (-0.3,0) --  (-0.7,0);
        \node at (1,0) {$=$};
        \draw (-.2,.3) -- (-.2,.7);
        \draw (+.2,.3) -- (+.2,.7);
        \node at (+.2,.9) {$0$};
        \filldraw[fill = white] (1.7,-.3) rectangle node {$A$} (2.3,.3);
        \draw (2.7,0) --  (2.3,0);
        \draw (1.7,0) --  (1.3,0);
        \node at (2.8,0) {,};
\end{tikzpicture}
\end{aligned}  
\end{equation} 
up to a gauge transformation. The above condition ensures that projecting the rightmost physical index $a_n$ to zero on each site exactly prepares the target state $\ket{A}$. To achieve this projection, a natural approach is to perform the single-site measurement in the computational basis and post-select on the measurement outcome $a_n=0$. However, there are also possibilities for projecting into $a_n=1,\cdots, d_a-1$; the probability of having the outcome $a_n=0$ for all $n$ will generically decay exponentially with the system size $N$. In the following, we will discuss the conditions on $\ket{A}$ and $\ket{B}$ that ensure the unwanted measurement outcomes can be compensated by implementing certain feedback operations without post-selection. First, we require that the local tensor $B$ satisfies
\begin{equation}
\begin{aligned}
\begin{tikzpicture}[baseline=(current bounding box),scale=1]
        \draw (2,.3) -- (2,.7);
        \filldraw[fill = white] (-.3,-.3) rectangle node {$B$} (.3,.3);
        \draw (0.7,0) --  (0.3,0);
        \draw (-0.3,0) --  (-0.7,0);
        \node at (1,0) {$=$};
        \draw (-.2,.3) -- (-.2,.7);
        \draw (+.2,.3) -- (+.2,.7);
        \node at (+.2,.9) {$t$};
        \filldraw[fill = white] (1.7,-.3) rectangle node {$A$} (2.3,.3);
        \draw (3.5,0) --  (2.3,0);
        \draw (1.7,0) --  (1.3,0);
        \filldraw[draw, circle, fill=white, minimum size=6pt](2.9,0) circle (8pt);
        \node at (2.9,0) {$V_t$};
        \node at (3.6,0) {,};
\end{tikzpicture}
\end{aligned}  
\label{eq:B_definition}
\end{equation}  
where $\{V_t\}$ are the defects of tensor $A$ that correspond to various measurement outcomes $t=0,1,\cdots, d_a-1$. As will be highlighted through examples, the state $\ket{B}$ depends on $\ket{A}$ and our choice of defects $\{V_t\}$. In addition, we will demand that those defects be pushable.

\begin{definition}\label{def:pushable}
An error defect $V_t$ is {\it right-pushable}, if there is a set of operators $\{V_t^{[n]}\}$ with $V_t^{[0]}=V_t$, and a set of unitary gates $\{U_t^{[n]}\}$ for any positive integer $n$, such that the defect $V_{t}^{[n-1]}$ on the left virtual bond of tensor $A$ can be pushed to the right virtual bond by applying $U_t^{[n]}$ on the physical bond:
\begin{equation}
\begin{aligned}
\begin{tikzpicture}[baseline=(current bounding box),scale=1]
        \draw (2,.3) -- (2,.7);
        \filldraw[fill = white] (-.3,-.3) rectangle node {$A$} (.3,.3);
        \draw (0.7,0) --  (0.3,0);
        \draw (-0.3,0) --  (-2,0);
        \filldraw[fill = white] (-1.6,-.3) rectangle node {$V_t^{[n-1]}$} (-.5,.3);
        \draw (0,.3) -- (0,1.3);
        \filldraw[fill = white] (-.3,.5) rectangle node {$U_t^{[n]}$} (.3,1.1);
        \node at (1,0) {$=$};
        \draw (3.5,0) --  (2.3,0);
        \draw (4,0) --  (1.3,0);
        \filldraw[fill = white] (1.7,-.3) rectangle node {$A$} (2.3,.3);
        \filldraw[draw, circle, fill=white, minimum size=6pt](2.9,0) circle (8pt);
        \filldraw[fill = white] (2.5,-.3) rectangle node {$V_t^{[n]}$} (3.6,.3);
        \node at (4.1,0) {.};
\end{tikzpicture}
\end{aligned}  
\label{eq:pushing_U_right}
\end{equation}  
Similarly, an error defect is {\it left-pushable}, if there is a set $\{V_t^{[n]}\}$ and a set of unitary gates $\{U_t^{[n]}\}$ for any positive integer $n$ such that
\begin{equation}
\begin{aligned}
\begin{tikzpicture}[baseline=(current bounding box),scale=1]
        \draw (0,.3) -- (0,1.3);
        \filldraw[fill = white] (-.3,-.3) rectangle node {$A$} (.3,.3);
        \draw (-.7,0) --  (-0.3,0);
        \draw (0.3,0) --  (2,0);
        \filldraw[fill = white] (-.3,.5) rectangle node {$U_t^{[n]}$} (.3,1.1);
        \filldraw[fill = white] (.5,-.3) rectangle node {$V_t^{[n-1]}$} (1.6,.3);
        \node at (2.3,0) {$=$};
        \draw (5.3,0) --  (2.6,0);
        \draw (4.6,.3) -- (4.6,.7);
        \filldraw[fill = white] (4.3,-.3) rectangle node {$A$} (4.9,.3);
        \filldraw[fill = white] (3,-.3) rectangle node {$V_t^{[n]}$} (4.1,.3);
        \node at (5.4,0) {.};     
\end{tikzpicture}
\end{aligned}  
\label{eq:pushing_U_left}
\end{equation}  

The above two conditions guarantee that the defect on a given virtual bond can be pushed to another virtual bond by unitary operations acting on physical legs. In particular, with periodic boundary conditions, all defects can be pushed to a single virtual bond. Then by post-selecting the defect-less outcome on that single bond, we can prepare $\ket{A}$ from $\ket{B}$ with a certain probability of success. We will come back to this perspective at the end of this subsection in the last remark. With open boundary conditions, if all defects can be pushed to the boundary and then corrected with a unitary operation, then $\ket{A}$ can be deterministically prepared from $\ket{B}$ with a single round MF circuit. State preparation with open boundary conditions will be discussed in detail in the next subsection (Sec.~\ref{sec:open_bdry}).

The above pushability conditions for a local tensor $A$ directly impose certain constraints on the transfer matrix. In particular, it is straightforward to see that the right-pushable condition in Eq.~\eqref{eq:pushing_U_right} implies the following equality of the transfer matrix:
\beq
\Big(\overline{V}_t^{[n-1]}\otimes V_t^{[n-1]}\Big) T_A = T_A \Big(\overline{V}_t^{[n]}\otimes V_t^{[n]}\Big),\\[2ex]\\
    \begin{tikzpicture}[baseline=(current bounding box),scale=1]
        \draw (0,.3) -- (0,.5);
        \filldraw[fill = white] (-.3,-.3) rectangle node {$A$} (.3,.3);
        \draw (-2,0) --  (-0.3,0);
        \draw (0.3,0) --  (0.7,0);
        \filldraw[fill = white] (-.3,.5) rectangle node {$\bar{A}$} (.3,1.1);
        \draw (-2,.8) --  (-0.3,.8);
        \draw (0.3,.8) --  (0.7,.8);
        \filldraw[fill = white] (-1.6,-.3) rectangle node {$V_t^{[n-1]}$} (-.5,.3);
        \filldraw[fill = white] (-1.6,.8-.3) rectangle node {$\bar{V}_t^{[n-1]}$} (-.5,.8+.3);
    \end{tikzpicture} 
    \ = \ \begin{tikzpicture}[baseline=(current bounding box),scale=1]
        \draw (0,.3) -- (0,.5);
        \filldraw[fill = white] (-.3,-.3) rectangle node {$A$} (.3,.3);
        \draw (2,0) --  (0.3,0);
        \draw (-0.3,0) --  (-0.7,0);
        \filldraw[fill = white] (-.3,.5) rectangle node {$\bar{A}$} (.3,1.1);
        \draw (2,.8) --  (0.3,.8);
        \draw (-0.3,.8) --  (-0.7,.8);
        \filldraw[fill = white] (.5,-.3) rectangle node {$V_t^{[n]}$} (1.6,.3);
        \filldraw[fill = white] (.5,.8-.3) rectangle node {$\bar{V}_t^{[n]}$} (1.6,.8+.3);
    \end{tikzpicture}\ ,
\label{eq:pushable_defect_right}
\eeq

Similarly, the left-pushable condition (Eq.~\eqref{eq:pushing_U_left}) gives the following identity:
\beq
    \Big(\overline{V}^{[n]}_t\otimes V^{[n]}_t\Big) T_A = T_A \Big(\overline{V}^{[n-1]}_t\otimes V^{[n-1]}_t\Big),\\[2ex]
    \begin{tikzpicture}[baseline=(current bounding box),scale=1]
        \draw (0,.3) -- (0,.5);
        \filldraw[fill = white] (-.3,-.3) rectangle node {$A$} (.3,.3);
        \draw (-2,0) --  (-0.3,0);
        \draw (0.3,0) --  (0.7,0);
        \filldraw[fill = white] (-.3,.5) rectangle node {$\bar{A}$} (.3,1.1);
        \draw (-2,.8) --  (-0.3,.8);
        \draw (0.3,.8) --  (0.7,.8);
        \filldraw[fill = white] (-1.6,-.3) rectangle node {$V_t^{[n]}$} (-.5,.3);
        \filldraw[fill = white] (-1.6,.8-.3) rectangle node {$\bar{V}_t^{[n]}$} (-.5,.8+.3);
    \end{tikzpicture}
    \ = \ \begin{tikzpicture}[baseline=(current bounding box),scale=1]
        \draw (0,.3) -- (0,.5);
        \filldraw[fill = white] (-.3,-.3) rectangle node {$A$} (.3,.3);
        \draw (2,0) --  (0.3,0);
        \draw (-0.3,0) --  (-0.7,0);
        \filldraw[fill = white] (-.3,.5) rectangle node {$\bar{A}$} (.3,1.1);
        \draw (2,.8) --  (0.3,.8);
        \draw (-0.3,.8) --  (-0.7,.8);
        \filldraw[fill = white] (.5,-.3) rectangle node {$V_t^{[n-1]}$} (1.6,.3);
        \filldraw[fill = white] (.5,.8-.3) rectangle node {$\bar{V}_t^{[n-1]}$} (1.6,.8+.3);
    \end{tikzpicture}\ .
    \label{eq:pushable_defect_left}
\eeq

\end{definition}

In the reverse direction, the equations in terms of transfer matrices also imply pushability of the MPS tensor. To see this, we notice that  Eq.~\eqref{eq:pushable_defect_right} suggests that the following tensors
\beq
    \begin{tikzpicture}[baseline=(current bounding box),scale=1]
        \draw (0,.3) -- (0,.7);
        \filldraw[fill = white] (-.3,-.3) rectangle node {$A$} (.3,.3);
        \draw (-2,0) --  (-0.3,0);
        \draw (0.3,0) --  (0.7,0);
        \filldraw[fill = white] (-1.6,-.3) rectangle node {$V_t^{[n-1]}$} (-.5,.3);
        \node at (.8,0) {,};
    \end{tikzpicture}\quad 
    \begin{tikzpicture}[baseline=(current bounding box),scale=1]
        \draw (0,.3) -- (0,.7);
        \filldraw[fill = white] (-.3,-.3) rectangle node {$A$} (.3,.3);
        \draw (2,0) --  (0.3,0);
        \draw (-0.3,0) --  (-0.7,0);
        \filldraw[fill = white] (.5,-.3) rectangle node {$V_t^{[n]}$} (1.6,.3);
        \node at (2.1,0) {,};
    \end{tikzpicture}\ 
\eeq
have the same transfer matrix. The two tensors also have the same physical bond dimension by definition. Therefore, there exists a unitary operator $U^{[n]}$ that relates the two tensors as shown in Eq.~\eqref{eq:pushing_U_right}~\cite{QChannelLecture,cirac2021matrix}.\footnote{See also Appendix B of Ref.~\cite{sahay2025classifying} and Theorem 2 of Ref.~\cite{smith2024constant} for similar discussions.} As such, the pushability defined in Eqs.~\eqref{eq:pushing_U_right}, \eqref{eq:pushing_U_left} in terms of the MPS tensor is exactly equivalent to Eqs.~\eqref{eq:pushable_defect_right}, \eqref{eq:pushable_defect_left} in terms of transfer matrices. The latter provides a powerful tool to classify MPSs preparable with MF circuits, which will be employed throughout the rest of the work.

Another important ingredient is the relation between the transfer matrices of $A$ and $B$. Specifically, Eq.~\eqref{eq:B_definition} yields

\beq
    T_B &= \ \begin{tikzpicture}[baseline=(current bounding box),scale=1]
        \draw (0-.2,.3) -- (0-.2,.5);
        \draw (0+.2,.3) -- (0+.2,.5);
        \filldraw[fill = white] (-.3,-.3) rectangle node {$B$} (.3,.3);
        \draw (0.7,0) --  (0.3,0);
        \draw (-0.3,0) --  (-0.7,0);
        \filldraw[fill = white] (-.3,.5) rectangle node {$\bar{B}$} (.3,1.1);
        \draw (0.7,.8) --  (0.3,.8);
        \draw (-0.3,.8) --  (-0.7,.8);
    \end{tikzpicture}
    = \sum_{t=0}^{d_a-1} \
    \begin{tikzpicture}[baseline=(current bounding box),scale=1]
        \draw (0,.3) -- (0,.5);
        \filldraw[fill = white] (-.3,-.3) rectangle node {$A$} (.3,.3);
        \draw (1.3,0) --  (0.3,0);
        \draw (-0.3,0) --  (-0.7,0);
        \filldraw[fill = white] (-.3,.5) rectangle node {$\bar{A}$} (.3,1.1);
        \draw (1.3,.8) --  (0.3,.8);
        \draw (-0.3,.8) --  (-0.7,.8);
        \filldraw[draw, circle, fill=white, minimum size=6pt](.8,0) circle (8pt);
        \node at (.8,0) {$V_t$};
        \filldraw[draw, circle, fill=white, minimum size=6pt](.8,0.8) circle (8pt);
        \node at (.8,0.8) {$\bar{V}_t$};
    \end{tikzpicture}\\[2ex]
    &= T_A \Big(\sum_t \bar{V}_t\otimes V_t\Big).
    \label{eq:T_B}
\eeq 

According to Lemma \ref{lemma:FDLU}, when $\ket{B}$ is FDLU-preparable, its transfer matrix $T_B$ must be of rank 1 after site blocking. Eq.~\eqref{eq:T_B} then implies that $T_A \Big(\sum_t \bar{V}_t\otimes V_t\Big)$ must be of rank-1, if $\ket{A}$ can be prepared from an FDLU-preparable state $\ket{B}$ using a 1-round MF circuit. Furthermore, we can define $r$-round MF-preparability as follows:

\begin{definition}
    An MPS $\ket{A}$ is {\it $r$-round MF-preparable} if it can be prepared from a product state using alternations of FDLU and MF circuits, with $r$ rounds of MF in total.
\end{definition}

Then, based on the previous discussions, we can introduce a hierarchy of MF circuits for state preparation: 

\begin{lemma}
    An MPS $\ket{A}$ with a set of pushable defects  $\{V_t\}$ is 1-round MF-preparable, if the associated MPS $\ket{B}$ with transfer matrix $T_{B}\propto T_A\cdot\big(\sum_t \overline{V}_t\otimes V_t\big)$ is FDLU-preparable. Similarly, an MPS $\ket{A}$ is $r$-round MF-preparable, if the associated MPS $\ket{B}$ is $(r-1)$-round MF-preparable.
    \label{lemma:A_B}
\end{lemma}

\begin{figure}[t]
    \centering
    \includegraphics[width=0.9\linewidth]{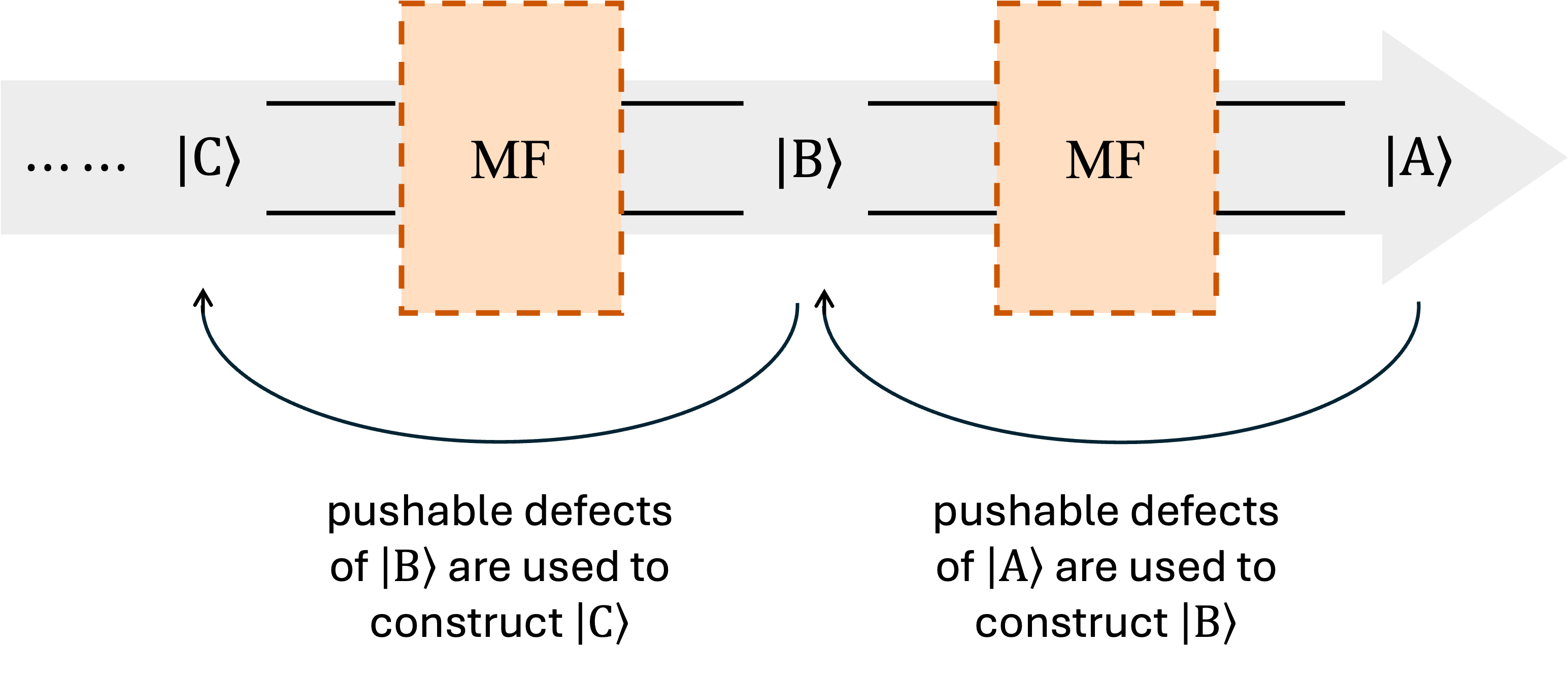}
    \caption{The preparation scheme we propose for the target state $\ket{A}$, in which we perform a sequence of MF circuits on a certain initial FDLU-preparable state, generating a sequence of states $\{\cdots,\ket{C},\ket{B},\ket{A}\}$. The construction of the state sequence from the target state $\ket{A}$ is in reverse order: the pushable defects of $\ket{A}$ are used to construct $\ket{B}$; the pushable defects of $\ket{B}$ are used to construct $\ket{C}$, etc. As discussed in the main text, certain rounds of MF circuits can be replaced by an FDLU circuit, which suggests that this scheme includes the general circuit architecture shown in Fig.~\ref{fig:circuit}.}
    \label{fig:circuit_2}
\end{figure}

As we have seen, from the MPS $\ket{A}$ and its pushable defects, we can construct an associated MPS $\ket{B}$ that is a single round MF away. Meanwhile, $\ket{B}$ allows for another set of pushable defects, enabling us to construct another MPS $\ket{C}$ that is a single round MF away. Repeating this process leads to a state preparation scheme for a target state $\ket{A}$ that uses multiple rounds of MF as sketched in Fig.~\ref{fig:circuit_2}. 

In the rest of the work, we will elaborate on this preparation scheme for various pushable defects and pushing relations. Before proceeding to a more detailed discussion specialized to OBC, we make a few remarks here:

\begin{enumerate}

\item In the scheme presented above, we explicitly use the set of pushable defects associated with a target MPS to construct a preparation circuit. In general, a state may admit many different MPS representations, which may not have the same set of pushable defects and the same class of pushing relations. For example, as we will see later in Sec.~\ref{subsec:RI-RIRX} and Sec.~\ref{subsec:RZ-FDLU}, the MPS tensors in class $RI$--$RIRX$ and class $RZ$--$RILI$ can always be rewritten to be in class $RIRI$--$RIRX$. However, according to the fundamental theorem, different MPS tensors are always related by gauge transformation up to a site blocking ~\cite{cirac2021matrix}. Correspondingly, the set of pushable defects along with their pushing relations differ by a similarity transformation. This implies that our preparation scheme for a target state does not depend on the specific MPS representation.

\item While the preparation scheme in Fig.~\ref{fig:circuit_2} only involves the use of MF circuits, which is seemingly distinct from the general circuit architecture shown in Fig.~\ref{fig:circuit} that involves alternating applications of MF and FDLU circuits, these two preparation schemes are actually equivalent. 

As an illustrative example, suppose $\ket{A}$ is prepared from $\ket{A'}$ using a brick wall unitary circuit of depth-2, we can write the MPS tensor $A$ as follows:

\beq
    \begin{tikzpicture}[baseline=(current bounding box),scale=1.5]
        \draw[rounded corners=8pt]
        (-.7,2.2) -- (.6,2.2) -- (.6,3);
        \draw[rounded corners=8pt]
        (2.4,1.4) -- (2.4,2.2) -- (1,2.2) -- (1,3);
        \draw[rounded corners=8pt]
        (2.8,1.4) -- (2.8,2.2) -- (3.2,2.2);
        \filldraw[fill = white] (2.2,1.6) rectangle node {$U$} (3,2);
        \filldraw[fill = white] (.4,2.4) rectangle node {$V$} (1.2,2.8);
        \draw (-.7,1.2) -- (3.2,1.2);
        \filldraw[fill = white] (2.2,1) rectangle node {$A'$} (3,1.4);
    \end{tikzpicture}
    \label{eq:fDLU_A_prime}
\eeq
Since $V$ is a unitary gate, any defects on the top left virtual bond are pushable. To see this, note that $W$ obeys $VW=(V W V^{\dagger})V$, thus the physical unitary $U_\text{phys} = V W^{\dagger} V^{\dagger}$ pushes the defect $W$ on the top left virtual bond to identity on the right virtual bond. Based on the results discussed in this section, we can construct an associated state $\ket{B}$ from which $\ket{A}$ is preparable using an MF circuit. According to Eq.~\eqref{eq:T_B} and Appendix \ref{app:RIRI-yyyy}, the associated state $\ket{B}$ differs from $\ket{A'}$ by a transversal circuit and the stacking of a disentangled product state. Therefore, the FDLU that prepares $\ket{A}$ from $\ket{A'}$ can be implemented by an MF circuit as well. 

Conversely, since all the defects are locally annihilated in the MF circuit, we can construct another brick-wall unitary circuit similar to Eq.~\eqref{eq:fDLU_A_prime} that relates the associated state $\ket{B}$ to the target state $\ket{A}$ (see Sec.~\ref{subsec:RIRI} for more details). 

Moreover, the reasoning above applies to general FDLU circuits, because any FDLU circuit can be reduced to a depth-2 brick wall architecture via site blocking~\cite{gross2012index}. Therefore, given a preparation circuit shown in Fig.~\ref{fig:circuit}, all FDLU circuits within can be replaced by certain MF circuits. On the other hand, given the preparation scheme demonstrated in Fig.~\ref{fig:circuit_2}, when all defects are pushed to the identity defect on the other side, the corresponding MF circuit can be replaced by an FDLU circuit, recovering the general circuit architecture shown in Fig.~\ref{fig:circuit}.

\item We would like to draw an analogy between the coarse-graining of an MPS and our $A, B$ tensors here. At the fixed point of coarse-graining, the transfer matrix formally becomes $T^{\infty}$, which is related to the original $T$ by a projection onto the dominant eigenvector subspace~\cite{zeng2015quantum}. Meanwhile, the transfer matrices $T_A$ and $T_B$ are related by an operator $\sum \bar{V}\otimes V$ (Eq.~\eqref{eq:T_B}), which is a projection onto a subspace whenever $\{V_t\}$ contains orthonormal defects.

For example, let us consider the deformed Greenberger–Horne–Zeilinger (GHZ) state
\beq
\ket{A_{\beta}}\propto e^{\sum_j\beta X_j}(\ket{0}^{\otimes N}+\ket{1}^{\otimes N}).
\eeq
Its MPS tensor can be given by
\begin{equation}
\begin{aligned}
\begin{tikzpicture}[baseline=(current bounding box),scale=1.2]
        \draw (-0.2,0) --  (1.4,0);
        \filldraw[fill = white] (0.3,-.3) rectangle node {$A_\beta$} (0.9,.3);
        \draw (.6,.3) -- (.6,.9);
        \node at (1.7,0) {$\equiv$};
        \draw (3.2,0) --  (2,0);
        \fill (2.6,0) circle (1pt);
        \draw (2.6,0.9) --  (2.6,0);
        \filldraw[draw, circle, fill=white, minimum size=6pt](2.6,.45) circle (8pt);
        \node at (2.6,.45) {$e^{\beta X}$};
        \node at (3.4,0) {.};
\end{tikzpicture}
\end{aligned}  
\end{equation} 
We use a black dot on the three-leg junction to denote the $\delta$-tensor, which gives one when all three indices are identical and gives zero otherwise. We will highlight two different associated states corresponding to two choices of different sets of pushable defects. When we take the set of pushable defects to be $\{I, Z\}$, the associated state $\ket{B}$ is the GHZ state $\ket{A_{0}}$, which is the fixed-point state for any finite value of $\beta$. Therefore, using one round of MF, we can ``reverse'' the coarse-graining to map the fixed-point GHZ state $\ket{B}$ to $\ket{A_\beta}$. 

On the other hand, we can take a larger set of pushable defects for this MPS, i.e., $\{I,Z,X,ZX\}$. The associated state $\ket{B'}$ for this larger set has a rank-1 transfer matrix, which means $\ket{B'}$ is FDLU-preparable. This is an example where, unlike the previous one, the states $\ket{A_\beta}$ and $\ket{B'}$ are not in the same phase. Indeed, when $\{V_t\}$ contains all Pauli defects, the operator $\sum \bar{V}\otimes V$ is always a rank-1 projector, which always makes the associated state FDLU-preparable. As we are about to see in Sec.~\ref{sec:all_pauli}, this enables the 1-round MF-preparation of a large class of MPSs.

\item For an MPS defined with periodic boundary conditions, the feedback correction pushes all the defects coming from unwanted measurement outcomes to a single virtual bond, and we post-select the defect-less outcome on that single bond in order to prepare the target state. The success probability must be independent of the system size $N$ to ensure that only a finite number of post-selections are required as $N\rightarrow\infty$. However, the $N$-independent probability is not guaranteed for any set of pushable defects and pushing relations. For example, for the MPS tensor of the GHZ state
\begin{equation}
\begin{aligned}
\begin{tikzpicture}[baseline=(current bounding box),scale=1.2]
        \draw (-0.2,0) --  (1.4,0);
        \filldraw[fill = white] (0.3,-.3) rectangle node {$A_0$} (0.9,.3);
        \draw (.6,.3) -- (.6,.9);
        \node at (1.7,0) {$\equiv$};
        \draw (3.2,0) --  (2,0);
        \fill (2.6,0) circle (1pt);
        \draw (2.6,0.9) --  (2.6,0);
        \node at (3.4,0) {,};
\end{tikzpicture}
\end{aligned}  
\end{equation}
we can choose the set of right-pushable defects to be $\{I,\sigma_+,\sigma_-,\sigma_z\}$, where $\sigma_a\mapsto \sigma_a$ for $a=+,-,z$. Since the $\sigma_\pm$ defects are not invertible, when pushing defects in this set on all virtual bonds to a single one, we cannot obtain the identity as long as $\sigma_\pm$ appears on any single bond. As a result, when preparing the GHZ state using the MF circuit associated to this set of defects, the success probability decays exponentially with system size $N$, hence the number of post-selections needed grows exponentially with $N$. Nevertheless, when the pushable defects are Pauli and are always pushed to Pauli defects, the success probability is always a constant as $N\rightarrow\infty$; see Appendix \ref{app:success}. For the states with open boundary condition that will be discussed in detail below, as long as the pushable defects are unitary, we can use a feedback correction to push them to the boundaries, leading to a deterministic preparation. 

\end{enumerate}

\subsection{Necessary and sufficient conditions for MF preparability with open boundary condition}\label{sec:open_bdry}

In the previous subsection, we introduced a general state preparation scheme based on the pushable defects of a target MPS. Here, by restricting the MPS to have open boundary conditions (OBC), and under certain assumptions, we will establish the presence of pushable defects as a necessary and sufficient condition for the MF preparation with OBC. In other words, for any MF-preparable state we can always construct the preparation circuit using the approach discussed in the previous subsection and summarized in Fig.~\ref{fig:circuit_2}.

We begin the discussion by defining an MPS with OBC: 
\begin{equation}
\begin{aligned}
\ket{A_{\text{open}}}=\sum_{\{i_n\},a,b}
\begin{tikzpicture}[baseline=(current bounding box),scale=1]      
    \draw (.3,0) --  (.7,0);
    \draw (.3,0) --  (.3,0.5);
    \node at (.3,.75) {$a$};
        \node at (1,.75) {$i_1$};
        \draw (1,.3) -- (1,.5);
        \filldraw[fill = white] (1-.3,-.3) rectangle node {$A$} (1+.3,.3);
        \draw (1+.3,0) --  (1+.7,0);
        \node at (2.,0) {$...$};
        \node at (2.,0.75) {$...$};
        \draw (2+.3,0) --  (2+.7,0);
        \node at (3,.75) {$i_N$};
        \draw (3,.3) -- (3,.5);
        \filldraw[fill = white] (3-.3,-.3) rectangle node {$A$} (3+.3,.3);
        \draw (3+.3,0) --  (3+.7,0);
        \draw (3.7,0) --  (3.7,0.5);
    \node at (3.7,.75) {$b$};
    \filldraw[draw, circle, fill=white, minimum size=6pt](.3,0) circle (6pt);
        \node at (.3,0) {$L$};
        \filldraw[draw, circle, fill=white, minimum size=6pt](3.7,0) circle (6pt);
        \node at (3.7,0) {$R$};
    \draw[opacity=0] (0,-.5) -- (1,-.5);
\end{tikzpicture} \ket{a,\{i_n\},b},
\end{aligned} 
\end{equation}
where we assume that matrices $L$ and $R$ are invertible in order to establish the main result of this subsection. 

Recall that our definition of a pushable defect $V$ relies on the existence of a set of operators $\{V^{[n]}\}$ and a set of unitary gates $\{U^{[n]}\}$. Now with OBC, we further assume that for a left/right-pushable defect $V$, the operators in the set $\{V^{[n]}\}$ conjugated by $L/R$ are unitary. Then we can push a defect to the boundary and apply a unitary $V^{[n]\dagger}$ to the boundary degree of freedom to completely remove it. As a result, on a system of size $N$ with OBC, we can {\it deterministically} prepare a target state $\ket{A_{\text{open}}}$ from the associated state $\ket{B_{\text{open}}}$ using one round of MF. 

Furthermore, we can reconstruct the set of pushable defects whenever the target state $\ket{A_{\text{open}}}$ is preparable from another state $\ket{B_{\text{open}}}$ for any system size $N$ using one round of MF, with an assumption on the MPS.\footnote{We thank Rahul Sahay and Ruben Verresen for the discussion on this point.} To start with, we note that given any MPS representation of a state, we can always bring the MPS tensor to a minimal form, where the following maps from vertical physical bonds to the virtual bond are surjective:\footnote{See also the discussions in Appendix B of Ref.~\cite{sahay2025classifying}.}
\begin{equation}
\begin{aligned}
\Gamma_L=\
\begin{tikzpicture}[baseline=(current bounding box),scale=1]      
    \draw (.3,0) --  (.3,0.7);
        \foreach \i in {1} {
        \draw (\i,.3) -- (\i,.7);
        \filldraw[fill = white] (\i-.3,-.3) rectangle node {$A$} (\i+.3,.3);
        \draw (\i+.3,0) --  (\i+.7,0);
        \draw (\i-.3,0) --  (\i-.7,0);
    }   
    \filldraw[draw, circle, fill=white, minimum size=6pt](.3,0) circle (6pt);
    \node at (.3,0) {$L$};
    \node at (1.9,0) {,};
\end{tikzpicture}\quad 
\Gamma_R=\
\begin{tikzpicture}[baseline=(current bounding box),scale=1]      
    \draw (1.7,0) --  (1.7,0.7);
        \foreach \i in {1} {
        \draw (\i,.3) -- (\i,.7);
        \filldraw[fill = white] (\i-.3,-.3) rectangle node {$A$} (\i+.3,.3);
        \draw (\i+.3,0) --  (\i+.7,0);
        \draw (\i-.3,0) --  (\i-.7,0);
    }   
    \filldraw[draw, circle, fill=white, minimum size=6pt](1.7,0) circle (6pt);
    \node at (1.7,0) {$R$};
    \node at (2.1,0) {.};
\end{tikzpicture}
\end{aligned} 
\end{equation}

We further assume that the transversal feedback correction is left-conditioned, i.e., the unitary correction gate on site $n$ depends only on the measurement outcomes at sites $i<n$. We take the perfect measurement outcome as $t=0$. When $N=1$, it corresponds to the following identity,
\begin{equation}
\begin{aligned}
\begin{tikzpicture}[baseline=(current bounding box),scale=1]
        \draw (-0.7,0) -- (-0.7,1.15);
        \filldraw[fill = white] (-.3,-.3) rectangle node {$B$} (.3,.3);
        \draw (0.7,0) --  (0.3,0);
        \draw (-0.3,0) --  (-0.7,0);
        \draw (0.7,0) -- (0.7,1.15);
        \draw (-.2,.3) -- (-.2,1.15);
        \draw (+.2,.3) -- (+.2,.7);
        \node at (+.2,.9) {$0$};
        \node at (1.2,0) {$=$};
        \filldraw[fill = white] (2.1,-.3) rectangle node {$A$} (2.7,.3);
        \draw (1.7,0) -- (1.7,1.15);
        \draw (2.4,.3) -- (2.4,1.15);
        \draw (3.1,0) --  (2.7,0);
        \draw (2.1,0) --  (1.7,0);
        \draw (3.1,0) -- (3.1,1.15);
        \filldraw[draw, circle, fill=white, minimum size=6pt](-.7,0) circle (6pt);
        \node at (-.7,0) {$L$};
        \filldraw[draw, circle, fill=white, minimum size=6pt](.7,0) circle (6pt);
        \node at (.7,0) {$R$};
        \filldraw[draw, circle, fill=white, minimum size=6pt](1.7,0) circle (6pt);
        \node at (1.7,0) {$L$};
        \filldraw[draw, circle, fill=white, minimum size=6pt](3.1,0) circle (6pt);
        \node at (3.1,0) {$R$};
        \node at (3.5,0) {.};
\end{tikzpicture}
\end{aligned}  
\end{equation} 
For other measurement outcomes, using the left-conditioned transversal feedback correction, we have
\begin{equation}
\begin{aligned}
\begin{tikzpicture}[baseline=(current bounding box),scale=1]
        \draw (-0.7,0) -- (-0.7,1.15);
        \filldraw[fill = white] (-.3,-.3) rectangle node {$B$} (.3,.3);
        \draw (0.7,0) --  (0.3,0);
        \draw (-0.3,0) --  (-0.7,0);
        \draw (0.7,0) -- (0.7,1.15);
        \node at (1.2,0) {$=$};
        \draw (-.2,.3) -- (-.2,1.15);
        \draw (+.2,.3) -- (+.2,.7);
        \node at (+.2,.9) {$t$};
        
        \draw (2.4,.3) -- (2.4,1.15);
        \filldraw[fill = white] (2.1,-.3) rectangle node {$A$} (2.7,.3);
        \draw (1.7,0) -- (1.7,1.15);
        \draw (3.1,0) --  (2.7,0);
        \draw (2.1,0) --  (1.7,0);
        \draw (3.1,0) -- (3.1,1.15);
        \filldraw[draw, circle, fill=white, minimum size=6pt](3.1,.6) circle (8pt);
        \node at (3.1,.6) {$W^\dagger_t$};

        \filldraw[draw, circle, fill=white, minimum size=6pt](-.7,0) circle (6pt);
        \node at (-.7,0) {$L$};
        \filldraw[draw, circle, fill=white, minimum size=6pt](.7,0) circle (6pt);
        \node at (.7,0) {$R$};
        \filldraw[draw, circle, fill=white, minimum size=6pt](1.7,0) circle (6pt);
        \node at (1.7,0) {$L$};
        \filldraw[draw, circle, fill=white, minimum size=6pt](3.1,0) circle (6pt);
        \node at (3.1,0) {$R$};
        \node at (3.5,0) {.};
\end{tikzpicture}
\end{aligned}  
\label{eq:pushable_defect_OBC}
\end{equation} 
Applying the inverse of $L$ and $R$ leads to Eq.~\eqref{eq:B_definition}:
\begin{equation}
\begin{aligned}
\begin{tikzpicture}[baseline=(current bounding box),scale=1]
        \draw (2,.3) -- (2,.7);
        \filldraw[fill = white] (-.3,-.3) rectangle node {$B$} (.3,.3);
        \draw (0.7,0) --  (0.3,0);
        \draw (-0.3,0) --  (-0.7,0);
        \node at (1,0) {$=$};
        \draw (-.2,.3) -- (-.2,.7);
        \draw (+.2,.3) -- (+.2,.7);
        \node at (+.2,.9) {$t$};
        \filldraw[fill = white] (1.7,-.3) rectangle node {$A$} (2.3,.3);
        \draw (3.5,0) --  (2.3,0);
        \draw (1.7,0) --  (1.3,0);
        \filldraw[draw, circle, fill=white, minimum size=6pt](2.9,0) circle (8pt);
        \node at (2.9,0) {$V_t$};
        \node at (3.6,0) {,};
\end{tikzpicture}
\end{aligned}  
\end{equation} 
with $V_t=R \bar{W}_t R^{-1}$.

When $N=2$, we assume that the measurement outcome at the first site is $t$ and that at the second site is perfect. The left-conditioned feedback correction gives rise to the following identity: 
\begin{equation}
\begin{aligned}
\begin{tikzpicture}[baseline=(current bounding box),scale=1]
        \draw (-0.7,0) -- (-0.7,1.15);
        \filldraw[fill = white] (-.3,-.3) rectangle node {$B$} (.3,.3);
        \draw (0.7,0) --  (0.3,0);
        \draw (-0.3,0) --  (-0.7,0);
        \draw (-.2,.3) -- (-.2,1.15);
        \draw (+.2,.3) -- (+.2,.7);
        \node at (+.2,.9) {$t$};
        \filldraw[fill = white] (.7,-.3) rectangle node {$B$} (1.3,.3);
        \draw (1-.2,.3) -- (1-.2,1.15);
        \draw (1+.2,.3) -- (1+.2,.7);
        \node at (1+.2,.9) {$0$};
        \draw (1.7,0) --  (1.3,0);
        \draw (1.7,0) -- (1.7,1.15);
        \node at (2.2,0) {$=$};
        
        \draw (3.4,.3) -- (3.4,1.15);
        \draw (4.4,.3) -- (4.4,1.15);
        \filldraw[fill = white] (3.1,-.3) rectangle node {$A$} (3.7,.3);
        \filldraw[fill = white] (4.1,-.3) rectangle node {$A$} (4.7,.3);
        \draw (2.7,0) -- (2.7,1.15);
        \draw (4.1,0) --  (3.7,0);
        \draw (3.1,0) --  (2.7,0);
        \draw (5.1,0) --  (4.7,0);
        \draw (5.1,0) -- (5.1,1.15);
        \filldraw[draw, circle, fill=white, minimum size=6pt](5.1,.6) circle (8pt);
        \node at (5.1,.6) {$W'^\dagger_t$};
        \filldraw[draw, circle, fill=white, minimum size=6pt](4.4,.7) circle (8pt);
        \node at (4.4,.7) {$U^\dagger_t$};

        \filldraw[draw, circle, fill=white, minimum size=6pt](-.7,0) circle (6pt);
        \node at (-.7,0) {$L$};
        \filldraw[draw, circle, fill=white, minimum size=6pt](1.7,0) circle (6pt);
        \node at (1.7,0) {$R$};
        \filldraw[draw, circle, fill=white, minimum size=6pt](2.7,0) circle (6pt);
        \node at (2.7,0) {$L$};
        \filldraw[draw, circle, fill=white, minimum size=6pt](5.1,0) circle (6pt);
        \node at (5.1,0) {$R$};
        \node at (5.5,0) {.};
\end{tikzpicture}
\end{aligned}  
\end{equation} 
Combining the above two equations along with the minimality condition, we have 
\begin{equation}
\begin{aligned}
\begin{tikzpicture}[baseline=(current bounding box),scale=1]
        \draw (2,.3) -- (2,.7);
        \filldraw[fill = white] (-.3,-.3) rectangle node {$A$} (.3,.3);
        \draw (0.7,0) --  (0.3,0);
        \draw (-0.3,0) --  (-1.5,0);
        \filldraw[draw, circle, fill=white, minimum size=6pt](-.8,0) circle (8pt);
        \node at (-.8,0) {$V_t$};
        \draw (0,.3) -- (0,1.3);
        \filldraw[draw, circle, fill=white, minimum size=6pt](0,.8) circle (8pt);
        \node at (0,.8) {$U_t$};
        \node at (1,0) {$=$};
        \draw (3.5,0) --  (2.3,0);
        \draw (3.5,0) --  (1.3,0);
        \filldraw[fill = white] (1.7,-.3) rectangle node {$A$} (2.3,.3);
        \filldraw[draw, circle, fill=white, minimum size=6pt](2.8,0) circle (8pt);
        \node at (2.8,0) {$V'_t$};
        \node at (3.6,0) {,};
\end{tikzpicture}
\end{aligned}  
\end{equation} 
where $V_t = R \bar{W}_t R^{-1}$ and $V'_t = R\bar{W}'_t R^{-1}$. Repeating the above procedure with increasing system size, based on the transversal feedback corrections, we can construct a right-pushable defect $V_t$ along with the set $\{V^{[n]}\}$ and $\{U^{[n]}\}$ with relations as in Eq.~\eqref{eq:pushing_U_right} for every measurement outcome $t$. Therefore, the preparability of a target state $\ket{A_{\text{open}}}$ is completely determined by its pushable defects, which can be summarized as the following lemma.

\begin{lemma}
    An MPS $\ket{A_{\text{open}}}$ can be prepared from another MPS $\ket{B_{\text{open}}}$ using a single round of single-site measurement and transversal left-conditioned feedback correction for any system size with open boundary condition, {\it if and only if} the MPS tensor $B$ is associated with $A$ as in Eq.~\eqref{eq:B_definition} through a set of right-pushable defects of $A$.
    \label{lemma:AB_OBC}
\end{lemma}

When we allow the use of FDLU in the preparation of a target state, we can allow any local measurement, since FDLU and site blocking can map all local measurements to single-site measurements. With that, we define a notion of MF-preparability with open boundary condition:
\begin{definition}
    An MPS $\ket{A_{\text{open}}}$ with OBC is {\it OMF-preparable}, if it can be prepared from a product state for any system size, using a finite sequence of FDLU and MF circuits. An MPS $\ket{A_{\text{open}}}$ is {\it OLMF-preparable}, if it is OMF-preparable and all the feedback corrections are left-conditioned.
\end{definition}

As shown in Fig.~\ref{fig:circuit}, for an OLMF-preparable state $\ket{A_{\text{open}}}$, there exists a series of states \{$\ket{A^{(j)}}$, $\ket{B^{(j)}}$\}, with $j=0,\cdots,r$ for a finite number $r$. To study the OLMF-preparation of this target state, one has to find all the intermediate states along with the circuits relating them,
\beq
    \ket{B^{(j)}} &= \text{FDLU}_j \ket{A^{(j)}},\\ \ket{A^{(j)}} &= \text{MF}_j \ket{B^{(j-1)}}.
\eeq
From our results in the previous section and Lemma \ref{lemma:AB_OBC} above, we are guaranteed that using the construction of associated states from sets of right-pushable defects, all the intermediate states along with the circuits can be constructed, which can be summarized as the following theorem: 

\begin{theorem}
    The preparation circuit for any OLMF-preparable state $\ket{A_{\text{open}}}$ can be obtained from the sequential constructions of associated states using right-pushable defects.
    \label{theorem:preparation}
\end{theorem}

In other words, our preparation scheme is powerful enough to give rise to the preparation circuit of any target state as long as it is OLMF-preparable. 

To demonstrate an OLMF-preparation circuit, we take the biased cat state as an example,
\beq
    \ket{\text{GHZ}_\gamma}\propto \ket{00\cdots}+\gamma\ket{11\cdots}.
\eeq
It can be represented as an MPS with OBC, by taking the same MPS tensor as that of the GHZ state, along with the two appropriate boundary tensors: $L=\gamma^{(1-Z)/2}$ and $R=I$:
\begin{equation}
\begin{aligned}
\ket{\text{GHZ}_\gamma}=
\ \begin{tikzpicture}[baseline=(current bounding box),scale=1.2]      
    \draw (.3,0) --  (1.7,0);
    \draw (.3,0) --  (.3,0.5);
        \draw (1,0) -- (1,.5);
        \node at (2.,0) {$...$};
        \draw (2+.3,0) --  (3+.7,0);
        \draw (3,0) -- (3,.5);
        \draw (3.7,0) --  (3.7,0.5);
        \filldraw[draw, circle, fill=white, minimum size=6pt](.3,0) circle (6pt);
        \node at (.3,0) {$L$};
        \fill (1,0) circle (1pt);
        \fill (3,0) circle (1pt);
\end{tikzpicture} \quad .
\end{aligned} 
\end{equation}
The MPS tensor of the GHZ state is simply the three-leg delta tensor, which allows both Pauli $Z$ and $X$ to be right-pushable. An associated tensor $B$ satisfying Eq.~\eqref{eq:pushable_defect_OBC} will be shown later in Eq.~\eqref{eq:B_fusion_measurement} in the next section, which generates an FDLU-preparable state $\ket{B}$. 

We demonstrate the preparation of the biased cat state in the following. We consider $N=3$ for simplicity, and the generalization to any $N>3$ is straightforward. Starting from the following product state of seven qubits,
\begin{equation}
\begin{aligned}
\begin{tikzpicture}[baseline=(current bounding box),scale=1.2]      
    \draw (.3,0) --  (1,0);
    \draw (0.3,0) --  (0.3,0.5);
    \draw (1,0) -- (1,.5);
    \draw (1.7,0) -- (1.7,.5);
    \node at (1.7,-0.3) {$\ket{+}$};
    \filldraw[draw, circle, fill=white, minimum size=6pt](.3,0) circle (6pt);
    \node at (.3,0) {$L$};

    \draw (2.5,0) --  (3.2,0);
    \draw (2.5,0) --  (2.5,0.5);
    \draw (3.2,0) -- (3.2,.5);
    \draw (3.9,0) -- (3.9,.5);
    \node at (3.9,-0.3) {$\ket{+}$};

    \draw (4.7,0) -- (4.7,.5);
    \node at (4.8,-0.3) {$\ket{+}$};
\end{tikzpicture}
\end{aligned} 
\end{equation}
we apply the following FDLU composed of $CZ$ and $H$ gates (one to seven label qubits from left to right)
\beq
    H_2 H_3 H_5 H_6 CZ_{23}CZ_{34}CZ_{56}CZ_{67},
\eeq
to obtain the associated state $\ket{B}$: 
\begin{equation}
\begin{aligned}
\begin{tikzpicture}[baseline=(current bounding box),scale=1.2]      
    \draw (.3,0) --  (6,0);
    \draw (0.3,0) --  (0.3,0.8);
        \draw (1,0) -- (1,.8);
        \draw (2,0) -- (2,.8);
        
        \draw (3,0) -- (3,.8);
        \draw (4,0) --  (4,0.8);
        \draw (5,0) --  (5,0.8);
        \draw (6,0) --  (6,0.8);
        
        \filldraw[draw, circle, fill=white, minimum size=6pt](.3,0) circle (6pt);
        \node at (.3,0) {$L$};
        \fill (1,0) circle (1pt);
        \fill (2,0) circle (1pt);
        \fill (3,0) circle (1pt);
        \fill (4,0) circle (1pt);
        \fill (5,0) circle (1pt);
        
        \filldraw[draw, circle, fill=myGray, minimum size=6pt](1,.4) circle (6pt);
        \node at (1,.4) {$H$};
        \filldraw[draw, circle, fill=myGray, minimum size=6pt](1.5,0) circle (6pt);
        \node at (1.5,0) {$H$};
        \filldraw[draw, circle, fill=myGray, minimum size=6pt](2,0.4) circle (6pt);
        \node at (2,0.4) {$H$};
        \filldraw[draw, circle, fill=myGray, minimum size=6pt](2.5,0) circle (6pt);
        \node at (2.5,0) {$H$};

        \filldraw[draw, circle, fill=myGray, minimum size=6pt](4,.4) circle (6pt);
        \node at (4,.4) {$H$};
        \filldraw[draw, circle, fill=myGray, minimum size=6pt](4.5,0) circle (6pt);
        \node at (4.5,0) {$H$};
        \filldraw[draw, circle, fill=myGray, minimum size=6pt](5,0.4) circle (6pt);
        \node at (5,0.4) {$H$};
        \filldraw[draw, circle, fill=myGray, minimum size=6pt](5.5,0) circle (6pt);
        \node at (5.5,0) {$H$};
\end{tikzpicture} \ .
\end{aligned} 
\end{equation}
Then we apply single-site measurements in the computational basis on the second, third, fifth, and sixth qubits. After a left-conditioned feedback correction, we obtain the target biased cat state.

We conclude this subsection by noting that OLMF-preparability is a stricter criterion than MF-preparability since it requires an open boundary and also requires the feedback corrections to be left-conditioned. It remains unclear whether the latter requirement can be removed, such that the preparation of any OMF-preparable state can be constructed using our scheme. We leave this question for future exploration. In the remainder of this work, we will systematically analyze the various sets of pushable defects and explore the preparability of the associated MPSs.

\section{All Pauli defects are pushable}
\label{sec:all_pauli}

In this section, we consider the classes of states for which the set of pushable defects contains all the Pauli operators in the virtual space, and construct their preparation circuits. In this setting, while our protocols are equivalent to the fusion measurement protocols~\cite{sahay2025classifying,smith2024constant,stephen2024preparing,zhang2024characterizing},  we will provide a more detailed classification of fusion measurement preparable states by considering various pushable defects and classifying their associated pushing relations. This extends the structure theorem in Ref.~\cite{zhang2024characterizing}, which states that all these states must admit a sequential Clifford structure. In particular, our formalism naturally reveals a relation between these states and non-invertible duality transformations.

To begin, we first briefly review the fusion measurement protocol used in Ref.~\cite{sahay2025classifying,smith2024constant,stephen2024preparing,zhang2024characterizing}. To prepare an MPS generated by the tensor $A$, we first prepare the following product state:
\beq
    \begin{tikzpicture}[baseline=(current bounding box),scale=1.5]
        \draw[rounded corners=8pt]
        (-.3,.5) -- (-.3,0) -- (-1.1,0) -- (-1.1,.5);
        \draw[rounded corners=8pt]
        (0,.5) -- (0,0) -- (.8,0) -- (.8,0.5);
        \draw[rounded corners=8pt]
        (1.1,.5) -- (1.1,0) -- (1.9,0) -- (1.9,0.5);
        \draw (-.7,0) -- (-.7,.5);
        \draw (.4,0) -- (.4,.5);
        \draw (1.5,0) -- (1.5,.5);
        \filldraw[fill = white] (-.7-.2,-.2) rectangle node {$A$} (-.7+.2,.2);
        \filldraw[fill = white] (.4-.2,-.2) rectangle node {$A$} (.4+.2,.2);
        \filldraw[fill = white] (1.5-.2,-.2) rectangle node {$A$} (1.5+.2,.2);
        \node at (2.4,0.2) {...};
        \node at (-1.6,0.2) {...};
        \end{tikzpicture}
\eeq
where each 3-leg tensor $A$ has one physical bond of dimension $d$ in the center and two ancillary bonds of dimension $D$ on the left and right. A measurement on two adjacent ancillary bonds in an orthonormal basis fuses (glues) two bonds together with potential defects. In a $D$-dimensional space, we define the generalized Pauli-$Z$ operator as $Z = \sum_{j=0}^{D-1}e^{2\pi i j/D}\ket{j}\bra{j}$, and the generalized Pauli-$X$ as $X = \sum_{j=0}^{D-1}\ket{(j+1)\text{ mod }D}\bra{j}$ such that they satisfy $ZX = \omega X Z$, $\omega = \exp(2\pi i/D)$. If we choose the Bell basis ($Z\otimes Z^\dagger$ and $X\otimes X^\dagger$), the resulting state after the fusion is given by:
\beq
    \begin{tikzpicture}[baseline=(current bounding box),scale=1.5]
        \draw (-1.3,0) -- (2.1,0);
        \draw (-.7,0) -- (-.7,.5);
        \draw (.4,0) -- (.4,.5);
        \draw (1.5,0) -- (1.5,.5);
        \filldraw[fill = white] (-.7-.2,-.2) rectangle node {$A$} (-.7+.2,.2);
        \filldraw[fill = white] (.4-.2,-.2) rectangle node {$A$} (.4+.2,.2);
        \filldraw[fill = white] (1.5-.2,-.2) rectangle node {$A$} (1.5+.2,.2);
        \filldraw[draw, circle, fill=myRed, minimum size=6pt](-.15,0) circle (5pt);
        \node at (-.15,0) {$P$};
        \filldraw[draw, circle, fill=myRed, minimum size=6pt](.95,0) circle (5pt);
        \node at (.95,0) {$P'$};
        \node at (2.4,0) {...};
        \node at (-1.6,0) {...};
        \end{tikzpicture}
\eeq
where $P$ and $P'$ are Pauli defects associated with different measurement outcomes. When all such defects are pushable for the tensor $A$, we can perform an on-site feedback correction to push them away, thereby preparing $\ket{A}$ using a single round of measurements.

While our formalism adopts the pushability conditions of defects similar to the references that discuss fusion measurement~\cite{sahay2025classifying,smith2024constant,stephen2024preparing,zhang2024characterizing}, there is a notable difference. Specifically, the classification of preparable states in Ref.~\cite{zhang2024characterizing,sahay2025classifying} restricts to states with only right-pushable defects, whereas a general fusion measurement protocol allows left-pushable defects as well. We will consider the pushing relations in both directions and discuss the corresponding classification. It turns out that this generalization does not lead to additional capability in state preparation.

\subsection{Single-round MF preparation}

According to the scheme presented in the previous section, we can construct an associated MPS tensor $B$ from $A$ using its pushable defects (Eq.~\eqref{eq:B_definition}), so that $\ket{B}$ is a single round of MF away from $\ket{A}$. When all Pauli defects are pushable for $A$, we can construct $B$ as follows: 
\begin{equation}
\begin{aligned}
\begin{tikzpicture}[baseline=(current bounding box),scale=1.2]
        \draw (-.5,.3) -- (-.5,.7);
        \draw (.1,.3) -- (.1,.7);
        \draw (-.2,.3) -- (-.2,.7);
        \filldraw[fill = white] (-.7,-.3) rectangle node {$B$} (.3,.3);
        \draw (.7,0) --  (.3,0);
        \draw (-.7,0) --  (-1.1,0);
        \node at (1,0) {$\equiv$};
        \filldraw[fill = white] (1.7,-.3) rectangle node {$A$} (2.3,.3);
        \draw (1.7,0) --  (1.3,0);
        \draw (2,.3) -- (2,.9);
        \draw (4.2,0) --  (2.3,0);
        \fill (2.6,0) circle (1pt);
        \draw (2.6,0.9) --  (2.6,0);
        \filldraw[draw, circle, fill=myGray, minimum size=6pt](2.6,.45) circle (6pt);
        \node at (2.6,.45) {$H$};
        \draw (2.6,0) --  (3.2,0);
        \fill (3.4,0) circle (1pt);
        \draw (3.4,0.9) --  (3.4,0);
        \filldraw[draw, circle, fill=myGray, minimum size=6pt](3.4,.45) circle (6pt);
        \node at (3.4,.45) {$H$};
        \filldraw[draw, circle, fill=myGray, minimum size=6pt](3,0) circle (6pt);
        \node at (3,0) {$H$};
        \filldraw[draw, circle, fill=myGray, minimum size=6pt](3.8,0) circle (6pt);
        \node at (3.8,0) {$H$};
\end{tikzpicture}
\label{eq:B_fusion_measurement}
\end{aligned}  
\end{equation} 
where we use a black dot on the three-leg junction to denote the $\delta$-tensor, and $H$ denotes the generalized Hadamard gate
\beq
    H = \frac{1}{\sqrt{d}}\sum_{j,k=0}^{d-1}e^{\frac{-2\pi i  jk}{d}}\ket{j}\bra{k},
    \label{eq:Hadamard}
\eeq
which satisfies $H Z H^{\dagger}=X$ and $H X H^{\dagger}=Z^{\dagger}$. The matrices in the diagram act from top to bottom and from left to right. The above construction of the tensor $B$ yields
\begin{equation}
\begin{aligned}
\begin{tikzpicture}[baseline=(current bounding box),scale=1.2]
        \draw (-.5,.3) -- (-.5,.7);
        \draw (-.2,.3) -- (-.2,.7);
        \draw (.1,.3) -- (.1,.7);
        \node at (-.2,.9) {$\bar{m}$};
        \node at (.1,.9) {$\bar{n}$};
        \filldraw[fill = white] (-.7,-.3) rectangle node {$B$} (.3,.3);
        \draw (.7,0) --  (.3,0);
        \draw (-.7,0) --  (-1.1,0);
        \node at (1,0) {$=$};
        \draw (2.4-.7,0) --  (2.4-1.1,0);
        \draw (2.4-.4,.3) -- (2.4-.4,.7);
        \filldraw[fill = white] (2.4-.7,-.3) rectangle node {$A$} (2.4-.1,.3);
        \draw (2.4-.1,0) -- (2.4+1.6,0); 
        \filldraw[draw, circle, fill=myBlue, minimum size=6pt](2.4+.3,0) circle (8pt);
        \node at (2.4+.3,0) {$Z^m$};
        \filldraw[draw, circle, fill=myRed, minimum size=6pt](2.4+1,0) circle (8pt);
        \node at (2.4+1,0) {$X^n$};
        \node at (2.4+1.7,0) {,};
\end{tikzpicture}
\label{eq:B_fusion_measurement_component}
\end{aligned}  
\end{equation} 
where $\overline{m}\equiv -m$. In other words, measuring the middle and the right physical legs on the tensor $B$ leads to the tensor $A$ up to certain Pauli defects on the virtual legs. This is exactly the condition in Eq.~\eqref{eq:B_definition}, with the pushable defects $\{V_t\}$ being generalized Pauli defects, and with the two right physical bonds of $B$ in Eq.~\eqref{eq:B_fusion_measurement_component} combined into one physical bond in Eq.~\eqref{eq:B_definition}. Following Eq.~\eqref{eq:T_B}, the transfer matrix $T_B$ is thus given by
\beq
    T_B 
    &=\sum_{m,n=0}^{D-1} \ \begin{tikzpicture}[baseline=(current bounding box),scale=1.2]
        \draw (0,.3) -- (0,.7);
        \filldraw[fill = white] (-.3,-.3) rectangle node {$A$} (.3,.3);
        \draw (1.9,0) --  (0.3,0);
        \draw (-.7,0) --  (-0.3,0);
        \filldraw[fill = white] (-.3,.7) rectangle node {$\bar{A}$} (.3,1.3);
        \draw (1.9,1) --  (0.3,1);
        \draw (-.7,1) --  (-0.3,1);
        \filldraw[draw, circle, fill=myBlue, minimum size=6pt](.7,.0) circle (8pt);
        \node at (.7,.0) {$Z^{m}$};
        \filldraw[draw, circle, fill=myBlue, minimum size=6pt](.7,1) circle (8pt);
        \node at (.7,1) {$Z^{\bar{m}}$};
        \filldraw[draw, circle, fill=myRed, minimum size=6pt](1.4,.0) circle (8pt);
        \node at (1.4,.0) {$X^{n}$};
        \filldraw[draw, circle, fill=myRed, minimum size=6pt](1.4,1) circle (8pt);
        \node at (1.4,1) {$X^{n}$};
    \end{tikzpicture} \\[2ex]
    &=\ 
    \begin{tikzpicture}[baseline=(current bounding box),scale=1.5]
        \draw (-0.3,0) --  (-0.6,0);
        \draw (-0.3,.5) --  (-0.6,.5);
        \draw (0,.3) -- (0,.5);
        \filldraw[fill = white] (-.3,-.2) rectangle node {$T_A$} (.3,.7);
        \draw[rounded corners=8pt]
        (.3,.5) -- (.6,.5) -- (.6,0) -- (.3,0);
        \draw[rounded corners=8pt]
        (1,.5) -- (.7,.5) -- (.7,0) -- (1,0);
    \end{tikzpicture}\quad,
    \label{eq:B_A}
\eeq
In particular, this transfer matrix is always of rank 1, and according to Lemma \ref{lemma:FDLU}, the MPS $\ket{B}$ is FDLU-preparable. Therefore, a single-round MF-preparation protocol for $\ket{A}$ involves preparing $\ket{B}$ from a product state using an FDLU, followed by a round of local measurement and feedback corrections. This protocol is equivalent to the fusion measurement protocol, as the extra Hadamard gates and delta tensors in Eq.~\eqref{eq:B_fusion_measurement} just transform the Bell basis measurement on the virtual bonds of $A$, to the computation basis measurement performed on the right two physical bonds of B.

In the following section, we study the structure of these states by taking a closer look at the pushing relation of each Pauli defect.

\subsection{Right-pushable Pauli defects}
\label{subsec:right_pushable}

We first focus on the states for which all the Pauli defects are right-pushable: 
\beq
    \big(\bar{Z}\otimes Z\big) T_A &= T_A \big(\bar{P}_Z\otimes P_Z\big),\\
    \big(\bar{X}\otimes X\big) T_A &= T_A \big(\bar{P}_X\otimes P_X\big),
    \label{eq:pushing_relation}
\eeq
where $P_Z$ and $P_X$ are products of generalized $Z$ and $X$ to some power, and $\bar{O}$ denotes the complex conjugation of an operator $O$. Ref.~\cite{zhang2024characterizing} proved that such MPSs are preparable from some product state $\ket{\psi}^{\otimes N}$ using a sequential Clifford circuit followed by a transversal unitary circuit, up to disentangled ancilla qudits. Here, we will provide a finer classification based on various classes of $P_X, P_Z$.

We begin the discussion by considering the $D=2$ bond dimension.  Since the error defects are defined up to a phase, the Pauli defects form a $\mathbb{Z}_2\times \mathbb{Z}_2$ group generated by the $Z$ and $X$ defects. The pushing relations are classified into three types by the subgroup of defects that are annihilated upon pushing. First, all defects are annihilated, i.e., all Pauli defects are pushed to identity $Z, X\mapsto I$ after site blocking. Second, all the Pauli defects but one are eventually pushed to identity. Then we can make a gauge transformation to map the remaining non-vanishing Pauli defect to $X$, i.e., $Z\mapsto I$ and $X\mapsto X$. Third, no defect is annihilated after pushing. The pushing relation in this case essentially defines an automorphism of the defect group $\mathbb{Z}_2\times \mathbb{Z}_2$. Since $\operatorname{Aut}(\mathbb{Z}_2\times \mathbb{Z}_2)=S_3$, the orders of all automorphisms are a factor of six. Therefore, after a blocking of six sites, any pushing relation of the third type becomes the identity automorphism, i.e., $Z\mapsto Z$ and $X\mapsto X$. 

To summarize, when $D=2$, the pushing relation can be brought to one of the above three types via a site blocking and gauge transformation. This result can be extended to general bond dimensions. Specifically, we prove in Appendix \ref{app:coarse_grain} that, we can decompose the $D$-dimensional virtual space into three subspaces such that after site blocking and appropriate gauge transformation: in the first subspace we have $Z, X\mapsto I$; in the second subspace we have $Z\mapsto I$ and $X\mapsto X$; and in the third subspace we have $Z\mapsto Z$ and $X\mapsto X$.

We will now discuss states that have one of the above three classes of pushing relations for bond dimension $D$. Since we focus on the pushability of Pauli defects, it is especially useful to expand a general transfer matrix $T_A$ in the Bell basis as follows:
\beq
    T_A=\sum_{a,b,c,d=0}^{D-1}\lambda_{a,b,c,d}\quad \begin{tikzpicture}[baseline=(current bounding box),scale=1.5]
        \draw[rounded corners=8pt]
        (-.8,.8) -- (0,.8) -- (0,0) -- (-.8,0);
        \draw[rounded corners=8pt]
        (1,.8) -- (.2,.8) -- (.2,0) -- (1,0);
        \filldraw[draw, circle, fill=myRed, minimum size=6pt](0.6,0) circle (6pt);
        \node at (0.6,0) {$X^{d}$};
        \filldraw[draw, circle, fill=myBlue, minimum size=6pt](0.6,.8) circle (6pt);
        \node at (0.6,.8) {$Z^{c}$};
        \filldraw[draw, circle, fill=myRed, minimum size=6pt](-0.4,0) circle (6pt);
        \node at (-0.4,0) {$X^{b}$};
        \filldraw[draw, circle, fill=myBlue, minimum size=6pt](-0.4,0.8) circle (6pt);
        \node at (-0.4,0.8) {$Z^a$};
        \end{tikzpicture}\ .
        \label{eq:transfer_matrix_bell}
\eeq
The coefficients $\lambda_{a,b,c,d}$ are constrained such that $T_A$ defines a completely positive map as shown in Eq.~\eqref{eq:CP_map}.

\subsubsection{class $RIRI$}
\label{subsec:RIRI}
We use $RIRI$ to denote that both $Z$ and $X$ are right-pushed to identity: $Z,X\mapsto I$. According to Eq.~\eqref{eq:pushing_relation}, the transfer matrix that supports this pushing relation satisfies:
\beq
    \big(\bar{Z}\otimes Z\big) T_A = 
    \big(\bar{X}\otimes X\big) T_A = T_A.
\eeq
Expanding $\big(\bar{Z}\otimes Z\big) T_A$ in the Bell basis we have
\beq
    &\sum_{a,b,c,d=0}^{D-1}\lambda_{a,b,c,d}\quad \begin{tikzpicture}[baseline=(current bounding box),scale=1.5]
        \draw[rounded corners=8pt]
        (-1.4,.8) -- (0,.8) -- (0,0) -- (-1.4,0);
        \draw[rounded corners=8pt]
        (1,.8) -- (.2,.8) -- (.2,0) -- (1,0);
        \filldraw[draw, circle, fill=myRed, minimum size=6pt](0.6,0) circle (6pt);
        \node at (0.6,0) {$X^{d}$};
        \filldraw[draw, circle, fill=myBlue, minimum size=6pt](0.6,.8) circle (6pt);
        \node at (0.6,.8) {$Z^{c}$};
        \filldraw[draw, circle, fill=myRed, minimum size=6pt](-0.4,0) circle (6pt);
        \node at (-0.4,0) {$X^{b}$};
        \filldraw[draw, circle, fill=myBlue, minimum size=6pt](-0.4,0.8) circle (6pt);
        \node at (-0.4,0.8) {$Z^a$};
        \filldraw[draw, circle, fill=myBlue, minimum size=6pt](-1,0) circle (6pt);
        \node at (-1,0) {$Z$};
        \filldraw[draw, circle, fill=myBlue, minimum size=6pt](-1,0.8) circle (6pt);
        \node at (-1,0.8) {$\bar{Z}$};
        \end{tikzpicture}\\[2ex]
        =&\sum_{a,b,c,d=0}^{D-1}e^{\frac{2\pi i b}{D}}\lambda_{a,b,c,d}\quad \begin{tikzpicture}[baseline=(current bounding box),scale=1.5]
        \draw[rounded corners=8pt]
        (-.8,.8) -- (0,.8) -- (0,0) -- (-.8,0);
        \draw[rounded corners=8pt]
        (1,.8) -- (.2,.8) -- (.2,0) -- (1,0);
        \filldraw[draw, circle, fill=myRed, minimum size=6pt](0.6,0) circle (6pt);
        \node at (0.6,0) {$X^{d}$};
        \filldraw[draw, circle, fill=myBlue, minimum size=6pt](0.6,.8) circle (6pt);
        \node at (0.6,.8) {$Z^{c}$};
        \filldraw[draw, circle, fill=myRed, minimum size=6pt](-0.4,0) circle (6pt);
        \node at (-0.4,0) {$X^{b}$};
        \filldraw[draw, circle, fill=myBlue, minimum size=6pt](-0.4,0.8) circle (6pt);
        \node at (-0.4,0.8) {$Z^a$};
        \end{tikzpicture}\ .
        \label{eq:ZZ_TA}
\eeq
Therefore, $\big(\bar{Z}\otimes Z\big) T_A =T_A$ constrains $\lambda_{a,b,c,d}$ to be non-zero only when $b=0$. Similarly, $\big(\bar{X}\otimes X\big) T_A =T_A$ constrains $\lambda_{a,b,c,d}$ to be non-zero only when $a=0$. The general form of $T_A$ with this pushing relation is thus given by:
\beq
        T_A&=\sum_{c,d=0}^{D-1}\lambda_{c,d}\quad \begin{tikzpicture}[baseline=(current bounding box),scale=1.5]
        \draw[rounded corners=8pt]
        (-.6,.8) -- (0,.8) -- (0,0) -- (-.6,0);
        \draw[rounded corners=8pt]
        (1,.8) -- (.2,.8) -- (.2,0) -- (1,0);
        \filldraw[draw, circle, fill=myRed, minimum size=6pt](0.6,.8) circle (6pt);
        \node at (0.6,.8) {$Z^{c}$};
        \filldraw[draw, circle, fill=myBlue, minimum size=6pt](0.6,0) circle (6pt);
        \node at (0.6,0) {$X^{d}$};
        \end{tikzpicture} = \quad \begin{tikzpicture}[baseline=(current bounding box),scale=1.5]
        \draw[rounded corners=8pt]
        (-.6,.8) -- (0,.8) -- (0,0) -- (-.6,0);
        \draw[rounded corners=8pt]
        (1,.8) -- (.4,.8) -- (.4,0) -- (1,0);
        \filldraw[draw, circle, fill=white, minimum size=6pt](0.4,.4) circle (6pt);
        \node at (0.4,.4) {$M$};
        \end{tikzpicture}
\eeq
where the coefficients $\{\lambda_{c,d}\}$ are constrained such that the matrix $M:= \sum_{c,d} \lambda_{c,d}Z^{c}X^{d}$ is a non-negative Hermitian matrix. This transfer matrix is of rank 1 and thus MPS $\ket{A}$ is FDLU-preparable according to Sec.~\ref{sec:FDLU}. Here we give a further description on $\ket{A}$, by diagonalizing $M=V^{\dagger}\Sigma V$ into $\Sigma = \operatorname{diag}(\sigma_0,\cdots, \sigma_{D-1})$. After a gauge transformation by $V$, the representative MPS tensor of $T_A$ is simply given by
\beq
        A=\quad \begin{tikzpicture}[baseline=(current bounding box),scale=1.5]
        \draw[rounded corners=8pt]
        (0,.6) -- (0,0) -- (-.7,0);
        \draw[rounded corners=8pt]
        (.2,.6) -- (.2,0) -- (1.2,0);
        \filldraw[draw, circle, fill=white, minimum size=6pt](0.7,0) circle (6pt);
        \node at (0.7,0) {$\Sigma^{\frac{1}{2}}$};
        \end{tikzpicture}\ ,
\eeq
up to a unitary gate acting on the physical bonds~\cite{cirac2021matrix,QChannelLecture}. This tensor generates a product state. Therefore, we conclude that the states of this class after site blocking is related to some product state via a transversal unitary circuit. Namely, $\ket{A}$ is FDLU-preparable. 

There is another way to see why $\ket{A}$ is FDLU-preparable. We will present it here since it will be helpful for later discussions. The pushing relation $Z,X\mapsto I$ requires unitary operators $U_Z$ and $U_X$ whose actions on the physical bond annihilate the corresponding defects as in Eq.~\eqref{eq:pushing_U_right}. Therefore, we will show below that the target state $\ket{A}$ can be obtained from the associated FDLU-preparable state $\ket{B}$ in Eq.~\eqref{eq:B_fusion_measurement_component} via a controlled-unitary circuit between the ancillary bonds and physical bonds: 
\begin{equation}
\begin{aligned}
&\begin{tikzpicture}[baseline=(current bounding box),scale=1.2]
        \node at (-1.4,0) {$\cdots$};
        \draw (-.4,.3) -- (-.4,1.4);
        \draw (1,.3) -- (1,1.4);
        \draw (2.4,.3) -- (2.4,1.4);
        \filldraw[fill = white] (-.7,-.3) rectangle node {$A$} (-.1,.3);
        \draw (.7,0) --  (-.1,0);
        \draw (-.7,0) --  (-1.1,0);
        \filldraw[fill = white] (.7,-.3) rectangle node {$A$} (1.3,.3);
        \draw (1.3,0) --  (2.1,0);
        \filldraw[fill = white] (2.1,-.3) rectangle node {$A$} (2.7,.3);
        \draw (2.7,0) --  (3.5,0);
        \node at (3.8,0) {$\cdots$};
        \draw (-.1,.8) -- (-.1,1.4);
        \draw (.5,.8) -- (.5,1.4);
        \node at (-.1,.6) {$\ket{+}$};
        \node at (.5,.6) {$\ket{+}$};
        \draw (-.1+1.4,.8) -- (-.1+1.4,1.4);
        \draw (.5+1.4,.8) -- (.5+1.4,1.4);
        \node at (-.1+1.4,.6) {$\ket{+}$};
        \node at (.5+1.4,.6) {$\ket{+}$};
        \draw (-.1+2.8,.8) -- (-.1+2.8,1.4);
        \draw (.5+2.8,.8) -- (.5+2.8,1.4);
        \node at (-.1+2.8,.6) {$\ket{+}$};
        \node at (.5+2.8,.6) {$\ket{+}$};
\end{tikzpicture}\\[2ex]
=\ &\begin{tikzpicture}[baseline=(current bounding box),scale=1.2]
        \node at (-1.4,0) {$\cdots$};
        \draw (-.5,.3) -- (-.5,.5);
        \draw (-.5,1) -- (-.5,1.2);
        \draw (-.5,1.7) -- (-.5,1.9);
        \draw (-.2,.3) -- (-.2,1.2);
        \draw (-.2,1.7) -- (-.2,1.9);
        \draw (.1,.3) -- (.1,1.9);
        \node at (-.6,1.45) {$U_Z$};
        \node at (-.6,.75) {$U_X$};
        \draw (-1,1.2) -- (-.4,1.2);
        \draw (-.4,1.2) -- (-.4,1.7);
        \draw (-1,1.7) -- (-.4,1.7);
        \draw (-1,.5) -- (-.4,.5);
        \draw (-.4,.5) -- (-.4,1);
        \draw (-1,1) -- (-.4,1);
        \filldraw[fill = white] (-.7,-.3) rectangle node {$B$} (.3,.3);
        \draw (.7,0) --  (.3,0);
        \draw (-.7,0) --  (-1.1,0);
        \draw (.9,.3) -- (.9,1.2);
        \draw (.9,1.7) -- (.9,1.9);
        \draw (1.5,.3) -- (1.5,1.9);
        \draw (1.2,.3) -- (1.2,1.2);
        \draw (1.2,1.7) -- (1.2,1.9);
        \filldraw[fill = white] (.7,-.3) rectangle node {$B$} (1.7,.3);
        \draw (1.7,0) --  (2.1,0);
        \filldraw[fill = white] (2.1,-.3) rectangle node {$B$} (3.1,.3);
        \draw (2.3,.3) -- (2.3,1.2);
        \draw (2.3,1.7) -- (2.3,1.9);
        \draw (2.9,.3) -- (2.9,.5);
        \draw (2.9,1) -- (2.9,1.9);
        \draw (2.6,.3) -- (2.6,1.2);
        \draw (2.6,1.7) -- (2.6,1.9);
        \draw (3.1,0) --  (3.5,0);
        \filldraw[fill = white] (0,.5) rectangle node {} (1,1);
        \node at (.2,.75) {$C$};
        \node at (.8,.75) {$U_X$};
        \draw (-.3,1.2) -- (0,1.2);
        \draw (0.2,1.2) -- (1,1.2);
        \draw (-.3,1.2) -- (-.3,1.7);
        \draw (1,1.2) -- (1,1.7);
        \draw (-.3,1.7) -- (0,1.7);
        \draw (0.2,1.7) -- (1,1.7);
        \node at (-.1,1.45) {$C$};
        \node at (.8,1.45) {$U_Z$};
        \filldraw[fill = white] (0+1.4,.5) rectangle node {} (1+1.4,1);
        \node at (.2+1.4,.75) {$C$};
        \node at (.8+1.4,.75) {$U_X$};
        \draw (-.3+1.4,1.2) -- (0+1.4,1.2);
        \draw (0.2+1.4,1.2) -- (1+1.4,1.2);
        \draw (-.3+1.4,1.2) -- (-.3+1.4,1.7);
        \draw (1+1.4,1.2) -- (1+1.4,1.7);
        \draw (-.3+1.4,1.7) -- (0+1.4,1.7);
        \draw (0.2+1.4,1.7) -- (1+1.4,1.7);
        \node at (-.1+1.4,1.45) {$C$};
        \node at (.8+1.4,1.45) {$U_Z$};
        \draw (-.3+2.8,1.2) -- (0+2.8,1.2);
        \draw (0.2+2.8,1.2) -- (.5+2.8,1.2);
        \draw (-.3+2.8,1.2) -- (-.3+2.8,1.7);
        \draw (-.3+2.8,1.7) -- (0+2.8,1.7);
        \draw (0.2+2.8,1.7) -- (.5+2.8,1.7);
        \node at (-.1+2.8,1.45) {$C$};
        \draw (.5+2.8,.5) -- (2.8,.5);
        \draw (2.8,1) -- (.5+2.8,1);
        \draw (2.8,.5) -- (2.8,1);
        \node at (.2+2.8,.75) {$C$};
        \node at (3.8,0) {$\cdots$};
\end{tikzpicture}
\label{eq:A_B_RIRI}
\end{aligned}  
\end{equation} 

The above equation can be seen from the local stabilizers. We denote a single-qudit operator $O$ supported on the three physical bonds at site $s$ from left to right as $O_s$, $\tilde{O}_s$, and $\dbtilde{O}_s$ respectively (see LHS of Eq.~\eqref{eq:B_fusion_measurement}). Following the definition of right-pushable defect in Eq.~\eqref{eq:pushing_U_right} and the properties of $B$ in Eq.~\eqref{eq:B_fusion_measurement}, operator $\tilde{X}_s \dbtilde{Z}_s^{\dagger} U_{Z,s+1}^{\dagger}$ and $\dbtilde{X}_s U_{X,s+1}^{\dagger}$ stabilize $\ket{B}$ for any $s$. Through the action of a controlled-unitary circuit as shown in the RHS of the above equation, these stabilizers are mapped to $\tilde{X}_s$ and $\dbtilde{X}_s$. As a result, the second and third physical bonds at each site are in $\ket{+}$ state and disentangled from the rest of the degrees of freedom, as shown in the LHS of the above equation.

\subsubsection{class $RIRX$}
\label{subsec:RIRX}
We use $RIRX$ to denote that $Z$ is right-pushed to $I$ and $X$ is right-pushed to $X$: $Z\mapsto I$, $X\mapsto X$. According to Eq.~\eqref{eq:pushing_relation}, the transfer matrix that supports this pushing relation satisfies:
\beq
    \big(\bar{Z}\otimes Z\big) T_A = T_A,\quad  
    \big(\bar{X}\otimes X\big) T_A = T_A \big(\bar{X}\otimes X\big).
\eeq
As in Eq.~\eqref{eq:ZZ_TA}, the first equation constrains $\lambda_{a,b,c,d}$ to be non-zero only when $b=0$. Expanding $\big(\bar{X}\otimes X\big) T_A \big(\bar{X}\otimes X\big)^{\dagger}$ by the Bell basis, we have 
\beq
    &\sum_{a,b,c,d=0}^{D-1}\lambda_{a,b,c,d}\quad \begin{tikzpicture}[baseline=(current bounding box),scale=1.5]
        \draw[rounded corners=8pt]
        (-1.4,.8) -- (0,.8) -- (0,0) -- (-1.4,0);
        \draw[rounded corners=8pt]
        (1.6,.8) -- (.2,.8) -- (.2,0) -- (1.6,0);
        \filldraw[draw, circle, fill=myRed, minimum size=6pt](0.6,0) circle (6pt);
        \node at (0.6,0) {$X^{d}$};
        \filldraw[draw, circle, fill=myBlue, minimum size=6pt](0.6,.8) circle (6pt);
        \node at (0.6,.8) {$Z^{c}$};
        \filldraw[draw, circle, fill=myRed, minimum size=6pt](-0.4,0) circle (6pt);
        \node at (-0.4,0) {$X^{b}$};
        \filldraw[draw, circle, fill=myBlue, minimum size=6pt](-0.4,0.8) circle (6pt);
        \node at (-0.4,0.8) {$Z^a$};
        \filldraw[draw, circle, fill=myRed, minimum size=6pt](-1,0) circle (6pt);
        \node at (-1,0) {$X$};
        \filldraw[draw, circle, fill=myRed, minimum size=6pt](-1,0.8) circle (6pt);
        \node at (-1,0.8) {$\bar{X}$};
        \filldraw[draw, circle, fill=myRed, minimum size=6pt](1.2,0) circle (6pt);
        \node at (1.2,0) {$X^{\dagger}$};
        \filldraw[draw, circle, fill=myRed, minimum size=6pt](1.2,0.8) circle (6pt);
        \node at (1.2,0.8) {$\bar{X}^{\dagger}$};
        \end{tikzpicture}\\[2ex]
        =&\sum_{a,b,c,d=0}^{D-1}e^{\frac{-2\pi i (c+a)}{D}}\lambda_{a,b,c,d}\quad \begin{tikzpicture}[baseline=(current bounding box),scale=1.5]
        \draw[rounded corners=8pt]
        (-.8,.8) -- (0,.8) -- (0,0) -- (-.8,0);
        \draw[rounded corners=8pt]
        (1,.8) -- (.2,.8) -- (.2,0) -- (1,0);
        \filldraw[draw, circle, fill=myRed, minimum size=6pt](0.6,0) circle (6pt);
        \node at (0.6,0) {$X^{d}$};
        \filldraw[draw, circle, fill=myBlue, minimum size=6pt](0.6,.8) circle (6pt);
        \node at (0.6,.8) {$Z^{c}$};
        \filldraw[draw, circle, fill=myRed, minimum size=6pt](-0.4,0) circle (6pt);
        \node at (-0.4,0) {$X^{b}$};
        \filldraw[draw, circle, fill=myBlue, minimum size=6pt](-0.4,0.8) circle (6pt);
        \node at (-0.4,0.8) {$Z^a$};
        \end{tikzpicture}\ .
        \label{eq:XX_TA_XX}
\eeq
This constrains $\lambda_{a,b,c,d}$ to be non-zero only when $c=-a$. As a result, the general form of $T_A$ with $RIRX$ pushing relation is as follows:
\beq
        T_A&=\sum_{a,d=0}^{D-1}\lambda_{a,d}\quad \begin{tikzpicture}[baseline=(current bounding box),scale=1.5]
        \draw[rounded corners=8pt]
        (-.8,.8) -- (-.2,.8) -- (-.2,0) -- (-.8,0);
        \draw[rounded corners=8pt]
        (1,.8) -- (.2,.8) -- (.2,0) -- (1,0);
        \filldraw[draw, circle, fill=myBlue, minimum size=6pt](-.2,0.4) circle (6pt);
        \node at (-0.2,0.4) {$Z^{a}$};
        \filldraw[draw, circle, fill=myBlue, minimum size=6pt](0.6,.8) circle (6pt);
        \node at (0.6,.8) {$Z^{\bar{a}}$};
        \filldraw[draw, circle, fill=myRed, minimum size=6pt](0.6,0) circle (6pt);
        \node at (0.6,0) {$X^{d}$};
        \end{tikzpicture}\\[2ex]
        &= \quad 
        \begin{tikzpicture}[baseline=(current bounding box),scale=1.5]
        \draw[rounded corners=8pt]
        (-.8,1.8) -- (-.2,1.8) -- (-.2,0) -- (-.8,0);
        \draw[rounded corners=8pt]
        (1,1.8) -- (.2,1.8) -- (.2,0) -- (1,0);
        \filldraw[fill = myRed] (-.4,.2) rectangle node {$CX^{\dagger}$} (.4,.6);
        \filldraw[fill = myRed] (-.4,1.2) rectangle node {$CX$} (.4,1.6);
        \filldraw[draw, circle, fill=white, minimum size=6pt](0.2,0.9) circle (6pt);
        \node at (0.2,0.9) {$M$};
        \end{tikzpicture}\ ,
\eeq
where we use the relation $CX_{1,2} Z^{\bar{a}}_2 X^d_2 CX^{\dagger}_{1,2} = Z_1^a Z^{\bar{a}}_2 X^d_2$ (1 and 2 label the qudits in the left and right respectively), and the Hermitian matrix $M:= \sum_{a,d}\lambda_{a,d}Z^{\bar{a}}X^{d}$. Similar to the previous class, by diagonalizing the non-negative Hermitian matrix $M$, we can write a representative MPS tensor as:
\beq
        A=&\quad \begin{tikzpicture}[baseline=(current bounding box),scale=1.5]
        \draw[rounded corners=8pt]
        (0,2) -- (0,0) -- (-.7,0);
        \draw[rounded corners=8pt]
        (.4,2) -- (.4,0) -- (1.1,0);
        \filldraw[fill = myRed] (-.2,.2) rectangle node {$CX^{\dagger}$} (.6,.6);
        \filldraw[draw, circle, fill=white, minimum size=6pt](0.4,1.1) circle (6pt);
        \node at (0.4,1.1) {$V$};
        \filldraw[draw, circle, fill=white, minimum size=6pt](0.4,1.6) circle (6pt);
        \node at (0.4,1.6) {$\Sigma^{\frac{1}{2}}$};
        \end{tikzpicture}=\quad
        \begin{tikzpicture}[baseline=(current bounding box),scale=1.5]
        \draw[rounded corners=8pt]
        (0,2) -- (0,0) -- (-.7,0);
        \draw[rounded corners=8pt]
        (.8,.7) -- (.8,0) -- (.4,0) -- (.4,0.8) -- (1.1,0.8);
        \draw (.8,.9) -- (.8,2);
        \filldraw[fill = myRed] (-.2,.2) rectangle node {$CX$} (.6,.6);
        \filldraw[draw, circle, fill=white, minimum size=6pt](0.8,1.2) circle (6pt);
        \node at (0.8,1.2) {$V$};
        \filldraw[draw, circle, fill=white, minimum size=6pt](0.8,1.7) circle (6pt);
        \node at (0.8,1.7) {$\Sigma^{\frac{1}{2}}$};
        \end{tikzpicture}\\[2ex]
        =&\quad
        \begin{tikzpicture}[baseline=(current bounding box),scale=1.5]
        \draw[rounded corners=8pt]
        (-.2,.8) -- (1,.8) -- (1,1.3);
        \draw[rounded corners=8pt]
        (0,-.3) -- (0,.4) -- (1.6,.4);
        \filldraw[fill = myRed] (.3,.2) rectangle node {} (.8,1);
        \draw[rounded corners=8pt] (.2,-.3) -- (.2,-.1) -- (1.3, -.1) -- (1.3, .3);
        \draw (1.3, .5) -- (1.3, 1.3);
        \node at (.1,-.5) {$\ket{\psi_{M}}$};
        \node at (.55,.4) {$X$};
        \node at (.55,.8) {$C$};
        \end{tikzpicture}\ \text{, where} \ket{\psi_M}=\
        \begin{tikzpicture}[baseline=(current bounding box),scale=1.5]
        \draw[rounded corners=8pt]
        (.4,1.8) -- (.4,0.5) -- (.8,0.5) -- (.8,1.8);
        \filldraw[draw, circle, fill=white, minimum size=6pt](0.8,.9) circle (6pt);
        \node at (0.8,.9) {$V$};
        \filldraw[draw, circle, fill=white, minimum size=6pt](0.8,1.4) circle (6pt);
        \node at (0.8,1.4) {$\Sigma^{\frac{1}{2}}$};
        \end{tikzpicture}\ .
        \label{eq:RIRX}
\eeq
In other words, $\ket{\psi_M}=\sum_{i=0}^{D-1}\sqrt{\sigma_i}\ket{v_i}\otimes\ket{i}$. In the above, $\{\sigma_i\}$ and $\{\ket{v_i}\}$ are respectively the eigenvalues and eigenvectors of the Hermitian matrix $M$. Choosing the time direction to be upward, we can see that the target state $\ket{A}$ is prepared from a product state using a sequential $CX$ circuit on a sublattice, i.e., the qudits labeled by $\{\ket{v_i}\}$. The above result also suggests that the target state can be prepared via the implementation of the Kramers-Wannier duality. We will discuss this perspective in more detail below (e.g. see Lemma \ref{lemma:KW_KT}).

\subsubsection{class $RZRX$}
\label{subsec:RZRX}
We use $RZRX$ to denote that $Z$ is right-pushed to $Z$ and $X$ is right-pushed to $X$: $Z\mapsto Z$, $X\mapsto X$. According to Eq.~\eqref{eq:pushing_relation}, the transfer matrix that supports this pushing relation satisfies:
\beq
    \big(\bar{Z}\otimes Z\big) T_A = T_A\big(\bar{Z}\otimes Z\big),\  
    \big(\bar{X}\otimes X\big) T_A = T_A \big(\bar{X}\otimes X\big).
\eeq
This constrains $\lambda_{a,b,c,d}$ in Eq.~\eqref{eq:transfer_matrix_bell} to be non-zero only when $c=\bar{a}$ and $d=\bar{b}$. The general form of the transfer matrix is as follows:
\beq
    T_A&=\sum_{a,b=0}^{D-1}\lambda_{a,b}\quad 
    \begin{tikzpicture}[baseline=(current bounding box),scale=1.5]
        \draw[rounded corners=8pt]
        (-.8,.8) -- (0,.8) -- (0,0) -- (-.8,0);
        \draw[rounded corners=8pt]
        (1,.8) -- (.2,.8) -- (.2,0) -- (1,0);
        \filldraw[draw, circle, fill=myRed, minimum size=6pt](0.6,0) circle (6pt);
        \node at (0.6,0) {$X^{\bar{b}}$};
        \filldraw[draw, circle, fill=myBlue, minimum size=6pt](0.6,.8) circle (6pt);
        \node at (0.6,.8) {$Z^{\bar{a}}$};
        \filldraw[draw, circle, fill=myRed, minimum size=6pt](-0.4,0) circle (6pt);
        \node at (-0.4,0) {$X^{b}$};
        \filldraw[draw, circle, fill=myBlue, minimum size=6pt](-0.4,0.8) circle (6pt);
        \node at (-0.4,0.8) {$Z^a$};
    \end{tikzpicture}\\[2ex]
    &= \sum_{a,b=0}^{D-1}\sum_{p,q=0}^{D-1}\frac{\lambda_{a,b}}{D}
    \begin{tikzpicture}[baseline=(current bounding box),scale=1.5]
        \draw[rounded corners=8pt]
        (-1.3,0) -- (1.3,0);
        \draw[rounded corners=8pt]
        (-1.3,.8) -- (1.3,.8);
        \filldraw[draw, circle, fill=myRed, minimum size=6pt](0.9,0) circle (6pt);
        \node at (0.9,0) {$X^{\bar{b}}$};
        \filldraw[draw, circle, fill=myBlue, minimum size=6pt](0.9,.8) circle (6pt);
        \node at (0.9,.8) {$Z^{\bar{a}}$};
        \filldraw[draw, circle, fill=myBlue, minimum size=6pt](-.3,.8) circle (6pt);
        \node at (-0.3,.8) {$\bar{Z}^{p}$};
        \filldraw[draw, circle, fill=myRed, minimum size=6pt](.3,.8) circle (6pt);
        \node at (.3,.8) {$\bar{X}^{q}$};
        \filldraw[draw, circle, fill=myRed, minimum size=6pt](-0.9,0) circle (6pt);
        \node at (-0.9,0) {$X^{b}$};
        \filldraw[draw, circle, fill=myBlue, minimum size=6pt](-0.9,0.8) circle (6pt);
        \node at (-0.9,0.8) {$Z^a$};
        \filldraw[draw, circle, fill=myBlue, minimum size=6pt](-.3,0) circle (6pt);
        \node at (-0.3,0) {$Z^{p}$};
        \filldraw[draw, circle, fill=myRed, minimum size=6pt](.3,0) circle (6pt);
        \node at (.3,0) {$X^{q}$};
    \end{tikzpicture}\\[2ex]
    &= \sum_{p,q=0}^{D-1} \mu_{p,q}
    \begin{tikzpicture}[baseline=(current bounding box),scale=1.5]
        \draw[rounded corners=8pt]
        (-.7,0) -- (.7,0);
        \draw[rounded corners=8pt]
        (-.7,.8) -- (.7,.8);
        \filldraw[draw, circle, fill=myBlue, minimum size=6pt](-.3,.8) circle (6pt);
        \node at (-0.3,.8) {$\bar{Z}^{p}$};
        \filldraw[draw, circle, fill=myRed, minimum size=6pt](.3,.8) circle (6pt);
        \node at (.3,.8) {$\bar{X}^{q}$};
        \filldraw[draw, circle, fill=myBlue, minimum size=6pt](-.3,0) circle (6pt);
        \node at (-0.3,0) {$Z^{p}$};
        \filldraw[draw, circle, fill=myRed, minimum size=6pt](.3,0) circle (6pt);
        \node at (.3,0) {$X^{q}$};
    \end{tikzpicture}\ ,
    \label{eq:RZRX}
\eeq
where 
\beq
    \mu_{p,q}\equiv \frac{1}{D}\sum_{a,b=0}^{D-1}e^{\frac{2\pi i}{D}(-pb+qa)}\lambda_{a,b}.
\eeq
In the second line of Eq.~\eqref{eq:RZRX}, we used the following identity
\beq
    \begin{tikzpicture}[baseline=(current bounding box),scale=1.5]
        \draw[rounded corners=8pt]
        (-.4,.8) -- (0,.8) -- (0,0) -- (-.4,0);
        \draw[rounded corners=8pt]
        (.6,.8) -- (.2,.8) -- (.2,0) -- (.6,0);
    \end{tikzpicture}\quad =\sum_{p,q=0}^{D-1} \frac{1}{D} \quad
    \begin{tikzpicture}[baseline=(current bounding box),scale=1.5]
        \draw[rounded corners=8pt]
        (-.7,0) -- (.7,0);
        \draw[rounded corners=8pt]
        (-.7,.8) -- (.7,.8);
        \filldraw[draw, circle, fill=myBlue, minimum size=6pt](-.3,.8) circle (6pt);
        \node at (-0.3,.8) {$\bar{Z}^{p}$};
        \filldraw[draw, circle, fill=myRed, minimum size=6pt](.3,.8) circle (6pt);
        \node at (.3,.8) {$\bar{X}^{q}$};
        \filldraw[draw, circle, fill=myBlue, minimum size=6pt](-.3,0) circle (6pt);
        \node at (-0.3,0) {$Z^{p}$};
        \filldraw[draw, circle, fill=myRed, minimum size=6pt](.3,0) circle (6pt);
        \node at (.3,0) {$X^{q}$};
    \end{tikzpicture}\ .
\eeq
To understand this identity, we note that the LHS is a Bell pair projector $\ket{\text{Bell}}\bra{\text{Bell}}$. This can be expressed as a projector $\frac{1}{D} \sum_{p=0}^{D-1}  (Z \otimes \overline{Z})^p  \frac{1}{D} \cdot \sum_{q=0}^{D-1}  (X \otimes \overline{X})^q$, which is the RHS up to a normalization factor.  

A representative MPS tensor can thus be given from the above equation, which has $D^2$-dimensional physical degrees of freedom labeled by the pair $(p,q)$: 
\beq
    A^{(p,q)} = \sqrt{\mu_{p,q}}Z^p X^q.
\eeq
Equivalently, the MPS tensor can be written as
\beq
        A=\quad \begin{tikzpicture}[baseline=(current bounding box),scale=1.5]
        \draw (-.5,0.4) -- (-.1,.4);
        \draw (.1,0.4) -- (1.1,.4);
        \draw (1.3,0.4) -- (3,.4);
        \draw[rounded corners=8pt]
        (0,-.2) -- (0,.8) -- (1,.8) -- (1, 1.3);
        \draw[rounded corners=8pt]
        (.3,-.2) -- (0.3, 0) -- (1.2,0) -- (1.2,.8) -- (2.2,.8) -- (2.2, 1.3);
        \filldraw[fill = myBlue] (.3,.2) rectangle node {} (.8,1);
        \filldraw[fill = myRed] (1.5,.2) rectangle node {} (2,1);
        \node at (0.15,-.4) {$\ket{\psi_{\mu}}$};
        \node at (.55,.4) {$Z$};
        \node at (.55,.8) {$C$};
        \node at (1.75,.4) {$X$};
        \node at (1.75,.8) {$C$};
        \end{tikzpicture}
        \label{eq:A_RZRX}
\eeq    
where $\ket{\psi_{\mu}}= \sum_{p,q=0}^{D-1}\sqrt{\mu_{p,q}}\ket{p}\otimes \ket{q}$. This is because the $CZ$ and $CX$ gates acting on state $\ket{p}\otimes\ket{q}$ creates $Z^p X^q$ matrix on the virtual bond, while the coefficient $\sqrt{\mu_{p,q}}$ is provided by state $\ket{\psi_{\mu}}$. Therefore, the target state $\ket{A}$ is prepared from a product state using a sequential circuit composed of $CZ$ and $CX$ gates. Furthermore, as we will discuss in Lemma \ref{lemma:KW_KT}, $\ket{A}$ is prepared via implementing the Kennedy-Tasaki duality.

\begin{lemma}
    An MPS $\ket{A}$ with pushing relation in class $RIRX$ ($RZRX$) can be prepared from a product state using the generalized KW (KT) duality followed by a transversal unitary circuit.
    \label{lemma:KW_KT}
\end{lemma}

As a brief remark, the conventional KW and KT dualities are defined for qubits. Here we generalize them to qudits, where they implement operator mappings with a structure similar to the qubit case. Lemma \ref{lemma:KW_KT} can be proved from their MPO structures.
\begin{proof}
     According to Eq.~\eqref{eq:RIRX} in Sec.~\ref{subsec:RIRX}, states in class $RIRX$ are prepared from some product state using the following MPO:
    \beq
        &\begin{tikzpicture}[baseline=(current bounding box),scale=1.5]
        \node at (-.3,.8) {...};
        \draw[rounded corners=8pt]
        (-.1,.8) -- (1,.8) -- (1,1.3);
        \draw[rounded corners=8pt]
        (0,-.2) -- (0,.4) -- (1.1,.4) -- (1.1,.8) -- (1.3,.8);
        \draw (1.4, -.2) -- (1.4, 1.3);
        \draw[rounded corners=8pt]
        (1.5,.8) -- (2.6,.8) -- (2.6,1.3);
        \draw[rounded corners=8pt]
        (0+1.6,-.2) -- (0+1.6,.4) -- (1.1+1.6,.4) -- (1.1+1.6,.8) -- (1.3+1.6,.8);
        \draw (1.4+1.6, -.2) -- (1.4+1.6, 1.3);
        \draw[rounded corners=8pt]
        (1.5+1.6,.8) -- (2.6+1.6,.8) -- (2.6+1.6,1.3);
        \draw[rounded corners=8pt]
        (0+3.2,-.2) -- (0+3.2,.4) -- (1.1+3.2,.4) -- (1.1+3.2,.8) -- (1.3+3.2,.8);
        \draw (1.3+3.3, -.2) -- (1.3+3.3, 1.3);
        \draw (1.3+3.4, .8) -- (1.3+3.6, .8);
        \node at (1.3+3.8,.8) {$...$};
        \filldraw[fill = myRed] (.3,.2) rectangle node {} (.8,1);
        \node at (.55,.8) {$C$};
        \node at (.55,.4) {$X$};
        \filldraw[fill = myRed] (1.9,.2) rectangle node {} (2.4,1);
        \node at (2.15,.8) {$C$};
        \node at (2.15,.4) {$X$};
        \filldraw[fill = myRed] (1.9+1.6,.2) rectangle node {} (2.4+1.6,1);
        \node at (3.75,.8) {$C$};
        \node at (3.75,.4) {$X$};
        \end{tikzpicture}\\[2ex]
        &=
        \begin{tikzpicture}[baseline=(current bounding box),scale=1.5]
        \node at (-.3,.8) {$...$};
        \draw[rounded corners=8pt]
        (-.1,.8) -- (1.2,.8) -- (1.2,1.3);
        \draw[rounded corners=8pt]
        (.2,-.5) -- (.2,.4) -- (1.7,.4) -- (1.7,.8) -- (1.9,.8);
        \draw (2, -.5) -- (2, 1.3);
        \draw[rounded corners=8pt]
        (2.1,.8) -- (3.4,.8) -- (3.4,1.3);
        \draw[rounded corners=8pt]
        (2.4,-.5) -- (2.4,.4) -- (4,.4) -- (4,.8) -- (4.2,.8);
        \draw (4.4,.8) -- (4.6,.8);
        \draw (4.3, -.5) -- (4.3, 1.3);
        \node at (4.8,.8) {$...$};
        \filldraw[fill = myBlue] (.4,.2) rectangle node {} (.9,1);
        \node at (.65,.8) {$C$};
        \node at (.65,.4) {$Z$};
        \filldraw[fill = myBlue] (2.6,.2) rectangle node {} (3.1,1);
        \node at (2.85,.8) {$C$};
        \node at (2.85,.4) {$Z$};
        \filldraw[draw, circle, fill=myGray, minimum size=6pt](1.3,.4) circle (6pt);
        \node at (1.3,.4) {$H^{\dagger}$};
        \filldraw[draw, circle, fill=myGray, minimum size=6pt](3.5,.4) circle (6pt);
        \node at (3.5,.4) {$H^{\dagger}$};
        \filldraw[draw, circle, fill=myGray, minimum size=6pt](.2,-.1) circle (6pt);
        \node at (.2,-.1) {$H$};
        \filldraw[draw, circle, fill=myGray, minimum size=6pt](2.4,-.1) circle (6pt);
        \node at (2.4,-.1) {$H$};
        \end{tikzpicture}
    \eeq
     We thus conclude that this MPO is a transversal Hadamard circuit followed by a generalized Kramers-Wannier duality (KW) generated by the tensor
     \beq
        \begin{tikzpicture}[baseline=(current bounding box),scale=1.5]
        \draw[rounded corners=8pt]
        (-.4,.8) -- (1.1,.8) -- (1.1,1.3);
        \draw[rounded corners=8pt]
        (0,-.1) -- (0,.4) -- (1.7,.4);
        \filldraw[fill = myBlue] (.3,.2) rectangle node {} (.8,1);
        \node at (.55,.8) {$C$};
        \node at (.55,.4) {$Z$};
        \filldraw[draw, circle, fill=myGray, minimum size=6pt](1.2,.4) circle (6pt);
        \node at (1.2,.4) {$H^{\dagger}$};
        \end{tikzpicture}\quad .
        \label{eq:KW}
     \eeq
     The operator maps implemented by KW are
     \beq
        \text{KW} X_s = Z_s Z_{s+1}^{\dagger} \text{KW},\quad  \text{KW} Z_s Z_{s+1}^{\dagger} = X_{s+1} \text{KW}.
     \eeq
     
     According to Eq.~\eqref{eq:A_RZRX}, a state in class $RZRX$ is prepared using an MPO generated by the following tensor:
     \beq
        \begin{tikzpicture}[baseline=(current bounding box),scale=1.5]
        \draw (-.4,0.8) -- (1,.8);
        \draw (1.2,0.8) -- (2.7,.8);
        \draw (2.9,0.8) -- (3.8,.8);
        \draw[rounded corners=8pt]
        (0,-.5) -- (0,.4) -- (1.1,.4) -- (1.1, 1.7);
        \draw[rounded corners=8pt]
        (1.7,-.5) -- (1.7,.4) -- (2.8,.4) -- (2.8, 1.7);
        \filldraw[fill = myRed] (.3,.2) rectangle node {} (.8,1);
        \node at (.55,.8) {$C$};
        \node at (.55,.4) {$X$};
        \filldraw[fill = myRed] (2,.2) rectangle node {} (2.5,1);
        \node at (2.25,.8) {$C$};
        \node at (2.25,.4) {$X$};
        \filldraw[draw, circle, fill=myGray, minimum size=6pt](1.6,.8) circle (6pt);
        \node at (1.6,.8) {$H$};
        \filldraw[draw, circle, fill=myGray, minimum size=6pt](3.3,.8) circle (6pt);
        \node at (3.3,.8) {$H^{\dagger}$};
        \filldraw[draw, circle, fill=myGray, minimum size=6pt](1.7,-.1) circle (6pt);
        \node at (1.7,-.1) {$H$};
        \filldraw[draw, circle, fill=myGray, minimum size=6pt](0,-.1) circle (6pt);
        \node at (0,-.1) {$H$};
        \filldraw[draw, circle, fill=myGray, minimum size=6pt](1.1,1.3) circle (6pt);
        \node at (1.1,1.3) {$H^{\dagger}$};
        \filldraw[draw, circle, fill=myGray, minimum size=6pt](2.8,1.3) circle (6pt);
        \node at (2.8,1.3) {$H^{\dagger}$};
        \label{eq:KT}
        \end{tikzpicture}
     \eeq
Above we have used that the control and target qudits are exchanged when a $CX$ gate is conjugated by Hadamard gates acting on both the target and the control qudits. The MPO between the transversal $H$ and the transversal $H^{\dagger}$ circuits is the generalized KT duality, which implements the following operator maps
     \beq
        \text{KT} X_s &= X_s  \text{KT},\quad  \text{KT}  Z_s Z_{s+1}^{\dagger} = Z_s \tilde{X}_{s}Z_{s+1}^{\dagger} \text{KT},\\
        \text{KT} \tilde{X}_{s} &= \tilde{X}_{s} \text{KT},\quad  \text{KT} \tilde{Z}_{s} \tilde{Z}_{s+1}^{\dagger} = \tilde{Z}_{s}X_{s+1}^{\dagger} \tilde{Z}_{s+1}^{\dagger} \text{KT},
     \eeq
     where we use $O_s$ and $\tilde{O}_s$ to denote operator $O$ acting on the first and second qudit at site $s$ respectively.
\end{proof}

To conclude, we have shown that a general Pauli right-pushing relation in a $D$-dimensional virtual space can be decomposed into three classes ($RIRI$, $RIRX,$ and $RZRX$) in three subspaces, after site blocking and appropriate gauge transformation. We further showed that MPSs with pushing relations in these classes are prepared from product states using $I$, KW, and KT respectively, up to a transversal circuit.

As both KW and KT dualities are sequential Clifford circuits, they can be implemented using Clifford teleportation. In the conventional Clifford teleportation, the circuits are implemented by adding extensive Bell pairs followed by Bell-basis measurements. On the other hand, Ref.~\cite{tantivasadakarn2024long} provides a protocol in which one adds ancilla qubits in the $\ket{+}$ state, and performs single-site $X$ measurements. In Appendix \ref{app:implement}, we present a diagrammatic implementation of KW and KT using the properties of the MPO tensors, where the KW implementation is similar to that of Ref.~\cite{tantivasadakarn2024long}.

\subsection{Left/right-pushable Pauli defects}
\label{subsec:left_right_pushable}

Now we consider the cases where the generalized Pauli defects in the $D$-dimensional virtual space are either left- or right-pushable to other Pauli defects. For a general bond dimension $D$, it is proven in Appendix \ref{app:coarse_grain} that we can decompose the virtual space into six subspaces such that after site blocking and appropriate gauge transformation, the pushing relations in the individual subspaces are: $Z,X\mapsto I$ and its spatial inversion; $Z\mapsto I$, $X\mapsto X$ and its spatial inversion; $Z\mapsto Z$, $X\mapsto X$; and $Z\mapsto I$, $I\mapsfrom X$.

To demonstrate this, we consider bond dimension $D=2$. In this case, there are only three Pauli operators, and we need at least two generators to be either left- or right-pushable. We can always make a gauge transformation such that the pushable defects are $Z$ and $X$. We have discussed the case where both $Z$ and $X$ are right-pushable. When both $Z$ and $X$ are left-pushable, the states simply differ by a spatial inversion. Therefore, the only intrinsic class of states corresponds to $Z\mapsto P_Z$ and $P_X\mapsfrom X$. We note that $X \mapsfrom X$ also suggests $X\mapsto X$ as defined in Eq.~\eqref{eq:pushing_relation}. For example, the pushing relation $Z \mapsto Z, X \mapsfrom X$ is equivalent to $RZRX$ discussed in Sec.~\ref{subsec:RZRX}. The only class of pushing relation that intrinsically contains both left and right pushing after blocking is therefore $Z\mapsto I$ and $I\mapsfrom X$.

A special pushing relation that we rule out is $Z\mapsto X$, $Z\mapsfrom X$. In this case, every defect is essentially localized and cannot be pushed two sites away. According to Definition \ref{def:pushable}, neither $Z$ nor $X$ defect is pushable. The corresponding state $\ket{A}$ is therefore not preparable using fusion measurement. We will discuss the structure of these states later in Sec.~\ref{subsec:non-pushable}.

\subsubsection*{class $RILI$}
\label{subsec:RILI}

We now discuss the only class of pushing relations that intrinsically contain both left and right pushing after blocking. We use $RILI$ to denote that $Z$ is right-pushed to $I$ and $X$ is left-pushed to $I$: $Z\mapsto I, I\mapsfrom X$. The transfer matrix that supports this pushing relation satisfies:
\beq
    \big(\bar{Z}\otimes Z\big) T_A = T_A = T_A \big(\bar{X}\otimes X\big).
\eeq
This constrains $\lambda_{a,b,c,d}$ in Eq.~\eqref{eq:transfer_matrix_bell} to be non-zero only when $b=c=0$. The general form of the transfer matrix is thus:
\beq
    T_A&=\sum_{a,d=0}^{d-1}\lambda_{a,d}\quad \begin{tikzpicture}[baseline=(current bounding box),scale=1.5]
        \draw[rounded corners=8pt]
        (-.8,.8) -- (0,.8) -- (0,0) -- (-.8,0);
        \draw[rounded corners=8pt]
        (1,.8) -- (.2,.8) -- (.2,0) -- (1,0);
        \filldraw[draw, circle, fill=myRed, minimum size=6pt](0.6,0) circle (6pt);
        \node at (0.6,0) {$X^{d}$};
        \filldraw[draw, circle, fill=myBlue, minimum size=6pt](-0.4,0.8) circle (6pt);
        \node at (-0.4,0.8) {$Z^a$};
        \end{tikzpicture}\ .
\eeq
In Appendix \ref{app:RILI}, we will show that a representative state for the above transfer matrix is obtained from a product state using transversal controlled-$Z$ gates. Here, we will instead perform a blocking of two sites. The transfer matrix is squared after site blocking:
\beq
    T_{A'}=T_A^2=\begin{tikzpicture}[baseline=(current bounding box),scale=1.5]
        \draw[rounded corners=8pt]
        (-.8,.8) -- (-.3,.8) -- (-.3,0) -- (-.8,0);
        \draw[rounded corners=8pt]
        (.8,.8) -- (.3,.8) -- (.3,0) -- (.8,0);
        \filldraw[draw, circle, fill=white, minimum size=6pt](-.3,0.4) circle (6pt);
        \node at (-0.3,0.4) {$R$};
        \filldraw[draw, circle, fill=white, minimum size=6pt](.3,0.4) circle (6pt);
        \node at (0.3,0.4) {$L$};
        \end{tikzpicture}\ ,
\eeq
where $R = \sum_a \lambda_{a,0}Z^a$ and $L = \sum_d \lambda_{0,d}X^d$. The blocked transfer matrix is thus of rank 1 and according to Lemma \ref{lemma:FDLU}, the state $\ket{A}$ is FDLU-preparable. A representative state for the above transfer matrix after site blocking is derived in Sec.~\ref{sec:FDLU}.

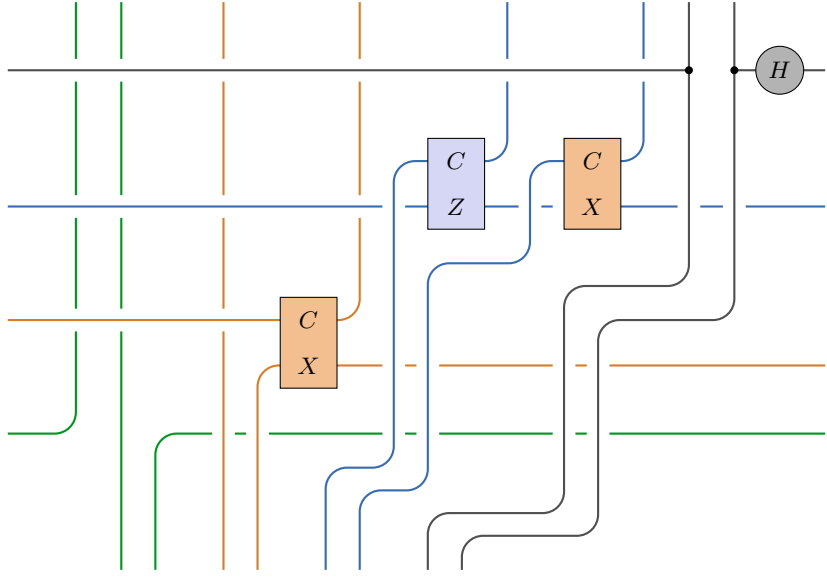
\begin{figure*}[t]
    \centering
    \begin{tikzpicture}[baseline=(current bounding box),scale=1.5]
    \tikzset{
    lineA/.style={draw=Lteal, thick},
    lineB/.style={draw=Lorange, thick},
    lineC/.style={draw=Lblue, thick},
    lineD/.style={draw=Lgray, thick},
    }

        \draw[lineA, rounded corners=8pt]
        (-.2,1.1) -- (-.2,.2) -- (-.8,.2);
        \draw[lineA] (-.2,1.3) -- (-.2,2.1);
        \draw[lineA] (-.2,2.3) -- (-.2,3.3);
        \draw[lineA] (-.2,3.5) -- (-.2,4);
        \draw[lineA] (.2,1.1) -- (.2,-1);
        \draw[lineA] (.2,1.3) -- (.2,2.1);
        \draw[lineA] (.2,2.3) -- (.2,3.3);
        \draw[lineA] (.2,3.5) -- (.2,4);
        \draw[lineA, rounded corners=8pt]
        (.5,-1) -- (.5,.2) -- (1,.2);
        \draw[lineA] (1.2,0.2) -- (1.3,0.2);
        \draw[lineA] (1.5,0.2) -- (2.5,0.2);
        \draw[lineA] (2.7,0.2) -- (2.8,0.2);
        \draw[lineA] (3,0.2) -- (4,0.2);
        \draw[lineA] (4.2,0.2) -- (4.3,0.2);
        \draw[lineA] (4.5,0.2) -- (6.4,0.2);
        
        \draw[lineB, rounded corners=8pt]
        (-.8,1.2) -- (2.3,1.2) -- (2.3,2.1);
        \draw[lineB, rounded corners=8pt]
        (1.4,-1) -- (1.4,.8) -- (2.5,.8);
        \draw[lineB] (2.3,2.3) -- (2.3,3.3);
        \draw[lineB] (2.3,3.5) -- (2.3,4);
        \draw[lineB] (2.7,0.8) -- (2.8,0.8);
        \draw[lineB] (3,0.8) -- (4,0.8);
        \draw[lineB] (4.2,0.8) -- (4.3,0.8);
        \draw[lineB] (4.5,0.8) -- (6.4,0.8);
        \filldraw[fill = myRed] (.3+1.3,.6) rectangle node {} (.8+1.3,1.4);
        \node at (1.85,.8) {$X$};
        \node at (1.85,1.2) {$C$};
        \draw[lineB] (1.1, -1) -- (1.1, 1.1);
        \draw[lineB] (1.1, 1.3) -- (1.1, 2.1);
        \draw[lineB] (1.1,2.3) -- (1.1,3.3);
        \draw[lineB] (1.1,3.5) -- (1.1,4);

        \draw[lineC] (-.8,2.2) -- (2.5,2.2);
        \draw[lineC] (2.7,2.4-.2) -- (3.7,2.4-.2);
        \draw[lineC] (3.9,2.4-.2) -- (4,2.4-.2);
        \draw[lineC] (4.6,2.4-.2) -- (5.1,2.4-.2);
        \draw[lineC] (5.3,2.4-.2) -- (5.5,2.4-.2);
        \draw[lineC] (5.7,2.4-.2) -- (6.4,2.4-.2);
        \draw[lineC, rounded corners=8pt]
        (2,-1) -- (2,-.1) -- (2.6,-.1) -- (2.6,2.8-.2) -- (3.6,2.8-.2) -- (3.6, 3.3);
        \draw[lineC, rounded corners=8pt]
        (2.3,-1) -- (2.3,-.3) -- (2.9,-.3) -- (2.9, 2-.3) -- (3.8,2-.3) -- (3.8,2.8-.2) -- (4.8,2.8-.2) -- (4.8, 3.3);
        \draw[lineC] (3.6,3.5) -- (3.6,4);
        \draw[lineC] (4.8,3.5) -- (4.8,4);
        \filldraw[fill = myBlue] (2.9,2) rectangle node {} (3.4,2.8);
        \node at (3.15,2.2) {$Z$};
        \node at (3.15,2.6) {$C$};
        \filldraw[fill = myRed] (4.1,2) rectangle node {} (4.6,2.8);
        \node at (4.35,2.2) {$X$};
        \node at (4.35,2.6) {$C$};

        \draw[lineD] (-.8,3.4) -- (5.2,3.4);
        \draw[lineD] (5.6,3.4) -- (6.4,3.4);
        \draw[lineD, rounded corners=8pt] 
        (2.9, -1) -- (2.9,-.5) -- (4.1,-.5) -- (4.1,1.5) -- (5.2,1.5) -- (5.2,4);
        \draw[lineD, rounded corners=8pt] 
        (3.2, -1) -- (3.2,-.7) -- (4.4,-.7) -- (4.4,1.2) -- (5.6,1.2) -- (5.6,4);
        \fill (5.6,3.4) circle (1pt);
        \fill (5.2,3.4) circle (1pt);
        \filldraw[draw, circle, fill=myGray, minimum size=6pt](0.6+5.4,0+3.4) circle (6pt);
        \node at (0.6+5.4,0+3.4) {$H$};
        \end{tikzpicture}
    \caption{When the virtual space is decomposed into four subspaces, respectively exhibiting pushing relations in class $RIRI$ (green), $RIRX$ (orange), $RZRX$ (blue), and $RILI$ (black), the MPS $\ket{A}$ is prepared from a product of multi-qudit states $\ket{\psi_{\Lambda}}$ using a composition of parallel KW and KT dualities shown in this figure, up to transversal unitary circuit. The state structure for the $RILI$ subspace is derived in Appendix \ref{app:RILI}.}
    \label{fig:parallel_MPO}
\end{figure*}

The fact that the target MPS $\ket{A}$ is FDLU-preparable is expected. Similar to the RIRI class in Sec.~\ref{subsec:RIRI}, since all the defects can be locally annihilated, $\ket{A}$ can be obtained from the associated FDLU-preparable state $\ket{B}$ in Eq.~\eqref{eq:B_fusion_measurement} via a controlled-unitary circuit between the ancillary bonds and physical bonds.

\subsection{Summary}
Having discussed the MPS structure for classes $RIRI$, $RIRX$, $RZRX$, and $RILI$, we summarize the results in the following lemma. 
\begin{lemma}
    If the generalized Pauli defects in the virtual space of an MPS $\ket{A}$ are left- or right-pushable to generalized Pauli defects, then up to a transversal unitary circuit
    \beq
        \ket{A} = D \ket{\psi}^{\otimes N}
    \eeq
    where $D$ is a composition of KT and KW on different sublattices. 
    \label{lemma:all_pauli}
\end{lemma}

In Appendix \ref{app:coarse_grain} we show that a $D$-dimensional virtual space with general pushing relations can be decomposed into subspaces with these four classes, or their spatial inversion, after site blocking and appropriate gauge transformation. Therefore, we conclude that when all Pauli defects in the virtual space are left or right-pushable, the state is always preparable from a product state using the MPO shown in Fig.~\ref{fig:parallel_MPO}. We demonstrate the MPS tensor for general pushing relations. The subspaces that exhibit $RIRI$, $RIRX$, $RZRX$, and $RILI$ are denoted in green, yellow, blue, and black, respectively.

\subsection{Generalization to decomposed subspaces}

\label{subsec:general}

Now we consider the cases where we decompose the virtual space into subspaces of dimensions $D_1,D_2,\cdots$, and Pauli defects in the subspaces are pushable. A general MPS transfer matrix can be expanded into the Bell basis of each subspace:
\beq
    T_A=&\sum_{\{a_i,b_i,c_i,d_i\}=0}^{D_i-1}\Lambda(\{a_i,b_i,c_i,d_i\})\quad
    \begin{tikzpicture}[baseline=(current bounding box),scale=1.5]
        \draw[rounded corners=8pt]
        (-.8,.8) -- (0,.8) -- (0,0) -- (-.8,0);
        \draw[rounded corners=8pt]
        (1,.8) -- (.2,.8) -- (.2,0) -- (1,0);
        \filldraw[draw, circle, fill=myRed, minimum size=6pt](0.6,0) circle (6pt);
        \node at (0.6,0) {$X^{d_1}$};
        \filldraw[draw, circle, fill=myBlue, minimum size=6pt](0.6,.8) circle (6pt);
        \node at (0.6,.8) {$Z^{c_1}$};
        \filldraw[draw, circle, fill=myRed, minimum size=6pt](-0.4,0) circle (6pt);
        \node at (-0.4,0) {$X^{b_1}$};
        \filldraw[draw, circle, fill=myBlue, minimum size=6pt](-0.4,0.8) circle (6pt);
        \node at (-0.4,0.8) {$Z^{a_1}$};
        \end{tikzpicture}\\[2ex]
        &\quad \bigotimes\
        \begin{tikzpicture}[baseline=(current bounding box),scale=1.5]
        \draw[rounded corners=8pt]
        (-.8,.8) -- (0,.8) -- (0,0) -- (-.8,0);
        \draw[rounded corners=8pt]
        (1,.8) -- (.2,.8) -- (.2,0) -- (1,0);
        \filldraw[draw, circle, fill=myRed, minimum size=6pt](0.6,0) circle (6pt);
        \node at (0.6,0) {$X^{d_2}$};
        \filldraw[draw, circle, fill=myBlue, minimum size=6pt](0.6,.8) circle (6pt);
        \node at (0.6,.8) {$Z^{c_2}$};
        \filldraw[draw, circle, fill=myRed, minimum size=6pt](-0.4,0) circle (6pt);
        \node at (-0.4,0) {$X^{b_2}$};
        \filldraw[draw, circle, fill=myBlue, minimum size=6pt](-0.4,0.8) circle (6pt);
        \node at (-0.4,0.8) {$Z^{a_2}$};
        \end{tikzpicture}\ \bigotimes\ \cdots.
        \label{eq:general_T}
\eeq
The states are still fusion measurement preparable since all Pauli operators in subspaces form a complete orthonormal basis of operators in the whole virtual space.

When the dimensions of the subspaces are coprime, the Pauli operators in the subspaces are also Pauli operators in the whole virtual space. As a result, Lemma \ref{lemma:all_pauli} applies, and the state is preparable from some product state using a composition of parallel KW and KT in each sublattice as shown in Fig.~\ref{fig:parallel_MPO}. For a general decomposition where the dimensions of the subspaces are not coprime, as long as the pushing relations are contained in each subspace after site blocking and gauge transformation, following a similar argument, we can show that the state is still preparable from a product state by a composition of parallel actions of KW and KT on different sublattices. 

A composition of these parallel KW and KT is a duality for some Abelian symmetry $A = \mathbb{Z}_{n_1}\times \mathbb{Z}_{n_2}\times\cdots$. These dualities flow to non-invertible symmetries described by the Tambara-Yamagami (TY) fusion categories in the continuum~\cite{tambara1998tensor,thorngren2024fusion}. A TY category contains an abelian symmetry $A$, along with a non-invertible duality $D$ of the abelian group. The multiplication rules of these dualities is given by
\beq
    D^2 = \sum_{g\in A} g,
\eeq
up to a potential mixing with translation on the lattice. Given an abelian group $A$, the TY categories are characterized by a symmetric, non-degenerate bicharacter $\chi$, and the Frobenius-Schur (FS) indicator $\epsilon = \pm 1$. 

For example, for $G = \mathbb{Z}_2$ symmetry generated by $\eta = \prod_{i}X_i$, the KW duality realizes such a TY category with a mixing with lattice translation $\text{KW}^2 = T(1+\eta)$, with $T$ being a lattice translation by one site. This $\text{TY}(\mathbb{Z}_2)$ category corresponds to a diagonal bicharacter and FS indicator $\epsilon = 1$. For $G = \mathbb{Z}_2^2$ symmetry generated by $\eta_e=\prod_i X_{2i}$ and $\eta_o=\prod_i X_{2i+1}$, the composition of KW dualities for both $\mathbb{Z}_2$ symmetries realizes a $\text{TY}(\mathbb{Z}_2^2)$ with a diagonal bicharacter, whereas the KT duality realizes another $\text{TY}(\mathbb{Z}_2^2)$, with an off-diagonal bicharacter. Now that we have arbitrary compositions of generalized KW and KT for each subsystem, the dualities give rise to arbitrary bicharacter.

\begin{lemma}
    In the decomposition of the virtual space into subspaces of dimensions $D_1,D_2,\cdots$, if the generalized Pauli defects in the subspaces of an MPS $\ket{A}$ are left or right-pushable to Pauli defects in the same subspace, $\ket{A}$ is a product state or it can be prepared from a product state using dualities that are described by Tambara-Yamagami categorical symmetries, up to a transversal unitary circuit.
    \label{lemma:Pauli_subspace}
\end{lemma}

We note that it is shown in Ref.~\cite{zhang2024characterizing} that all MPSs allowing any Pauli defects to be right-pushable to Pauli defects can be prepared using a sequential Clifford circuit from a product state up to another transversal unitary circuit, and vice versa. Lemma \ref{lemma:all_pauli} and \ref{lemma:Pauli_subspace} consider both left and right pushing relations, but are restricted to cases when the pushing relations are contained in each subspace. In particular, when all non-trivial Pauli defects are right pushed to non-trivial Pauli defects, the pushing relation corresponds to a permutation of the Abelianized group of defects. Hence, after blocking a finite number of sites, the pushing relation becomes $P\rightarrow P$ for any Pauli defect $P$ and is contained in each subspace. 

More generally, there could be pushing relations that are not contained in each subspace after gauge transformation and site blocking. In what follows, we discuss two examples where the virtual space is decomposed into two $D=2$ subspaces and demonstrate their state structures. 

The first example is given by the pushing relations
\beq
    X_1\mapsto X_1,\quad X_2\mapsto X_1, \quad Z_1\mapsto Z_1,\quad Z_2\mapsto Z_1.
\eeq
They cannot be brought to those within each $D=2$ subspace by gauge transformation and site blocking. The corresponding transfer matrix is of the following form:
\beq
    T_A=&\sum_{a,b,c,d=0}^{D_i-1}\lambda_{a,b,c,d}\quad
    \begin{tikzpicture}[baseline=(current bounding box),scale=1.5]
        \draw[rounded corners=8pt]
        (-.8,.8) -- (0,.8) -- (0,0) -- (-.8,0);
        \draw[rounded corners=8pt]
        (1,.8) -- (.2,.8) -- (.2,0) -- (1,0);
        \filldraw[draw, circle, fill=myRed, minimum size=6pt](0.6,0) circle (6pt);
        \node at (0.6,0) {$X^{b}$};
        \filldraw[draw, circle, fill=myBlue, minimum size=6pt](0.6,.8) circle (6pt);
        \node at (0.6,.8) {$Z^{a}$};
        \filldraw[draw, circle, fill=myRed, minimum size=6pt](-0.4,0) circle (6pt);
        \node at (-0.4,0) {$X^{b}$};
        \filldraw[draw, circle, fill=myBlue, minimum size=6pt](-0.4,0.8) circle (6pt);
        \node at (-0.4,0.8) {$Z^{a}$};
        \end{tikzpicture}\\[2ex]
        &\quad \bigotimes\
        \begin{tikzpicture}[baseline=(current bounding box),scale=1.5]
        \draw[rounded corners=8pt]
        (-.8,.8) -- (0,.8) -- (0,0) -- (-.8,0);
        \draw[rounded corners=8pt]
        (1,.8) -- (.2,.8) -- (.2,0) -- (1,0);
        \filldraw[draw, circle, fill=myRed, minimum size=6pt](0.6,0) circle (6pt);
        \node at (0.6,0) {$X^{d}$};
        \filldraw[draw, circle, fill=myBlue, minimum size=6pt](0.6,.8) circle (6pt);
        \node at (0.6,.8) {$Z^{c}$};
        \filldraw[draw, circle, fill=myRed, minimum size=6pt](-0.4,0) circle (6pt);
        \node at (-0.4,0) {$X^{b}$};
        \filldraw[draw, circle, fill=myBlue, minimum size=6pt](-0.4,0.8) circle (6pt);
        \node at (-0.4,0.8) {$Z^{a}$};
        \end{tikzpicture}\ .
\eeq
Nevertheless, one can show that the states in this class are still preparable from a product state using the KT duality. 

As for the second example, we consider the following pushing relations
\beq
    X_1\mapsto X_1,\quad X_2\mapsto X_2 Z_1, \quad Z_1\mapsto I,\quad Z_2\mapsto I.
\eeq

The corresponding transfer matrix is given by:
\beq
    T_A=&\sum_{a,b,b',c,d=0}^{D_i-1}\delta_{b'=b+d}\lambda_{a,b,c,d}\quad
    \begin{tikzpicture}[baseline=(current bounding box),scale=1.5]
        \draw[rounded corners=8pt]
        (-.8,.8) -- (0,.8) -- (0,0) -- (-.8,0);
        \draw[rounded corners=8pt]
        (1,.8) -- (.2,.8) -- (.2,0) -- (1,0);
        \filldraw[draw, circle, fill=myRed, minimum size=6pt](0.6,0) circle (6pt);
        \node at (0.6,0) {$X^{c}$};
        \filldraw[draw, circle, fill=myBlue, minimum size=6pt](0.6,.8) circle (6pt);
        \node at (0.6,.8) {$Z^{a}$};
        \filldraw[draw, circle, fill=myBlue, minimum size=6pt](-0.4,0.8) circle (6pt);
        \node at (-0.4,0.8) {$Z^{a}$};
        \end{tikzpicture}\\[2ex]
        &\quad \bigotimes\
        \begin{tikzpicture}[baseline=(current bounding box),scale=1.5]
        \draw[rounded corners=8pt]
        (-.8,.8) -- (0,.8) -- (0,0) -- (-.8,0);
        \draw[rounded corners=8pt]
        (1,.8) -- (.2,.8) -- (.2,0) -- (1,0);
        \filldraw[draw, circle, fill=myRed, minimum size=6pt](0.6,0) circle (6pt);
        \node at (0.6,0) {$X^{d}$};
        \filldraw[draw, circle, fill=myBlue, minimum size=6pt](0.6,.8) circle (6pt);
        \node at (0.6,.8) {$Z^{b'}$};
        \filldraw[draw, circle, fill=myBlue, minimum size=6pt](-0.4,0.8) circle (6pt);
        \node at (-0.4,0.8) {$Z^{b}$};
        \end{tikzpicture}\ .
\eeq
It can be shown that the states in this class can be prepared from a product state using the following MPO:
\beq
        \begin{tikzpicture}[baseline=(current bounding box),scale=1.5]
        \draw[rounded corners=8pt]
        (-.5,.1) -- (.2,.1) -- (.2,2.4);
        \draw[rounded corners=8pt]
        (.6,-.6) -- (.6,2.1) -- (2.3,2.1);
        \draw[rounded corners=8pt]
        (.7,-.3) -- (1.2,-.3) -- (1.2,2);
        \draw (-.5,-.3) -- (.5,-.3);
        \draw (1.2,2.2) -- (1.2,2.4);
        \draw[rounded corners=8pt]
        (1.6,-.6) -- (1.6,1.5) -- (2.3,1.5);
        \filldraw[fill = myRed] (0,.5) rectangle node {} (.8,1);
        \node at (.2,.75) {$C$};
        \node at (.6,.75) {$X$};
        \filldraw[fill = myRed] (1,.5) rectangle node {} (1.8,1);
        \node at (1.2,.75) {$C$};
        \node at (1.6,.75) {$X$};
        \filldraw[fill = myBlue] (.4,1.2) rectangle node {} (1.4,1.7);
        \node at (.6,1.45) {$C$};
        \node at (1.2,1.45) {$Z$};
        \end{tikzpicture}
     \eeq
As a result, the states in this class are preparable from a product state using two KW dualities followed by {\it non-transversal} entangling $CZ$ gates:
\beq
    \ket{\Psi} = \Big(\prod_i CZ_{2i,2i+1}\Big) \cdot \text{KW}_{\text{even}} \cdot \text{KW}_{\text{odd}}\ket{\psi}^{\otimes N}.
\eeq

To summarize, we have presented two classes of states that lie beyond the criterion of Lemma \ref{lemma:Pauli_subspace}, but are nonetheless related to some Tambara-Yamagami dualities followed by a possible non-transversal circuit. It is an open question whether states with a general class of pushing relations are always related to Tambara-Yamagami dualities up to certain non-transversal gates.

\section{Only Pauli $Z$ defects are pushable}\label{sec:pauli_z}
In the previous section, we have assumed that for the target state $\ket{A}$, all defects of the form $X^{n} Z^{m}$ for integers $n,m$ are pushable. Namely, those pushable defects form a complete basis of operators acting on the virtual space. This allows us to construct an FDLU-preparable associated state $\ket{B}$ from which $\ket{A}$ can be prepared with one round MF circuit, and the whole preparation circuit is equivalent to the protocol based on fusion measurement.

In this section, we consider the case where the pushable defects of the target state $\ket{A}$ do not form a complete orthonormal basis, and therefore $\ket{A}$ cannot be prepared by one-round fusion measurement in general. Specifically, we will only demand the pushability of Pauli-$Z$ defects. One immediate consequence is that the associated $\ket{B}$ is not necessarily FDLU preparable. To see this, we recall Lemma \ref{lemma:A_B}, which dictates the relation between the transfer matrices between $\ket{A}$ and $\ket{B}$:
\beq
    T_B = T_A\cdot \bigg(\sum_{t} \bar{V}_t\otimes V_t\bigg),
\eeq
where $\{V_t\}$ is the set of pushable defects of $\ket{A}$. When $\{V_t\}$ does not form a complete orthogonal basis, $\sum_{t} \bar{V}_t\otimes V_t$ is not a rank-1 projector into the maximally entangled Bell state. This implies the transfer matrix $T_B$ is not necessarily of rank 1, and therefore $\ket{B}$ is not necessarily FDLU preparable. It follows that there are two simple scenarios. (i) Even though $\sum_{t} \bar{V}_t\otimes V_t$ is not a rank-1 projector, we can still impose certain constraints on $T_A$ so that $T_B$ is of rank 1 and therefore $\ket{B}$ is still FDLU-preparable. (ii) When $T_B$ is not of rank 1, the associated MPS tensor $B$ may still allow all Pauli defects to be pushable, so that $\ket{B}$ can be prepared with the protocol discussed in the previous section. Correspondingly, $\ket{B}$ is  1-round MF-preparable, and $\ket{A}$ is 2-round MF-preparable. This strategy can be extended to the protocols of multiple rounds of MF circuits.

In the following, we will explore classes of states for which all Pauli $Z$ defects are pushable without requiring the pushability of Pauli $X$ defects, and discuss their preparation circuits. It is worth mentioning that a certain class of states can be prepared using generalized cosine symmetries~\cite{aasen2020topological,seifnashri2025gauging,cao2025global}, a class of continuous non-invertible symmetries. This result is the 2-round MF analogue of the last section, i.e., a natural generalization of fusion measurement preparable states prepared using TY dualities.

\subsection{Single and two-round MF preparation}
According to our preparation scheme, we construct an associated MPS tensor $B$ from $A$ using its pushable defects such that $B$ satisfies Eq.~\eqref{eq:B_definition}. Consequently, the associated state $\ket{B}$ is single-round MF circuits away from $\ket{A}$. When all Pauli $Z$ defects are either left- or right-pushable, the associated tensor $B$ is given by:

\begin{equation}
\begin{aligned}
\begin{tikzpicture}[baseline=(current bounding box),scale=1.2]
        \draw (0,.3) -- (.0,.7);
        \draw (-.4,.3) -- (-.4,.7);
        \filldraw[fill = white] (-.7,-.3) rectangle node {$B$} (.3,.3);
        \draw (.7,0) --  (.3,0);
        \draw (-.7,0) --  (-1.1,0);
        \node at (1,0) {$\equiv$};
        \filldraw[fill = white] (1.7,-.3) rectangle node {$A$} (2.3,.3);
        \draw (1.7,0) --  (1.3,0);
        \draw (2,.3) -- (2,.9);
        \draw (2.9,0) --  (2.3,0);
        \fill (2.6,0) circle (1pt);
        \draw (2.6,0.9) --  (2.6,0);
        \filldraw[draw, circle, fill=myGray, minimum size=6pt](2.6,.45) circle (6pt);
        \node at (2.6,.45) {$H$};
        \node at (3.1,0) {,};
\end{tikzpicture}
\label{eq:B_RZ}
\end{aligned}  
\end{equation} 
which satisfies
\begin{equation}
\begin{aligned}
\begin{tikzpicture}[baseline=(current bounding box),scale=1.2]
        \draw (-.4,.3) -- (-.4,.7);
        \draw (0,.3) -- (0,.7);
        \node at (0,.9) {$m$};
        \filldraw[fill = white] (-.7,-.3) rectangle node {$B$} (.3,.3);
        \draw (.7,0) --  (.3,0);
        \draw (-.7,0) --  (-1.1,0);
        \node at (1,0) {$=$};
        \draw (2.4-.7,0) --  (2.4-1.1,0);
        \draw (2.4-.4,.3) -- (2.4-.4,.7);
        \filldraw[fill = white] (2.4-.7,-.3) rectangle node {$A$} (2.4-.1,.3);
        \draw (2.4-.1,0) -- (2.4+.9,0); 
        \filldraw[draw, circle, fill=myBlue, minimum size=6pt](2.4+.4,0) circle (8pt);
        \node at (2.4+.4,0) {$Z^m$};
        \node at (2.4+1.1,0) {.};
\end{tikzpicture}
\label{eq:all_Z_component}
\end{aligned}  
\end{equation} 

The MPS $\ket{B}$ has all $Z$ defects pushable by definition. In addition, when $\ket{B}$ also has all $X$ defects pushable, according to our previous results, after blocking and gauge transformation, $\ket{B}$ is either FDLU- or 1-round MF-preparable. Consequently, the state $\ket{A}$ is either 1- or 2-round MF-preparable. In the following, we expand the discussion for each of these classes. We denote the classes of states as $xx$-$yyyy$. Here, $xx$ is either $RI$ or $RZ$, the only two classes of pushing relations for $A$ after site blocking and spatial inversion: $Z\mapsto I$ and $Z\mapsto Z$. $yyyy$ denotes the pushing relation of the associated MPS tensor $B$, which is one of the classes studied in Sec.~\ref{sec:all_pauli}.

\subsection{Class $RI$}

We use $RI$ to denote that $Z$ is right-pushed to $I$: $Z\mapsto I$. The general transfer matrix $T_A$ is given by
\beq
    T_A=\sum_{a,b,c=0}^{D-1}\lambda_{a,b,c}\quad \begin{tikzpicture}[baseline=(current bounding box),scale=1.5]
        \draw[rounded corners=8pt]
        (-.8,.8) -- (0,.8) -- (0,0) -- (-.8,0);
        \draw[rounded corners=8pt]
        (1,.8) -- (.2,.8) -- (.2,0) -- (1,0);
        \filldraw[draw, circle, fill=myRed, minimum size=6pt](0.6,0) circle (6pt);
        \node at (0.6,0) {$X^{\bar{b}}$};
        \filldraw[draw, circle, fill=myBlue, minimum size=6pt](0.6,.8) circle (6pt);
        \node at (0.6,.8) {$Z^{c}$};
        \filldraw[draw, circle, fill=myBlue, minimum size=6pt](-0.4,0.8) circle (6pt);
        \node at (-0.4,0.8) {$Z^a$};
        \end{tikzpicture}\ .
\eeq
Let us take the virtual bond dimension $D=2$ for simplicity. The transfer matrix with the $RI$ pushing relation is given by
\beq
        T_A &= \quad 
        \begin{tikzpicture}[baseline=(current bounding box),scale=1.5]
        \draw[rounded corners=8pt]
        (-.8,.8) -- (0,.8) -- (0,0) -- (-.8,0);
        \draw[rounded corners=8pt]
        (1,.8) -- (.2,.8) -- (.2,0) -- (1,0);
        \filldraw[fill = white] (-.2,.2) rectangle node {$M_A$} (.4,.6);
        \end{tikzpicture}\ ,
    \label{eq:RI_transfer_matrix}
\eeq
where $M_A$ is a two-qubit Hermitian operator
\beq
    M_A &= \lambda_0 I +\lambda_1 Z_1 +\lambda_2 Z_2 +\lambda_3 Z_1 Z_2  +\lambda_4 X_2 \\
        &\quad\quad + \lambda_5 Z_1 X_2 + i\lambda_6 Z_2 X_2 + i\lambda_7 Z_1 Z_2 X_2,
        \label{eq:T_A_RI}
\eeq
with eight real parameters $\lambda$'s; see also Eq.~\eqref{eq:T_A_RI_app}. We use indices $1$ and $2$ to denote Pauli operators on the left and right bonds, respectively. According to Eq.~\eqref{eq:T_B} with pushable defects $I$ and $Z$, the associated tensor $B$ in Eq.~\eqref{eq:B_RZ} has a transfer matrix given by the Hermitian operator:
\beq
        M_B = \lambda_0 I +\lambda_1 Z_1 +\lambda_2 Z_2 +\lambda_3 Z_1 Z_2.
        \label{eq:T_B_RI}
\eeq

When $Z$ defect is right-pushed to $I$, there exists a unitary operator $U_Z$ whose action on the physical bond annihilates the $Z$ defect. Following a similar derivation as around Eq.~\eqref{eq:A_B_RIRI}, since the target state $\ket{A}$ can be obtained from the associated state $\ket{B}$ in Eq.~\eqref{eq:B_RZ} via measuring the right physical bonds with applying feedback $U_Z$ correction when $t=1$, it can also be obtained from $\ket{B}$ via a controlled-unitary circuit between the left and right physical bonds: 
\begin{equation}
\begin{aligned}
&\begin{tikzpicture}[baseline=(current bounding box),scale=1.2]
        \node at (-1.4,0) {$\cdots$};
        \draw (-.4,.3) -- (-.4,1.2);
        \draw (1,.3) -- (1,1.2);
        \draw (2.4,.3) -- (2.4,1.2);
        \filldraw[fill = white] (-.7,-.3) rectangle node {$A$} (-.1,.3);
        \draw (.7,0) --  (-.1,0);
        \draw (-.7,0) --  (-1.1,0);
        \filldraw[fill = white] (.7,-.3) rectangle node {$A$} (1.3,.3);
        \draw (1.3,0) --  (2.1,0);
        \filldraw[fill = white] (2.1,-.3) rectangle node {$A$} (2.7,.3);
        \draw (2.7,0) --  (3.5,0);
        \node at (3.8,0) {$\cdots$};
        \draw (.1,.8) -- (.1,1.2);
        \node at (.1,.6) {$\ket{+}$};
        \draw (.1+1.4,.8) -- (.1+1.4,1.2);
        \node at (.1+1.4,.6) {$\ket{+}$};
        \draw (.1+2.8,.8) -- (.1+2.8,1.2);
        \node at (.1+2.8,.6) {$\ket{+}$};
\end{tikzpicture}\\[2ex]
=\ &\begin{tikzpicture}[baseline=(current bounding box),scale=1.2]
        \node at (-1.4,0) {$\cdots$};
        \draw (-.5,.3) -- (-.5,.5);
        \draw (-.5,1) -- (-.5,1.2);
        \draw (.1,.3) -- (.1,1.2);
        \filldraw[fill = white] (-.7,-.3) rectangle node {$B$} (.3,.3);
        \draw (.7,0) --  (.3,0);
        \draw (-.7,0) --  (-1.1,0);
        \draw (.9,.3) -- (.9,1.2);
        \draw (1.5,.3) -- (1.5,1.2);
        \filldraw[fill = white] (.7,-.3) rectangle node {$B$} (1.7,.3);
        \draw (1.7,0) --  (2.1,0);
        \filldraw[fill = white] (2.1,-.3) rectangle node {$B$} (3.1,.3);
        \draw (2.3,.3) -- (2.3,1.2);
        \draw (2.9,.3) -- (2.9,.5);
        \draw (2.9,1) -- (2.9,1.2);
        \draw (3.1,0) --  (3.5,0);
        \draw (-.3,.5) --  (-.3,1);
        \draw (-.3,.5) --  (-1,.5);
        \draw (-.3,1) --  (-1,1);
        \node at (-.5,.75) {$U_Z$};
        \filldraw[fill = white] (-.1,.5) rectangle node {} (1.1,1);
        \node at (.1,.75) {$C$};
        \node at (.9,.75) {$U_Z$};
        \filldraw[fill = white] (-.1+1.4,.5) rectangle node {} (1.1+1.4,1);
        \node at (.1+1.4,.75) {$C$};
        \node at (.9+1.4,.75) {$U_Z$};
        \draw (2.7,.5) --  (2.7,1);
        \draw (2.7,.5) --  (3.3,.5);
        \draw (2.7,1) --  (3.3,1);
        \node at (2.9,.75) {$C$};
        \node at (3.8,0) {$\cdots$};
\end{tikzpicture}
\label{eq:RI_A_B}
\end{aligned}  
\end{equation}

The transfer matrix of the associated tensor $B$ only differs from $T_A$ on the right virtual bonds according to Eq.~\eqref{eq:T_B}. Given that $A$ admits $Z\mapsto I$, the $Z$ defect of the $B$ tensor must also be right-pushed to $I$. Therefore, the pushing relations of $A$ and $B$ must be one of the classes $RI$--$RIRI$, $RI$--$RILI$, and $RI$--$RIRX$ after site blocking and gauge transformation. In the following, we discuss the structure of these classes of states separately.

\subsubsection{class $RI$--$\text{FDLU}$}
\label{subsec:RI-FDLU}
We use $RI$--$\text{FDLU}$ to denote the class of states that satisfy $Z\mapsto I$ for $A$, and the associated state $\ket{B}$ is FDLU-preparable. This class contains $RI$--$RIRI$ and $RI$--$RILI$, which correspond to $\lambda_1=\lambda_3=0$ and $\lambda_2=\lambda_3=0$, respectively. As long as $\lambda_0 \lambda_3 = \lambda_1 \lambda_2$, $T_B$ is of rank 1, and state $\ket{B}$ is FDLU-preparable according to Lemma \ref{lemma:FDLU}. Consequently, $\ket{A}$ is also FDLU-preparable using the relation presented in Eq.~\eqref{eq:RI_A_B} between $\ket{A}$ and $\ket{B}$.

Using the representative state for a rank-1 transfer matrix derived in Eq.~\eqref{eq:FDLU_rep_state} of Sec.~\ref{sec:FDLU}, we present the general state structure of this class:
\begin{equation}
\begin{aligned}
&\begin{tikzpicture}[baseline=(current bounding box),scale=1.2]
        \node at (-1.4,0) {$\cdots$};
        \draw (-.4,.3) -- (-.4,1.2);
        \draw (1,.3) -- (1,1.2);
        \draw (2.4,.3) -- (2.4,1.2);
        \filldraw[fill = white] (-.7,-.3) rectangle node {$A$} (-.1,.3);
        \draw (.7,0) --  (-.1,0);
        \draw (-.7,0) --  (-1.1,0);
        \filldraw[fill = white] (.7,-.3) rectangle node {$A$} (1.3,.3);
        \draw (1.3,0) --  (2.1,0);
        \filldraw[fill = white] (2.1,-.3) rectangle node {$A$} (2.7,.3);
        \draw (2.7,0) --  (3.5,0);
        \node at (3.8,0) {$\cdots$};
        \draw (.1,.8) -- (.1,1.2);
        \node at (.1,.6) {$\ket{+}$};
        \draw (.1+1.4,.8) -- (.1+1.4,1.2);
        \node at (.1+1.4,.6) {$\ket{+}$};
        \draw (.1+2.8,.8) -- (.1+2.8,1.2);
        \node at (.1+2.8,.6) {$\ket{+}$};
\end{tikzpicture}\\[2ex]
=\ &\begin{tikzpicture}[baseline=(current bounding box),scale=1.2]
        \node at (-1.4,0) {$\cdots$};
        \draw (-.5,.3) -- (-.5,.5);
        \draw (-.5,1) -- (-.5,1.2);
        \draw (.1,.3) -- (.1,1.2);
        \draw[rounded corners=8pt]
        (-.5,-.2) -- (-.5,-.5) -- (.1,-.5) -- (.1,-.2);
        \filldraw[draw, circle, fill=white, minimum size=6pt](-.2,-.5) circle (6pt);
        \node at (-.2,-.5) {$S$};
        \draw[rounded corners=8pt]
        (-.5+1.4,-.2) -- (-.5+1.4,-.5) -- (.1+1.4,-.5) -- (.1+1.4,-.2);
        \filldraw[draw, circle, fill=white, minimum size=6pt](-.2+1.4,-.5) circle (6pt);
        \node at (-.2+1.4,-.5) {$S$};
        \draw[rounded corners=8pt]
        (-.5+2.8,-.2) -- (-.5+2.8,-.5) -- (.1+2.8,-.5) -- (.1+2.8,-.2);
        \filldraw[draw, circle, fill=white, minimum size=6pt](-.2+2.8,-.5) circle (6pt);
        \node at (-.2+2.8,-.5) {$S$};
        \draw (.9,.3) -- (.9,1.2);
        \draw (1.5,.3) -- (1.5,1.2);
        \draw (2.3,.3) -- (2.3,1.2);
        \draw (2.9,.3) -- (2.9,.5);
        \draw (2.9,1) -- (2.9,1.2);
        \draw (-.3,.5) --  (-.3,1);
        \draw (-.3,.5) --  (-1,.5);
        \draw (-.3,1) --  (-1,1);
        \node at (-.5,.75) {$U_Z$};
        \filldraw[fill = white] (-.1,.5) rectangle node {} (1.1,1);
        \node at (.1,.75) {$C$};
        \node at (.9,.75) {$U_Z$};
        \filldraw[fill = white] (-.1+1.4,.5) rectangle node {} (1.1+1.4,1);
        \node at (.1+1.4,.75) {$C$};
        \node at (.9+1.4,.75) {$U_Z$};
        \draw (2.7,.5) --  (2.7,1);
        \draw (2.7,.5) --  (3.3,.5);
        \draw (2.7,1) --  (3.3,1);
        \node at (2.9,.75) {$C$};
        
        \draw (-.3,-.2) --  (-.3,.3);
        \draw (-.3,-.2) --  (-1,-.2);
        \draw (-.3,.3) --  (-1,.3);
        \filldraw[fill = white] (-.1,-.2) rectangle node {} (1.1,.3);
        \node at (.5,.05) {$U$};
        \filldraw[fill = white] (-.1+1.4,-.2) rectangle node {} (1.1+1.4,.3);
        \node at (.5+1.4,.05) {$U$};
        \draw (2.7,-.2) --  (2.7,.3);
        \draw (2.7,-.2) --  (3.3,-.2);
        \draw (2.7,.3) --  (3.3,.3);
        
        \node at (3.8,0) {$\cdots$};
\end{tikzpicture}
\end{aligned}  
\end{equation}  
where $U$ is a potential transversal gate that acts on the representative state.

\subsubsection{class $RI$--$RIRX$}
\label{subsec:RI-RIRX}
We use $RI$--$RIRX$ to denote the pushing relation $Z\mapsto I$ for $A$ and $Z\mapsto I$, $X\mapsto X$ for $B$. According to Eq.~\eqref{eq:T_B_RI}, this requires $\lambda_1= \lambda_2 = 0$, and consequently
this class of target state $\ket{A}$ has a transfer matrix in Eq.~\eqref{eq:T_A_RI} or Eq.~\eqref{eq:T_A_RI_app} with $\lambda_1= \lambda_2 = 0$. Recall that class $RIRX$ is 1-round MF-preparable using fusion measurement, we obtain that the target state $\ket{A}$ is also 1-round MF-preparable using Eq.~\eqref{eq:RI_A_B}.  

Using the representative state for class $RIRX$ derived in Eq.~\eqref{eq:RIRX} of Sec.~\ref{subsec:RIRX}, we present the general state structure of this class:
\begin{equation}
\begin{aligned}
&\begin{tikzpicture}[baseline=(current bounding box),scale=1.2]
        \node at (-2.1,0) {$\cdots$};
        \draw (-.4,.3) -- (-.4,1.2);
        \draw (1,.3) -- (1,1.2);
        \draw (2.4,.3) -- (2.4,1.2);
        \filldraw[fill = white] (-.7,-.3) rectangle node {$A$} (-.1,.3);
        \draw (.7,0) --  (-.1,0);
        \draw (-.7,0) --  (-1.8,0);
        \filldraw[fill = white] (.7,-.3) rectangle node {$A$} (1.3,.3);
        \draw (1.3,0) --  (2.1,0);
        \filldraw[fill = white] (2.1,-.3) rectangle node {$A$} (2.7,.3);
        \draw (2.7,0) --  (3.1,0);
        \node at (3.4,0) {$\cdots$};
        \draw (.1,.8) -- (.1,1.2);
        \node at (.1,.6) {$\ket{+}$};
        \draw (.1+1.4,.8) -- (.1+1.4,1.2);
        \node at (.1+1.4,.6) {$\ket{+}$};
        \draw (.1-1.4,.8) -- (.1-1.4,1.2);
        \node at (.1-1.4,.6) {$\ket{+}$};
\end{tikzpicture}\\[2ex]
=\ &\begin{tikzpicture}[baseline=(current bounding box),scale=1.2]
        \draw[rounded corners=8pt] (-1.8, -1) -- (-1.3, -1) -- (-1.3,-.8);
        \draw (-1.3,-.6) -- (-1.3,-.2);
        \draw (-1.3,.3) -- (-1.3,1.2);
        
        \draw[rounded corners=8pt]
        (-1.8,-.7) -- (-.5,-.7) -- (-.5,-.2);
        \draw[rounded corners=8pt]
        (-1.5,-1.8) -- (-1.5,-1.2) -- (-.5,-1.2) -- (-.5,-.7) -- (.9,-.7) -- (.9,-.2);
        \filldraw[fill = myRed] (-1.2,-1.4) rectangle node {} (-.7,-.5);
        \draw[rounded corners=8pt] (-1.3,-1.8) -- (-1.3,-1.6) -- (-.4, -1.6) -- (-.4, -1) -- (.1, -1) -- (.1,-.8);
        \draw (.1,-.6) -- (.1,-.2);
        \node at (-1.4,-2) {$\ket{\psi_{M}}$};
        \node at (-.95,-1.2) {$X$};
        \node at (-.95,-.7) {$C$};

        \draw[rounded corners=8pt]
        (-1.5+1.4,-1.8) -- (-1.5+1.4,-1.2) -- (-.5+1.4,-1.2) -- (-.5+1.4,-.7) -- (.9+1.4,-.7) -- (.9+1.4,-.2);
        \filldraw[fill = myRed] (-1.2+1.4,-1.4) rectangle node {} (-.7+1.4,-.5);
        \draw[rounded corners=8pt] (-1.3+1.4,-1.8) -- (-1.3+1.4,-1.6) -- (-.4+1.4, -1.6) -- (-.4+1.4, -1) -- (.1+1.4, -1) -- (.1+1.4,-.8);
        \draw (.1+1.4,-.6) -- (.1+1.4,-.2);
        \node at (-1.4+1.4,-2) {$\ket{\psi_{M}}$};
        \node at (-.95+1.4,-1.2) {$X$};
        \node at (-.95+1.4,-.7) {$C$};
        
        \draw[rounded corners=8pt]
        (-1.6+2.8,-.7) -- (-.5+2.8,-.7) -- (-.5+2.8,-.2);
        \draw[rounded corners=8pt]
        (-1.5+2.8,-1.8) -- (-1.5+2.8,-1.2) -- (-.5+2.8,-1.2) -- (-.5+2.8,-.7) -- (.1+2.8,-.7);
        \filldraw[fill = myRed] (-1.2+2.8,-1.4) rectangle node {} (-.7+2.8,-.5);
        \draw[rounded corners=8pt] (-1.3+2.8,-1.8) -- (-1.3+2.8,-1.6) -- (-.4+2.8, -1.6) -- (-.4+2.8, -1) -- (.1+2.8, -1);
        \node at (-1.4+2.8,-2) {$\ket{\psi_{M}}$};
        \node at (-.95+2.8,-1.2) {$X$};
        \node at (-.95+2.8,-.7) {$C$};

        \node at (-2.1,-.4) {$\cdots$};
        \node at (3.4,-.4) {$\cdots$};
        
        \draw (-.5,.3) -- (-.5,1.2);
        \draw (.1,.3) -- (.1,1.2);
        \draw (.9,.3) -- (.9,1.2);
        \draw (1.5,.3) -- (1.5,1.2);
        \draw (2.3,.3) -- (2.3,1.2);
        \filldraw[fill = white] (-.1-1.4,.5) rectangle node {} (1.1-1.4,1);
        \node at (.1-1.4,.75) {$C$};
        \node at (.9-1.4,.75) {$U_Z$};
        \filldraw[fill = white] (-.1,.5) rectangle node {} (1.1,1);
        \node at (.1,.75) {$C$};
        \node at (.9,.75) {$U_Z$};
        \filldraw[fill = white] (-.1+1.4,.5) rectangle node {} (1.1+1.4,1);
        \node at (.1+1.4,.75) {$C$};
        \node at (.9+1.4,.75) {$U_Z$};
        
        \draw (-1.1,-.2) --  (-1.1,.3);
        \draw (-1.1,-.2) --  (-1.6,.-.2);
        \draw (-1.1,.3) --  (-1.6,.3);
        \filldraw[fill = white] (-.7,-.2) rectangle node {$U$} (.3,.3);
        \filldraw[fill = white] (.7,-.2) rectangle node {$U$} (1.7,.3);
        \draw (2.1,-.2) --  (2.1,.3);
        \draw (2.1,-.2) --  (2.6,.-.2);
        \draw (2.1,.3) --  (2.6,.3);
\end{tikzpicture}
\label{eq:RI-RIRX}
\end{aligned}  
\end{equation}  

In summary, the states $\ket{A}$ and  $\ket{B}$ differ by a depth-1 controlled-$U_Z$ circuit entangling adjacent sites as shown in Eq.~\eqref{eq:RI_A_B}, where the associated state $\ket{B}$ is in class $RIRX$ and is fusion measurement preparable. We note that the class $RI$--$RIRX$ is in general beyond the fusion measurement protocol with transversal feedback based on pushable defects. This is because we need the defect $X \cos{\alpha} +Y \sin{\alpha} $ for certain $\alpha$ to be pushable, to extend $\{I,Z\}$ to a complete set of orthonormal defects as required by the fusion measurement preparation. However, it is straightforward to see that this operator is neither left nor right pushable for generic values of $\lambda$'s.   Another heuristic understanding is from the fact that all fusion measurement preparable states from pushing Pauli defects can be prepared using a sequential Clifford circuit~\cite{zhang2024characterizing}. However, neither $CU_Z$ nor $U$ gates in the above circuit are necessarily Clifford.

As a final remark, we rewrite the MPS tensor for $\ket{A}\otimes \ket{+}^{\otimes N}$ using Eq.~\eqref{eq:RI-RIRX} as
\beq
    \begin{tikzpicture}[baseline=(current bounding box),scale=1.5]
        \draw[rounded corners=8pt]
        (-1,.8) -- (1,.8) -- (1,1.3);
        \draw[rounded corners=8pt]
        (0,-.3) -- (0,.4) -- (2.2,.4);
        \filldraw[fill = myRed] (.3,.2) rectangle node {} (.8,1);
        \draw[rounded corners=8pt] (.2,-.3) -- (.2,-.1) -- (1.3, -.1) -- (1.3, .3);
        \draw (1.3, .5) -- (1.3, 1.3);
        \node at (.1,-.5) {$\ket{\psi_{M}}$};
        \node at (.55,.4) {$X$};
        \node at (.55,.8) {$C$};
        \draw[rounded corners=8pt]
        (-1,2) -- (-.4,2) -- (-.4,2.8);
        \draw[rounded corners=8pt]
        (1,1.3) -- (1,2) -- (0,2) -- (0,2.8);
        \draw[rounded corners=8pt]
        (1.3,1.7) -- (1.3,2) -- (2.2,2);
        \filldraw[fill = white] (.8,1.3) rectangle node {$U$} (1.5,1.7);
        \filldraw[fill = white] (-.6,2.2) rectangle node {} (.2,2.6);
        \node at (0,2.4) {$U_Z$};
        \node at (-.4,2.4) {$C$};
        \end{tikzpicture}\ .
        \label{eq:RI-RIRX-+}
\eeq
From the above tensor, we can conclude that class $RI$--$RIRX$ can be prepared using fusion measurement with {\it non-transversal feedback correction}. This is because all unitary defects in the top left virtual bond can be pushed to the right. In addition, due to the $RIRX$ property of $B$, the $Z$ and $X$ defects in the bottom virtual bond can be pushed to the right. Thus, all defects in the top and bottom virtual bonds can be pushed using unitaries supported on two sites; see more details in Appendix \ref{app:RIRI-yyyy}. To summarize, for any state $\ket{A}$ in class $RI$--$RIRX$, after stacking with the product of $\ket{+}$, we can write the MPS tensor as in Eq.~\eqref{eq:RI-RIRX-+}. This tensor lies within class $RIRI$--$RIRX$ as will be discussed in Sec.~\ref{sec:Pauli_subspace}. In other words, we conclude that class $RI$--$RIRX$ is contained in class $RIRI$--$RIRX$. 

\subsection{Class $RZ$}
We use $RZ$ to denote that $Z$ is right-pushed to $Z$: $Z\mapsto Z$. The general transfer matrix $T_A$ is given by
\beq
    T_A=\sum_{a,b,c=0}^{D-1}\lambda_{a,b,c}\quad \begin{tikzpicture}[baseline=(current bounding box),scale=1.5]
        \draw[rounded corners=8pt]
        (-.8,.8) -- (0,.8) -- (0,0) -- (-.8,0);
        \draw[rounded corners=8pt]
        (1,.8) -- (.2,.8) -- (.2,0) -- (1,0);
        \filldraw[draw, circle, fill=myRed, minimum size=6pt](0.6,0) circle (6pt);
        \node at (0.6,0) {$X^{\bar{b}}$};
        \filldraw[draw, circle, fill=myBlue, minimum size=6pt](0.6,.8) circle (6pt);
        \node at (0.6,.8) {$Z^{c}$};
        \filldraw[draw, circle, fill=myRed, minimum size=6pt](-0.4,0) circle (6pt);
        \node at (-0.4,0) {$X^{b}$};
        \filldraw[draw, circle, fill=myBlue, minimum size=6pt](-0.4,0.8) circle (6pt);
        \node at (-0.4,0.8) {$Z^a$};
        \end{tikzpicture}\ .
\eeq

When taking the virtual bond dimension $D=2$, the general transfer matrix with $RZ$ pushing relation is as in Eq.~\eqref{eq:RI_transfer_matrix}, with the two-qubit Hermitian operator $M_A$ given by
\beq
    M_A &= \lambda_0 I +\lambda_1 Z_1 +\lambda_2 Z_2 +\lambda_3 Z_1 Z_2  +\lambda_4 X_1 X_2 \\
        &\quad\quad + i\lambda_5 Z_1 X_1 X_2 + i\lambda_6 X_1 Z_2 X_2 + \lambda_7 Z_1 X_1 Z_2 X_2.
    \label{eq:RZ_transfer_matrix}
\eeq
The associated tensor $B$ in Eq.~\eqref{eq:B_RZ} has a transfer matrix given by the Hermitian operator
\beq
        M_B =& \lambda_0 I +\lambda_1 Z_1 +\lambda_2 Z_2 +\lambda_3 Z_1 Z_2.
        \label{eq:T_B_RZ}
\eeq

Unlike the $RI$ classes, for the $RZ$ class, the pushing relation for tensor $A$ is $Z\mapsto Z$, and therefore $\ket{A}$ cannot be obtained from $\ket{B}$ using FDLUs as is possible for $RI$. In the following, we discuss two classes among the $RZ$ classes of states: 1-round MF-preparable class $RZ$--$\text{FDLU}$, and 2-round MF-preparable class $RZ$--$RIRX$. We note that we do not need to discuss class $RZ$--$RZRX$, because the transfer matrix of the associated state as shown in Eq.~\eqref{eq:T_B_RZ} only allows $Z\mapsto I$ and $I\mapsfrom Z$.

\subsubsection{class $RZ$--$\text{FDLU}$}
\label{subsec:RZ-FDLU}
We use $RZ$--$\text{FDLU}$ to denote the class of MPSs that satisfy $Z\mapsto Z$ for $A$, and $B$ is FDLU-preparable. When $\lambda_0 \lambda_3 = \lambda_1 \lambda_2$, this transfer matrix is of rank 1, then $\ket{B}$ is FDLU-preparable and accordingly $\ket{A}$ is 1-round MF-preparable. Two special classes contained are $RZ$--$RIRI$ and $RZ$--$RILI$, which correspond to $\lambda_1=\lambda_3=0$ and $\lambda_2=\lambda_3=0$ in the transfer matrix in Eq.~\eqref{eq:RZ_transfer_matrix}, respectively. 

In Appendix \ref{app:RZ-RILI}, we derive the state structure for class $RZ$--$RILI$, which is as follows:
\begin{equation}
\begin{aligned}
&\begin{tikzpicture}[baseline=(current bounding box),scale=1.2]
        \node at (-1,0) {$\cdots$};
        \draw (1,.3) -- (1,1.2);
        \draw (2.4,.3) -- (2.4,1.2);
        \draw (.7,0) --  (-.2,0);
        \filldraw[fill = white] (.7,-.3) rectangle node {$A$} (1.3,.3);
        \draw (1.3,0) --  (2.1,0);
        \filldraw[fill = white] (2.1,-.3) rectangle node {$A$} (2.7,.3);
        \draw (2.7,0) --  (3,0);
        \node at (3.3,0) {$\cdots$};
        \draw (.1,.8) -- (.1,1.2);
        \node at (.1,.6) {$\ket{+}$};
        \draw (.1+1.4,.8) -- (.1+1.4,1.2);
        \node at (.1+1.4,.6) {$\ket{+}$};
\end{tikzpicture}\\[2ex]
=&\ \begin{tikzpicture}[baseline=(current bounding box),scale=1.2]
        \draw (-1.3,.3) -- (-1.3,1.2);

        \draw[rounded corners=8pt] (-1.3,-.2) -- (-1.3,-.5) -- (-2,-.5);
        \draw[rounded corners=8pt] (-2, -.8) -- (-1.2, -.8) -- (-1.2,-3.7);
        \draw[rounded corners=8pt]
        (-.5,-.2) -- (-.5,-2) -- (.2,-2) -- (.2,-3.7);
        \filldraw[fill = myRed] (-1.4,-1.5) rectangle node {} (-.3,-1);
        \draw[rounded corners=8pt] (-.4, -1.7) -- (.1, -1.7) -- (.1,-.2);
        \draw[rounded corners=8pt] (-.8,-3.7) -- (-.8,-1.7) -- (-.6, -1.7);
        \draw (.1,-.6) -- (.1,-.2);
        \node at (-1,-3.9) {$\ket{\psi_{M}}$};
        \node at (-1.2,-1.25) {$C$};
        \node at (-.5,-1.25) {$X$};

        \draw[rounded corners=8pt]
        (.9,-.2) -- (.9,-3.3) -- (1.5,-3.3);
        \filldraw[fill = myRed] (0,-2.8) rectangle node {} (1.1,-2.3);
        \draw (1, -3) -- (1.5, -3);
        \draw[rounded corners=8pt] (.6,-3.7) -- (.6,-3) -- (.8, -3);
        \draw (.1,-.6) -- (.1,-.2);
        \node at (.4,-3.9) {$\ket{\psi_{M}}$};
        \node at (.2,-2.55) {$C$};
        \node at (.9,-2.55) {$X$};

        \node at (-2.4,-.65) {$\cdots$};
        \node at (1.9,-3.15) {$\cdots$};
        
        \draw (-.5,.3) -- (-.5,1.2);
        \draw (.1,.3) -- (.1,1.2);
        \draw (.9,.3) -- (.9,1.2);
        \filldraw[fill = white] (-.1-1.4,.5) rectangle node {} (1.1-1.4,1);
        \node at (.1-1.4,.75) {$C$};
        \node at (.9-1.4,.75) {$U_Z$};
        \filldraw[fill = white] (-.1,.5) rectangle node {} (1.1,1);
        \node at (.1,.75) {$C$};
        \node at (.9,.75) {$U_Z$};
        
        \draw (-1.1,-.2) --  (-1.1,.3);
        \draw (-1.1,-.2) --  (-1.8,.-.2);
        \draw (-1.1,.3) --  (-1.8,.3);
        \node at (-1.6,.05) {$U$};
        \filldraw[fill = white] (-.7,-.2) rectangle node {$U$} (.3,.3);
        \draw (.7,-.2) --  (.7,.3);
        \draw (.7,-.2) --  (1.4,.-.2);
        \draw (.7,.3) --  (1.4,.3);
        \node at (1.2,.05) {$U$};
\end{tikzpicture}\ .
\label{eq:RZ-RILI}
\end{aligned}  
\end{equation}  
As a result, the states in class $RZ$--$RILI$ are prepared from a product state using a sequential $CX$ circuit, followed by a brick wall circuit composed of $CU_Z$ and $U\cdot \text{SWAP}$ gates. Similar to the class $RI$--$RIRX$ above, this class of states is in general also beyond fusion measurement protocol, as it is impossible to extend $\{I,Z\}$ to a complete set of orthonormal basis. However, we can write the MPS tensor for $\ket{A}\otimes \ket{+}^{\otimes N}$ using Eq.~\eqref{eq:RZ-RILI} as
\beq
    \begin{tikzpicture}[baseline=(current bounding box),scale=1.5]
        \draw[rounded corners=8pt]
        (-1,1.2) -- (.6,1.2) -- (.6,-.3);
        \draw[rounded corners=8pt]
        (1,1.3) -- (1,.4) -- (2.2,.4);
        \filldraw[fill = myRed] (.4,.6) rectangle node {} (1.2,1);
        \draw[rounded corners=8pt] (.8,-.3) -- (.8,-.1) -- (1.3, -.1) -- (1.3, .3);
        \draw (1.3, .5) -- (1.3, 1.3);
        \node at (.7,-.5) {$\ket{\psi_{M}}$};
        \node at (1,.8) {$X$};
        \node at (.6,.8) {$C$};
        \draw[rounded corners=8pt]
        (-1,2) -- (-.4,2) -- (-.4,2.8);
        \draw[rounded corners=8pt]
        (1,1.3) -- (1,2) -- (0,2) -- (0,2.8);
        \draw[rounded corners=8pt]
        (1.3,1.7) -- (1.3,2) -- (2.2,2);
        \filldraw[fill = white] (.8,1.3) rectangle node {$U$} (1.5,1.7);
        \filldraw[fill = white] (-.6,2.2) rectangle node {} (.2,2.6);
        \node at (0,2.4) {$U_Z$};
        \node at (-.4,2.4) {$C$};
        \end{tikzpicture}\ .
\eeq
From the above tensor, we can conclude that class $RZ$--$RILI$ can be prepared using fusion measurement with {\it non-transversal feedback correction}; see also Appendix \ref{app:RIRI-yyyy}. This is because all unitary defects in the top left virtual bond can be pushed to the right. In addition, the $Z$ and $X$ defects in the bottom right virtual bond can be pushed to the left, so all defects are pushable using unitaries supported on two sites. This implies that the above tensor lies within class $RIRI$--$RZLI$ as will be discussed in Sec.~\ref{sec:Pauli_subspace}, and the class $RZ$--$RILI$ is contained in class $RIRI$--$RZLI$.

\subsubsection{class $RZ$--$RIRX$}
\label{subsec:RZ-RIRX}
We use $RZ$--$RIRX$ to denote the pushing relation $Z\mapsto Z$ for $A$ and $Z\mapsto I$, $X\mapsto X$ for $B$. In particular, when $\lambda_1=\lambda_2=0$ in the transfer matrix in Eq.~\eqref{eq:T_B_RZ}, $T_B$ allows Pauli $X$ to be also pushable and therefore is in class $RIRX$. Consequently, $\ket{B}$ can be prepared using the 1-round MF-circuit that implements generalized KW, and $\ket{A}$ is 2-round MF-preparable. 

We derive a representative tensor $A$ in Appendix \ref{app:RZ} for $T_A$ with generic values of $\lambda$'s. In the same Appendix, we show that when  $\lambda_1=\lambda_2=0$, the representative tensor is of the form
\beq
        A=\ \begin{tikzpicture}[baseline=(current bounding box),scale=1.5]
        \draw (-.3,0.4) -- (-.1,.4);
        \draw (.1,0.4) -- (2.4,.4);
        \draw (2.6,0.4) -- (4.8,.4);
        \draw[rounded corners=8pt]
        (0,-.4) -- (0,.8) -- (2.4,.8) -- (2.4, 1.2) -- (4.4, 1.2) -- (4.4, 1.5);
        \draw[rounded corners=8pt]
        (.3,-.4) -- (0.3, 0) -- (2.5,0) -- (2.5,.8) -- (4.7,.8) -- (4.7, 1.5);
        \filldraw[fill = myRed] (.3,.2) rectangle node {} (.8,1); 
        \node at (.55,.8) {$C$};
        \node at (.55,.4) {$X$};
        \filldraw[draw, circle, fill=myBlue, minimum size=6pt](1.2,0.4) circle (6pt);
        \node at (1.2,0.4) {$R^z_{\alpha_{1}}$};
        \filldraw[fill = myBlue] (1.6,.2) rectangle node {} (2.2,1); 
        \node at (1.9,.8) {$C$};
        \node at (1.9,.4) {$R^z_{\alpha_{2}}$};
        \filldraw[fill = myBlue] (2.8,.2) rectangle node {} (3.4,1); 
        \node at (3.1,.8) {$C$};
        \node at (3.1,.4) {$R^z_{\alpha_{3}}$};
        \filldraw[fill = myBlue] (3.6,.2) rectangle node {} (4.2,1.4); 
        \node at (3.9,1.2) {$C$};
        \node at (3.9,.8) {$C$};
        \node at (3.9,.4) {$R^z_{\alpha_{4}}$};
        \node at (0.15,-.6) {$\ket{\psi_{r}}$};
        \end{tikzpicture}
        \label{eq:A_RZ-RZRX}
\eeq 
where the $\ket{\psi_{r}}= \sum_{m,n=0,1}r_{m,n}\ket{m}\otimes \ket{n}$ is an arbitrary two-qubit state. We have denoted the single-qubit rotation gate, the two-qubit controlled-rotation gate, and the three-qubit controlled-controlled-rotation gate as
\beq
    &R^z_{\alpha} = e^{i\frac{\alpha Z}{2}},\quad CR^z_{\alpha} = e^{i\frac{\alpha (1-Z_1)Z_2}{4}},\\
    &CCR^z_{\alpha} = e^{i\frac{\alpha (1-Z_1)(1-Z_2)Z_3}{8}}.
\eeq
where the top two bonds are the controlled qubits, and the bottom bond is the target qubit. Both the rotation angles $\alpha_i$ and the parameters $r_{m,n}$ in state $\ket{\psi_r}$ can be arbitrary.

According to Eq.~\eqref{eq:T_B} with pushable defects $I$ and $Z$, the transfer matrix $T_B$ of the associated tensor is given by
\beq
    T_B =c\ \begin{tikzpicture}[baseline=(current bounding box),scale=1.5]
        \draw[rounded corners=8pt]
        (-.5,.8) -- (0,.8) -- (0,0) -- (-.5,0);
        \draw[rounded corners=8pt]
        (.7,.8) -- (.2,.8) -- (.2,0) -- (.7,0);
        \end{tikzpicture}\ +  c'\ \begin{tikzpicture}[baseline=(current bounding box),scale=1.5]
        \draw[rounded corners=8pt]
        (-.8,.8) -- (-.3,.8) -- (-.3,0) -- (-.8,0);
        \draw[rounded corners=8pt]
        (.8,.8) -- (.3,.8) -- (.3,0) -- (.8,0);
        \filldraw[draw, circle, fill=myBlue, minimum size=6pt](-.3,0.4) circle (6pt);
        \node at (-0.3,0.4) {$Z$};
        \filldraw[draw, circle, fill=myBlue, minimum size=6pt](.3,0.4) circle (6pt);
        \node at (0.3,0.4) {$Z$};
        \end{tikzpicture}\ ,
\eeq
where $c=\langle\psi_r|\psi_r\rangle$ and $c'=\bra{\psi_r}Z\otimes I\ket{\psi_r}$. Therefore, indeed the associated state $\ket{B}$ is in class $RIRX$, and consequently $\ket{A}$ is 2-round MF-preparable. In particular, when $\bra{\psi_r}Z\otimes I\ket{\psi_r} = 0$, which corresponds to $\lambda_3 = 0$ in Eq.~\eqref{eq:T_B_RZ}, the associated transfer matrix $T_B$ is of rank 1. Therefore, $\ket{B}$ is FDLU-preparable and state $\ket{A}$ is 1-round MF-preparable. In Sec.~\ref{subsec:example_RZ-RIRX}, we will discuss in detail specific examples of this class. 

The MPO in Eq.~\eqref{eq:A_RZ-RZRX} describes a sequential circuit composed of controlled-$X$ gates and certain non-Clifford controlled-rotation gates. In particular, when only $\alpha_1=\alpha$ is non-zero, this MPO is related to the cosine symmetry $\{L_{\alpha}\}$~\cite{aasen2020topological,seifnashri2025gauging} by a conjugation of Hadamard gates. The cosine symmetry is a set of continuous non-invertible operators that have fusion rules like the product of cosine functions
\beq
    L_{\alpha}L_{\alpha'} = L_{\alpha+\alpha'} + L_{\alpha-\alpha'}.
\eeq
The cosine symmetries are respected by the gauged XXZ spin chain, i.e., the XXZ model after the Kramers-Wannier duality~\cite{thorngren2024fusion2,aasen2020topological,seifnashri2025gauging}. The parameter $\alpha$ is $2\pi$-periodic, and for special values of $\alpha$, the operator is particularly familiar: $L_0 = 1+ \prod_i X_i$, $L_\pi = \prod_i X_{2i}+ \prod_i X_{2i+1}$, $L_{\pi/2} = \text{KT}$.

When taking $\alpha_2=\alpha_4=0$, the MPO is unitarily equivalent to a generalization of the cosine symmetry, denoted as $L_{\alpha_1, \alpha_3}$. In Appendix \ref{app:cosine}, we show that
\beq
    L_{\alpha_1, \alpha_3} = \text{KW}^{\dagger} \big(\prod_i R^z_{i, \alpha_1} CR^z_{\tilde{i} i, \alpha_3} \big) \text{KW},
    \label{eq:generalized_cosine}
\eeq
where $R^z_{i, \alpha_1}$ is a single-qubit rotation on the first qubit of site $i$, and $CR^z_{\tilde{i} i, \alpha_3}$ is the controlled-rotation from the second to the first qubit of site $i$. As a result, the fusion rule of this MPO is given by
\beq
    L_{\alpha_1, \alpha_3}L_{\alpha'_1, \alpha'_3} = L_{\alpha_1+\alpha'_1, \alpha_3+\alpha'_3} + L_{\alpha_1-\alpha'_1, \alpha_3-\alpha'_3}.
\eeq

Therefore, we have found that a class of states with $RZ$--$RIRX$ pushing relations can be prepared from product states using the generalized cosine symmetry. According to Eq.~\eqref{eq:generalized_cosine}, the MPS tensor of these states can be rewritten in class $RIRX$--$RIRX$. This class will be discussed in Sec.~\ref{sec:Pauli_subspace}. 

\subsection{Non-example: $Z\leftrightarrow X$}
\label{subsec:non-pushable}

In Sec.~\ref{subsec:left_right_pushable}, we mention a case with the pushing relation denoted as $Z\leftrightarrow X$, where neither $Z$ nor $X$ defect is pushable. This pushing relation is related to $RZ$ by an extra Hadamard gate on the right virtual bond. Therefore, from the class $RZ$ states in Eq.~\eqref{eq:A_RZ}, we can obtain the MPS tensor for a generic state with $Z\leftrightarrow X$ pushing relation by applying a Hadamard.  The structure of a subset of states obtained this way from the $RZ$--$RIRX$ class in Eq.~\eqref{eq:A_RZ-RZRX} is below:
\beq
        A=\ \begin{tikzpicture}[baseline=(current bounding box),scale=1.5]
        \draw (-.2,0.4) -- (-.1,.4);
        \draw (.1,0.4) -- (2.4,.4);
        \draw (2.6,0.4) -- (4.8,.4);
        \draw[rounded corners=8pt]
        (0,-.4) -- (0,.8) -- (2.4,.8) -- (2.4, 1.2) -- (4.4, 1.2) -- (4.4, 1.5);
        \draw[rounded corners=8pt]
        (.3,-.4) -- (0.3, 0) -- (2.5,0) -- (2.5,.8) -- (4.7,.8) -- (4.7, 1.5);
        \filldraw[fill = myRed] (.3,.2) rectangle node {} (.8,1); 
        \node at (.55,.8) {$C$};
        \node at (.55,.4) {$X$};
        \filldraw[draw, circle, fill=myBlue, minimum size=6pt](1.2,0.4) circle (6pt);
        \node at (1.2,0.4) {$R^z_{\alpha_{1}}$};
        \filldraw[fill = myBlue] (1.6,.2) rectangle node {} (2.2,1); 
        \node at (1.9,.8) {$C$};
        \node at (1.9,.4) {$R^z_{\alpha_{2}}$};
        \filldraw[fill = myBlue] (2.8,.2) rectangle node {} (3.4,1); 
        \node at (3.1,.8) {$C$};
        \node at (3.1,.4) {$R^z_{\alpha_{3}}$};
        \filldraw[fill = myBlue] (3.6,.2) rectangle node {} (4.2,1.4); 
        \node at (3.9,1.2) {$C$};
        \node at (3.9,.8) {$C$};
        \filldraw[draw, circle, fill=myGray, minimum size=6pt](4.5,.4) circle (6pt);
        \node at (4.5,.4) {$H$};
        \node at (3.9,.4) {$R^z_{\alpha_{4}}$};
        \node at (0.15,-.6) {$\ket{\psi_{r}}$};
        \end{tikzpicture}
\eeq  

The state $\ket{A}$ generated by the above tensor is preparable from the product state $\ket{\psi_r}^{\otimes N}$ using a sequential circuit composed of $CX$, $H$, SWAP, and (controlled-)controlled-rotation gates. However, we suspect that this state does not admit any pushable defect, which would imply that it cannot be prepared by finite rounds of measurements and feedback.

\section{All Pauli defects in a subspace are pushable}
\label{sec:Pauli_subspace}
Another natural case is where all the Pauli defects in a subspace are pushable. The associate MPS $\ket{B}$ has a transfer matrix $T_B$ obtained from performing a projection on $T_A$ into that subspace. While the Pauli defects of $T_A$ in the complementary subspace are not pushable, they can become pushable in $T_B$ after the projection. Let us be more explicit by taking a transfer matrix with decomposed virtual spaces:
\beq
    T_A=&\sum_{\{a_i,b_i,c_i,d_i\}=0}^{D_i-1}\Lambda(\{a_i,b_i,c_i,d_i\})\quad \begin{tikzpicture}[baseline=(current bounding box),scale=1.5]
        \draw[rounded corners=8pt]
        (-.8,.8) -- (0,.8) -- (0,0) -- (-.8,0);
        \draw[rounded corners=8pt]
        (1,.8) -- (.2,.8) -- (.2,0) -- (1,0);
        \filldraw[draw, circle, fill=myRed, minimum size=6pt](0.6,0) circle (6pt);
        \node at (0.6,0) {$X^{d_1}$};
        \filldraw[draw, circle, fill=myBlue, minimum size=6pt](0.6,.8) circle (6pt);
        \node at (0.6,.8) {$Z^{c_1}$};
        \filldraw[draw, circle, fill=myRed, minimum size=6pt](-0.4,0) circle (6pt);
        \node at (-0.4,0) {$X^{b_1}$};
        \filldraw[draw, circle, fill=myBlue, minimum size=6pt](-0.4,0.8) circle (6pt);
        \node at (-0.4,0.8) {$Z^{a_1}$};
        \end{tikzpicture}\\[2ex]
        &\quad \bigotimes\ 
        \begin{tikzpicture}[baseline=(current bounding box),scale=1.5]
        \draw[rounded corners=8pt]
        (-.8,.8) -- (0,.8) -- (0,0) -- (-.8,0);
        \draw[rounded corners=8pt]
        (1,.8) -- (.2,.8) -- (.2,0) -- (1,0);
        \filldraw[draw, circle, fill=myRed, minimum size=6pt](0.6,0) circle (6pt);
        \node at (0.6,0) {$X^{d_2}$};
        \filldraw[draw, circle, fill=myBlue, minimum size=6pt](0.6,.8) circle (6pt);
        \node at (0.6,.8) {$Z^{c_2}$};
        \filldraw[draw, circle, fill=myRed, minimum size=6pt](-0.4,0) circle (6pt);
        \node at (-0.4,0) {$X^{b_2}$};
        \filldraw[draw, circle, fill=myBlue, minimum size=6pt](-0.4,0.8) circle (6pt);
        \node at (-0.4,0.8) {$Z^{a_2}$};
        \end{tikzpicture}\quad \quad\bigotimes\ \cdots
        \label{eq:T_B_2_subspaces}
\eeq
Suppose all Pauli defects in the first subspace are pushable, then the associated transfer matrix $T_B$ can be defined by performing a Bell projection in this subspace. In other words, $T_B$ has the same structure as in Eq.~\eqref{eq:T_B_2_subspaces}, but with the extra condition that $c_1=d_1=0$. It can happen that the coefficient $\Lambda|_{c_1=d_1=0}$ has more symmetries in the complementary subspace than $\Lambda$ for generic values of $c_1$ and $d_1$, which allows Pauli defects in the second subspace to be pushable. When this is true, $\ket{B}$ is either FDLU- or 1-round MF-preparable depending on the class of pushing relation. Consequently, $\ket{A}$ is either FDLU, or 1-round MF, or 2-round MF-preparable.

We denote by $xxxx$--$yyyy$ the class of states $\ket{A}$ that have pushing relations of class $xxxx$ for Pauli defects in a subspace, while the associated $\ket{B}$ have pushing relations of class $yyyy$ for Pauli defects in the complementary subspace. For instances, there are classes $RIRI$--$RIRI$ and $RIRI$--$RILI$, which should be FDLU-preparable. 

While we expect there to be interesting state structures for general $xxxx$--$yyyy$ classes, a detailed study is beyond the scope of this work. We will alternatively discuss some families of states in these classes as examples in the following.

\subsubsection{class $RIRX$--$\text{\text{FDLU}}$}
\label{subsec:RIRIX-FDLU}

The first family of states is given by $\text{KW}\cdot D_M\ket{\psi}^{\otimes N}$, where $D_M$ is an FDLU generated by tensor $M$. The MPS tensor is as follows:
\beq
    A = \quad \begin{tikzpicture}[baseline=(current bounding box),scale=1.5]
        \draw (0,-.25) -- (0,-.5);
        \filldraw[fill = white] (-.25,-.25) rectangle node {$M$} (.25,.25);
        \draw (.25,0) --  (1.7,0);
        \draw (-.25,0) --  (-.7,0);
        \draw[rounded corners=8pt]
        (-.7,1.2) -- (1.1,1.2) -- (1.1,1.7);
        \draw[rounded corners=8pt]
        (0,0.25) -- (0,.8) -- (1.7,.8);
        \filldraw[fill = myBlue] (.3,.6) rectangle node {} (.8,1.4);
        \node at (.55,.8) {$Z$};
        \node at (.55,1.2) {$C$};
        \filldraw[draw, circle, fill=myGray, minimum size=6pt](1.2,.8) circle (6pt);
        \node at (1.2,.8) {$H$};
        \node at (0,-.7) {$\ket{\psi}$};
    \end{tikzpicture}\ .
\eeq
The Pauli $Z$ defect in the top virtual subspace is pushed to $I$, and the $X$ defect in this subspace is pushed to $X$. The associated state $\ket{B}$ after a Bell projection in this subspace is given by $D_M\ket{\psi}^{\otimes N}$. Therefore, this family of states is in class $RIRX$--$\text{FDLU}$. Similarly, the states $\text{KT}\cdot D_M\ket{\psi}^{\otimes N}$ are in the class $RZRX$--$\text{FDLU}$.

\subsubsection{class $RIRI$--$RIRX$}
\label{subsec:RIRI-RIRX}

The second family of states is given by $\text{FDLU}\cdot \text{KW} \ket{\psi}^{\otimes N}$. These states are 1-round MF-preparable, because KW can be implemented using 1-round MF circuit as shown in Appendix \ref{app:implement}. Here, we demonstrate that these states are within class $RIRI$--$RIRX$. It suffices to discuss the depth-2 brick wall circuit,  since any FDLU can be reduced to this architecture via site blocking~\cite{gross2012index}. The MPS tensor is shown below:
\beq
    A = \quad \begin{tikzpicture}[baseline=(current bounding box),scale=1.5]
        \draw[rounded corners=8pt]
        (-.7,2.2) -- (.6,2.2) -- (.6,3);
        \draw[rounded corners=8pt]
        (-.7,.8) -- (1.3,.8) -- (1.3,1.2) -- (2.4,1.2) -- (2.4,2.2) -- (1,2.2) -- (1,3);
        \draw[rounded corners=8pt]
        (0,-.2) -- (0,.4) -- (1.5,.4) -- (1.5,.8) -- (2.8,.8) -- (2.8,2.2) -- (3.2,2.2);
        \draw[rounded corners=8pt]
        (1.7,-.2) -- (1.7,.4) -- (3.2,.4);
        \filldraw[fill = myBlue] (.3,.2) rectangle node {} (.8,1);
        \node at (.55,.8) {$C$};
        \node at (.55,.4) {$Z$};
        \node at (0,-.4) {$\ket{\psi}$};
        \filldraw[draw, circle, fill=myGray, minimum size=6pt](1.1,.4) circle (6pt);
        \node at (1.1,.4) {$H$};
        \filldraw[fill = myBlue] (2,.2) rectangle node {} (2.5,1);
        \node at (2.25,.8) {$C$};
        \node at (2.25,.4) {$Z$};
        \node at (1.7,-.4) {$\ket{\psi}$};
        \filldraw[draw, circle, fill=myGray, minimum size=6pt](2.8,.4) circle (6pt);
        \node at (2.8,.4) {$H$};
        \filldraw[fill = white] (2.2,1.6) rectangle node {$U$} (3,2);
        \filldraw[fill = white] (.4,2.4) rectangle node {$V$} (1.2,2.8);
    \end{tikzpicture}\ .
\eeq
It is obvious that all Pauli defects on the top virtual bond are right-pushable. The associated state $\ket{B}$ is obtained from projecting out this virtual subspace, as derived in Appendix \ref{app:RIRI-yyyy}. The associated tensor $B$ is then given by $\text{KW} \ket{\psi}^{\otimes N}$, and is in class $RIRX$. We thus conclude that $\ket{A}$ is in class $RIRI$--$RIRX$.

In Sec.~\ref{subsec:RI-RIRX} and Sec.~\ref{subsec:RZ-FDLU}, we showed that we can always rewrite the MPS tensors in class $RI$--$RIRX$ and class $RZ$--$RILI$ in the form of the above equation, with the unitary $V$ replaced by a controlled-$U_Z$ gate. Hence, these two classes are contained in class $RIRI$--$RIRX$. The 1-round MF-preparation protocol of the target state $\ket{A}$ is thus composed of the fusion measurement preparation of $\ket{B}$ followed by a brick-wall circuit. The construction of a brick-wall circuit that relates $\ket{B}$ to $\ket{A}$ is discussed in Eq.~\eqref{eq:A_B_RIRI} of Sec.~\ref{subsec:RIRI}, which uses the unitaries that push through the Pauli defects.

We further note that states in class $RIRI$--$RIRX$ are generally beyond the scope of fusion measurement. As we showed above, state $\text{FDLU}\cdot \text{KW}\ket{\psi}^{\otimes N}$ is in this class after site blocking. For generic FDLU, the state almost always admits a non-zero correlation length and non-flat entanglement spectrum. According to Ref.~\cite{sahay2025classifying}, these states are impossible to prepare from a product state using fusion measurement and defects right-pushing.

\subsubsection{and beyond}
\label{subsec:beyond}

We can continue to look at the family of states given by
\beq
    \text{FDLU}\cdot \text{KW} \cdots  \text{KW}\cdot \text{FDLU}\ket{\psi}^{\otimes N},
\eeq
for arbitrary FDLU circuits. Following our previous argument, after a blocking of finite sites, we can write the MPS tensor of these states as in class $RIRI$--$RIRX$--$\cdots$--$RIRX$--$RIRI$. 

Similarly, the family of states given by
\beq
    \text{FDLU}\cdot \text{KW} \cdot \text{FDLU} \cdots \ket{A},
\eeq
for a state $\ket{A}$ with pushing class $yyyy$. These states have MPS tensor in class $RIRI$--$RIRX$--$RIRI$--$\cdots$--$yyyy$. In Sec.~\ref{subsec:RIRX-RZRX}, we will study an example in class $RIRX$--$RZRX$, which can be prepared using 2-round MF circuit within our scheme. Furthermore, we found that the MF circuit can be made more compact, such that the state is preparable in one round with non-transversal feedback correction.

\section{Examples}\label{sec:examples}
So far we have illustrated the general structures for states associated with various classes of pushing relations. In this section, we will focus on concrete examples of states and their preparation circuits.

\subsection{Class $RZRX$: fusion measurement}
\label{subsec:example_fusion_measurement}
\subsubsection{AKLT state} 
\label{subsec:example_AKLT}
The spin-$1$ AKLT state provides a paradigmatic example of symmetry-protected topological (SPT) order~\cite{aklt,guwen2009tensor,pollmann2010entanglement}, and is also known as a resource state for measurement-based quantum computation (MBQC)~\cite{mbqc_2001}. It is realized in Ref.~\cite{smith2023deterministic} that 1D AKLT state can be prepared using fusion measurement. Here, instead of following the fusion measurement protocol, we demonstrate how the AKLT state can be obtained from a product state by implementing the spin-$\frac{1}{2}$ KT duality, based on the class of pushing relations. 

The conventional MPS tensor for the AKLT state is given by Pauli operators: $A^1 = \frac{1}{\sqrt{3}}X$, $A^2 = \frac{1}{\sqrt{3}}Y$, and $A^3 = \frac{1}{\sqrt{3}}Z$~\cite{smith2023deterministic}. The transfer matrix of this tensor is 
\beq
    T_{\text{AKLT}} = \frac{1}{3} \left(\ 
    \begin{tikzpicture}[baseline=(current bounding box),scale=1.2]
        \draw (0,0.8) -- (1,.8);
        \draw (0,0) -- (1,.0);
        \filldraw[draw, circle, fill=myBlue, minimum size=6pt](.5,.8) circle (6pt);
        \node at (.5,.8) {$Z$};
        \filldraw[draw, circle, fill=myBlue, minimum size=6pt](.5,0) circle (6pt);
        \node at (.5,0) {$Z$};
    \end{tikzpicture}\
    +\ \begin{tikzpicture}[baseline=(current bounding box),scale=1.2]
        \draw (0,0.8) -- (1,.8);
        \draw (0,0) -- (1,.0);
        \filldraw[draw, circle, fill=myRed, minimum size=6pt](.5,.8) circle (6pt);
        \node at (.5,.8) {$X$};
        \filldraw[draw, circle, fill=myRed, minimum size=6pt](.5,0) circle (6pt);
        \node at (.5,0) {$X$};
    \end{tikzpicture}\
    +\ \begin{tikzpicture}[baseline=(current bounding box),scale=1.2]
        \draw (0,0.8) -- (1.4,.8);
        \draw (0,0) -- (1.4,.0);
        \filldraw[draw, circle, fill=myBlue, minimum size=6pt](.4,.8) circle (6pt);
        \node at (.4,.8) {$Z$};
        \filldraw[draw, circle, fill=myBlue, minimum size=6pt](.4,0) circle (6pt);
        \node at (.4,0) {$Z$};
        \filldraw[draw, circle, fill=myRed, minimum size=6pt](1,.8) circle (6pt);
        \node at (1,.8) {$X$};
        \filldraw[draw, circle, fill=myRed, minimum size=6pt](1,0) circle (6pt);
        \node at (1,0) {$X$};
    \end{tikzpicture}\ 
    \right).
\eeq
This transfer matrix commutes with both $Z\otimes Z$ and $X\otimes X$. As a result, the pushing relations are $Z\mapsto Z$, $X\mapsto X$, and the AKLT state is in class $RZRX$. According to Lemma \ref{lemma:KW_KT}, the AKLT state can be prepared from a product state using KT duality followed by a transversal circuit:
\beq
    \ket{\text{AKLT}}=\Big(\prod_i U_i\Big)\text{KT}\ket{\psi}^{\otimes N},
    \label{eq:AKLT}
\eeq
where $\ket{\psi} = \frac{1}{\sqrt{3}}(\ket{-+}+\ket{--}+\ket{+-})$ is the two-qubit state $\ket{\psi_\mu}$ in Eq.~\eqref{eq:A_RZRX} conjugated by a transversal Hadamard gate. The two-qubit unitary gate $U_{i}$ supported on site $2i-1$ and site $2i$ maps our embedding of spin 1 triplet ($\ket{-+}$, $\ket{--}$, and $\ket{+-}$) to the standard embedding ($\ket{00}$, $\frac{1}{\sqrt{2}}(\ket{01}+\ket{10})$, and $\ket{11}$).

In the work of Kennedy and Tasaki that originated the term KT duality, a non-local transformation on a spin-1 chain is proposed~\cite{kennedy1992hidden}. Upon the standard embedding, that transformation becomes $(\prod_i U_i)\text{KT}(\prod_i U_i)^{\dagger}$ on a qubit chain. Such a KT duality maps between the AKLT state and a $\mathbb{Z}_2\times \mathbb{Z}_2$ symmetry-breaking state. From the above equation, the $\mathbb{Z}_2\times \mathbb{Z}_2$ symmetry-breaking state is simply the symmetrization of the product state $(U\ket{\psi})^{\otimes N}$. By applying a duality transformation on $(U\ket{\psi})^{\otimes N}$, we obtain the AKLT state, which is a non-fixed-point SPT state with finite correlation length. Therefore, Eq.~\eqref{eq:AKLT} provides another preparation protocol of the AKLT state besides fusion measurement, since KT can be implemented using the MF circuit as discussed in Appendix \ref{app:implement}.

\subsubsection{$\operatorname{Rep}(D_8)$ categorical symmetry breaking}
\label{subsec:example_D8}
There has been a growing interest in exploring generalized symmetries and their symmetry breaking patterns to discover new phases of many-body states. Below, we focus on a particular generalized symmetry known as the $\operatorname{Rep}(D_8)$ categorical symmetry and discuss the preparation of an exotic state associated with it. 

First, consider a qubit chain with an even number of sites $2N$, $\operatorname{Rep}(D_8)$  symmetry is realized by a non-invertible $\mathsf{D}$ operator that implements KW dualities on two sublattices separately, along with two $\mathbb{Z}_2$ symmetry operators $\eta_e$ and $\eta_o$, which are the product of Pauli $X$ on even and odd sublattices respectively~\cite{doherty2009identifying,seifnashri2024cluster}. We introduce the following state, which corresponds to the symmetry-broken phase of $\operatorname{Rep}(D_8)$ symmetry: 

\begin{equation}
\ket{\Psi_{+G}} = \ket{+\cdots +}_e \ket{\text{GHZ}}_o + \ket{\text{GHZ}}_e \ket{+\cdots +}_o,
\label{eq:example_+G}
\end{equation}
where we use sub-indices $e$ and $o$ to denote states on the even and odd sublattice respectively, and $\ket{\text{GHZ}}$ is the Greenberger–Horne–Zeilinger state $\ket{\text{GHZ}} = \frac{1}{\sqrt{2}}\big(\ket{0\cdots 0}+\ket{1\cdots 1}\big)$.

In the following, by constructing an MPS representation of this state, we will study the pushable defects and the associated pushing relations, and find that it can be prepared by a fusion measurement protocol. The preparation of this state has not been discussed in the literature to the best of our knowledge.

First, Since $\ket{\Psi_{+G}}$ is an equal-weight superposition between $\ket{+\cdots +}_e\otimes\ket{\text{GHZ}}_o$ and $\ket{\text{GHZ}}_e\otimes\ket{+\cdots +}_o$, we can construct its MPS tensor by composing that of $\ket{+\cdots +}$ and $\ket{\text{GHZ}}$ separately on different sublattices, with an extra virtual bond used to impose the superposition. The explicit MPS tensor is given by
\beq
    A_{+G} = \quad \begin{tikzpicture}[baseline=(current bounding box),scale=1.5]
        \draw (.1,0) --  (2.1,0);
        \draw (-.1,0) --  (-.8,0);
        \draw (-.8,1.2) -- (1.3,1.2);
        \draw (1.5,1.2) -- (1.7,1.2);
        \draw (1.9,1.2) -- (2.1,1.2);
        \draw[rounded corners=8pt]
        (-.5,0) -- (-.5,.8) -- (1.4,.8) -- (1.4,1.7);
        \fill (-.5,0) circle (1pt);
        \draw[rounded corners=8pt]
        (0,-0.5) -- (0,.4) -- (1.8,.4) -- (1.8,1.7);
        \filldraw[fill = white] (.3,.2) rectangle node {} (1.1,1.4);
        \node at (0,-.7) {$\ket{+}$};
        \node at (.7,1.2) {$C$};
        \node at (.7,.6) {SWAP};
    \end{tikzpicture}\ ,
    \label{eq:+_GHZ}
\eeq
where $C\text{SWAP}$ denotes a three-qubit controlled-SWAP gate with control on the top virtual bond and target on the bottom two bonds. The Pauli $Z$ defects on both the top and bottom virtual bonds push to $Z$ without further action on the physical bonds. The Pauli $X$ on the top virtual bond pushes to $X$ with a SWAP gate on the physical bonds. The Pauli $X$ on the bottom virtual bond pushes to $X$ with an $X\otimes X$ gate on the physical bonds. Therefore, all the Pauli defects are right-pushable, meaning this state is preparable using a fusion measurement protocol. 

As shown above, the pushing relations are contained in the two-dimensional subspaces, thus according to Lemma $\ref{lemma:Pauli_subspace}$, the target state can be prepared by implementing a certain Tambara-Yamagami duality. In particular, since the pushing relations of both subspaces are in $RZRX$ class, the corresponding duality should be $\text{KT}\otimes \text{KT}$. To see that, we first derive the transfer matrix of the MPS tensor:
\beq
    T_{+G} = \frac{1}{4}\left[ 2\ 
    \begin{tikzpicture}[baseline=(current bounding box),scale=1.2]
        \draw (0,0.8) -- (1,.8);
        \draw (0,0) -- (1,.0);
        \draw (0,1.4) -- (1,1.4);
        \draw (0,-.6) -- (1,-.6);
    \end{tikzpicture} \ +\ \begin{tikzpicture}[baseline=(current bounding box),scale=1.2]
        \draw (0,0.8) -- (1,.8);
        \draw (0,0) -- (1,.0);
        \draw (0,1.4) -- (1,1.4);
        \draw (0,-.6) -- (1,-.6);
        \filldraw[draw, circle, fill=myBlue, minimum size=6pt](.5,1.4) circle (6pt);
        \node at (.5,1.4) {$Z$};
        \filldraw[draw, circle, fill=myBlue, minimum size=6pt](.5,-.6) circle (6pt);
        \node at (.5,-.6) {$Z$};
    \end{tikzpicture}\ +\ 
    \begin{tikzpicture}[baseline=(current bounding box),scale=1.2]
        \draw (0,0.8) -- (1,.8);
        \draw (0,0) -- (1,.0);
        \filldraw[draw, circle, fill=myBlue, minimum size=6pt](.5,.8) circle (6pt);
        \node at (.5,.8) {$Z$};
        \filldraw[draw, circle, fill=myBlue, minimum size=6pt](.5,0) circle (6pt);
        \node at (.5,0) {$Z$};
        \draw (0,1.4) -- (1,1.4);
        \draw (0,-.6) -- (1,-.6);
        \filldraw[draw, circle, fill=myBlue, minimum size=6pt](.5,1.4) circle (6pt);
        \node at (.5,1.4) {$Z$};
        \filldraw[draw, circle, fill=myBlue, minimum size=6pt](.5,-.6) circle (6pt);
        \node at (.5,-.6) {$Z$};
    \end{tikzpicture}\
    \right].
\eeq
From this transfer matrix, we can write the representative MPS tensor as in Eq.~\eqref{eq:A_RZRX}, which leads to the following result:
\beq
    \ket{\Psi_{+G}}\otimes\ket{\tilde{+}}^{\otimes 2N}=\Big(\prod_i U_{i}\Big)\text{KT}\otimes \text{KT}\Big(\ket{\psi}\otimes\ket{\tilde{+}\tilde{+}}\Big)^{\otimes N}
    \label{eq:+G}
\eeq
where the two-qubit state $\ket{\psi}$ is defined on each $2i-1$ and $2i$-th site, 
\beq
    \ket{\psi}&=\frac{1}{2}\Big(\sqrt{2}\ket{++}+\ket{-+} +\ket{--}\Big),
\eeq
and $U_i$ is a two-qubit unitary gate defined on the $2i-1$ and $2i$-th site, that maps $X_e^k X_o^l\ket{\psi}$ to $\text{SWAP}^{k+l}\ket{l+}$ for $k,l=0,1$. One ancillary qubit is introduced for each site in the above, denoted with a tilde. By throwing them away, Eq.~\eqref{eq:+G} can be further simplified
\beq
    \ket{\Psi_{+G}}=\Big(\prod_i U_{i}\Big)\frac{1+\prod_i X_{2i}}{2}\frac{1+\prod_i X_{2i+1}}{2}\ket{\psi}^{\otimes N}.
    \label{eq:+G_simplified}
\eeq

\subsection{Class $RI$--$RIRX$: 1 round}
\label{subsec:example_RI-RIRX}

Our third example is a family of states prepared from product states using imaginary time evolution of a product state under a Hamiltonian $H$ composed of nearest-neighbor and next-nearest-neighbor interactions. 

In Appendix $\ref{app:RI}$, we have derived that all states in class $RI$ can be prepared from product states using a product of certain diagonal three-qubit non-unitary gates $\mathcal{E}'$. When the parameters in the transfer matrix $T_A$ in Eq.~\eqref{eq:T_A_RI} and Eq.~\eqref{eq:T_A_RI_app} satisfy $\lambda_1=\lambda_2=0$, the associated tensor $B$ is in class $RIRX$, and as a result the target state $\ket{A}$ is in class $RI$--$RIRX$ and is 1-round MF-preparable. For simplicity of discussion, we choose $\lambda_1=\lambda_2=\lambda_6=\lambda_7=0$, corresponding to the following non-unitary three-qubit gate centered at site $2i$ 
\beq
    \mathcal{E}'_i=CZ_{2i,2i+1} e^{\kappa Z_{2i-1}+\gamma' Z_{2i}Z_{2i+1}+\gamma Z_{2i-1}Z_{2i}Z_{2i+1}},
\eeq
with $\tanh{2\kappa}=-\tanh{2\gamma}\tanh{2\gamma'}$. See Appendix $\ref{app:RI}$ for details. According to Eq.~\eqref{eq:RI_non-unitary}, this family of states is given by
\beq
    \ket{\Psi_{\kappa,\gamma,\gamma'}} = e^{H_{\kappa,\gamma,\gamma'}}\Big(\prod_{i}CZ_{2i,2i+1}\Big) \ket{\psi}^{\otimes N},
\eeq
where $\ket{\psi}$ is a two-qubit state at site $2i-1$ and site $2i$
\beq
    \ket{\psi}=a\ket{00}+b\ket{01}+c\ket{10}+d\ket{11},
\eeq
and the Hamiltonian 
\beq
    H_{\kappa,\gamma,\gamma'}=\sum_{i}\kappa Z_{2i-1}+\gamma' Z_{2i}Z_{2i+1}+\gamma Z_{2i-1}Z_{2i}Z_{2i+1},
\eeq
for the real parameters with constraints
\beq
    a^2+b^2&=c^2+d^2,\\
    \frac{b^2-a^2}{d^2-c^2}&=\frac{\tanh{(2\gamma-2\gamma')}}{\tanh{(2\gamma+2\gamma')}},\\
    \tanh{2\kappa}&= -\tanh{2\gamma}\tanh{2\gamma'}.
\eeq

In what follows, we consider two special subsets of states for which the Hamiltonian is simplified, corresponding to $\gamma'=0$ and $\gamma=0$, respectively. From the above constraints, for the former $\kappa=0$, $a^2=c^2$, and $b^2=d^2$. As a result, these states are given by
\beq
    \ket{\Psi_\gamma} = e^{\gamma\sum_i Z_{2i-1}Z_{2i}Z_{2i+1}}\Big(\prod_{i}CZ_{2i,2i+1}\Big) \ket{\psi}^{\otimes N},
\eeq
where $\ket{\psi}=\ket{+}\otimes\ket{\phi}$ and $\ket{\phi}$ is an arbitrary real single-qubit state. In order to write the MPS tensors, it is convenient to use the following identity
\beq
    \begin{tikzpicture}[baseline=(current bounding box),scale=1.5]
        \draw (-0.3,-.3) -- (-0.3,.3);
        \draw (0.3,-.3) -- (0.3,.3);
        \filldraw[fill = myBlue] (-.4,-.2) rectangle node {$e^{\gamma Z Z}$} (.4,.2);
    \end{tikzpicture}\quad = \quad 
    \begin{tikzpicture}[baseline=(current bounding box),scale=1.5]
        \draw (.5,0) -- (-.5,0);
        \draw (-0.5,-.3) -- (-0.5,.3);
        \draw (0.5,-.3) -- (0.5,.3);
        \fill (.5,0) circle (1pt);
        \fill (-.5,0) circle (1pt);
        \filldraw[draw, circle, fill=white, minimum size=6pt](0,0) circle (6pt);
        \node at (0,0) {$V$};
    \end{tikzpicture}\,
\eeq
where $V = e^{\epsilon X}$ for $\gamma>0$ and $V = X e^{\epsilon X}$ for $\gamma<0$ with $\tanh{\epsilon}=e^{-2|\gamma|}$. The MPS tensor for this type of states, specializing to $\gamma>0$, can thus be written as 
\beq
    \begin{tikzpicture}[baseline=(current bounding box),scale=1.5]
        \draw (-0.5,-.3) -- (-0.5,1.3);
        \draw (0.5,-.3) -- (0.5,1.3);
        \fill (.5,0) circle (1pt);
        \fill (-.5,0) circle (1pt);
        \fill (.5,1) circle (1pt);
        \fill (-.5,1) circle (1pt);
        \draw (.5,1) -- (1.5,1);
        \draw (-.5,1) -- (-.9,1);
        \filldraw[draw, circle, fill=myGray, minimum size=6pt](1,1) circle (6pt);
        \node at (1,1) {$H$};
        \fill (-.5,0.4) circle (1pt);
        \draw (-.5,.4) -- (-.9,.4);
        \filldraw[draw, circle, fill=white, minimum size=6pt](0,0) circle (.1);
        \draw (.6,0.4) -- (1.5,.4);
        \draw[rounded corners=8pt]
        (0,-.1) -- (0,0.4) -- (0.4,.4);
        \filldraw[draw, circle, fill=myRed, minimum size=6pt](1,0.4) circle (6pt);
        \node at (1,0.4) {$e^{\epsilon X}$};
        \draw (.5,0) -- (-.5,0);
        \node at (-0.5,-.5) {$\ket{+}$};
        \node at (0.5,-.5) {$\ket{\phi}$};
    \end{tikzpicture}\ .
\eeq
After this rewriting, it turns out that all Pauli defects are right-pushable. Therefore, this subset of states is fusion measurement preparable.

For the other simplified subset of states corresponding to $\gamma=0$,  we have $\kappa=0$, $a^2=d^2$, and $b^2=c^2$. These states are given by 
\beq
    \ket{\Psi_{\gamma'}} = e^{\gamma'\sum_i Z_{2i}Z_{2i+1}}\Big(\prod_{i}CZ_{2i,2i+1}\Big) \ket{\psi}^{\otimes N},
\eeq
where $\ket{\psi}=(I+X\otimes X)\ket{0}\otimes\ket{\phi}$ and $\ket{\phi}$ is an arbitrary real single-qubit state. We can write the MPS tensor as
\beq
    \begin{tikzpicture}[baseline=(current bounding box),scale=1.5]
        \draw (-0.2,0.1) -- (-0.2,1.3);
        \draw (0.1,0.1) -- (0.1,1.3);
        \fill (.1,1) circle (1pt);
        \fill (-.2,1) circle (1pt);
        \draw (.1,1) -- (1,1);
        \draw (-.2,1) -- (-.6,1);
        \filldraw[draw, circle, fill=myGray, minimum size=6pt](.55,1) circle (6pt);
        \node at (.55,1) {$H$};
        \fill (-.2,0.4) circle (1pt);
        \fill (.1,0.4) circle (1pt);
        \draw (-.2,.4) -- (-.6,.4);
        \draw (.1,0.4) -- (1,.4);
        \filldraw[draw, circle, fill=myRed, minimum size=6pt](.55,0.4) circle (6pt);
        \node at (.55,0.4) {$e^{\epsilon X}$};
        \node at (-.05,-.1) {$\ket{\psi}$};
    \end{tikzpicture}\ ,
\eeq
where $\ket{\psi}$ satisfies $\ket{\psi}=X\otimes X\ket{\psi}$. The set of right-pushable defects $\{Z\otimes I, I\otimes Z, X\otimes X\}$ does not form a complete basis. However, the defect $X\otimes I$ is left-pushable, which makes up a complete basis. Therefore, states of this type are also fusion measurement preparable. 

Both types of states are injective and respect the time-reversal symmetry along with a $\mathbb{Z}_2$ internal symmetry: $\prod_{i}X_{2i-1}Z_{2i}$ for $\ket{\Psi_\gamma}$ and $\prod_i Y_i$ for $\ket{\Psi_{\gamma'}}$. By studying symmetry fractionalization, it can be shown that states $\ket{\Psi_{\gamma}}$ and $\ket{\Psi_{\gamma'}}$ are in a trivial and non-trivial $\mathbb{Z}_2\times \mathbb{Z}_2^{\text T}$ SPT phase, respectively. Moreover, if we apply a transversal circuit $\prod_i CZ_{2i-1,2i}$ to $\ket{\Psi_{\gamma'}}$, the resultant state respects time-reversal symmetry and a $\mathbb{Z}_2$ symmetry generated by $\prod_i X_i$, and is in the same SPT phase as the ground state of the cluster Hamiltonian $H=\sum_i Z_{i-1}X_i Z_{i+1}$.

The fact that many states in this class are fusion measurement preparable is not a coincidence. As discussed in Sec.~\ref{subsec:RI-RIRX}, any state in class $RI$--$RIRX$ can be prepared using fusion measurement with non-transversal feedback corrections.

\subsection{Class $RZ$--$\text{FDLU}$: 1 round}
\label{subsec:example_RZ-RIRX}

As discussed in Sec.~\ref{subsec:RZ-RIRX}, the MPS tensor in class $RZ$--$RIRX$ are generally of the form 
\beq
        A=\ \begin{tikzpicture}[baseline=(current bounding box),scale=1.5]
        \draw (-.3,0.4) -- (-.1,.4);
        \draw (.1,0.4) -- (2.4,.4);
        \draw (2.6,0.4) -- (4.8,.4);
        \draw[rounded corners=8pt]
        (0,-.4) -- (0,.8) -- (2.4,.8) -- (2.4, 1.2) -- (4.4, 1.2) -- (4.4, 1.5);
        \draw[rounded corners=8pt]
        (.3,-.4) -- (0.3, 0) -- (2.5,0) -- (2.5,.8) -- (4.7,.8) -- (4.7, 1.5);
        \filldraw[fill = myRed] (.3,.2) rectangle node {} (.8,1); 
        \node at (.55,.8) {$C$};
        \node at (.55,.4) {$X$};
        \filldraw[draw, circle, fill=myBlue, minimum size=6pt](1.2,0.4) circle (6pt);
        \node at (1.2,0.4) {$R^z_{\alpha_{1}}$};
        \filldraw[fill = myBlue] (1.6,.2) rectangle node {} (2.2,1); 
        \node at (1.9,.8) {$C$};
        \node at (1.9,.4) {$R^z_{\alpha_{2}}$};
        \filldraw[fill = myBlue] (2.8,.2) rectangle node {} (3.4,1); 
        \node at (3.1,.8) {$C$};
        \node at (3.1,.4) {$R^z_{\alpha_{3}}$};
        \filldraw[fill = myBlue] (3.6,.2) rectangle node {} (4.2,1.4); 
        \node at (3.9,1.2) {$C$};
        \node at (3.9,.8) {$C$};
        \node at (3.9,.4) {$R^z_{\alpha_{4}}$};
        \node at (0.15,-.6) {$\ket{\psi}$};
        \end{tikzpicture}
\eeq 
for arbitrary angles $\{\alpha_i\}$ and an arbitrary two-qubit state $\ket{\psi}$. Here, we discuss a subset of the states with $\ket{\psi}=\ket{+}\otimes\ket{+}$. In this case, the states are in the 1-round MF-preparable class $RZ$--$\text{FDLU}$, as the associated state $\ket{B}$ has a rank-$1$ transfer matrix.

When we further take $\alpha_1=\alpha$ to be the only non-zero angle, the qubits on the even sublattice are disentangled. Discarding these qubits, we can write the state as
\beq
    \ket{\Psi} =\sum_{\{\sigma_j\} } \text{Tr}\left[\prod_j X^{\sigma_j}e^{i\frac{\alpha}{2} Z} \right] \ket{\{\sigma_j\}},
\eeq
with $\sigma_i=0,1$. When $\alpha=0$, this state is simply the $GHZ$ state. For generic values of $\alpha$, the sequential circuit in this case is the cosine symmetry conjugated by Hadamard, see Sec.~\ref{subsec:RZ-RIRX} and Appendix \ref{app:cosine}. As a result, this state is given by
\beq
    \ket{\Psi}&=\Big(\prod_j H_j\Big)\text{KW}^{\dagger}  e^{\sum_j\frac{i\alpha Z_j}{2}} \text{KW}\Big(\prod_j H_j\Big)\ket{+}^{\otimes N}\\
    &=\Big(\prod_j H_j\Big)\text{KW}^{\dagger}  e^{\sum_j\frac{i\alpha Z_j}{2}}\ket{+}^{\otimes N}.
    \label{eq:121}
\eeq
When a state is given by KW acting on a product state, its MPS tensor has all Pauli defects pushable, and therefore the state in Eq.~\eqref{eq:121} is preparable using fusion measurement. When $\alpha=0$ or $\pi$, this state is essentially some GHZ state. For other values of $\alpha$, this state is injective. 

Alternatively, we take $\alpha_4=\alpha$ to be the only non-zero angle, and $\ket{\psi}=\ket{+}\otimes\ket{+}$. The target state becomes 
\beq
    \ket{\Psi'} =\sum_{\{\sigma_j\} } \text{Tr}\left[\prod_j X^{\sigma_{2j-1}}(e^{i\frac{\alpha}{2} Z})^{\sigma_{2j-1}\sigma_{2j}} \right] \ket{\{\sigma_{j}\}}.
\eeq
This state seems to be beyond fusion measurement, as we cannot find a complete orthonormal set of pushable defects.

\subsection{Class $RIRX$--$RZRX$: 2 rounds}
\label{subsec:RIRX-RZRX}
\subsubsection{$\mathbb{Z}_2$ non-onsite symmetry breaking}

On a 1d qubit chain with periodic boundary conditions, we consider the following state
\begin{equation}
\ket{\Psi_{+c}} = \ket{+}^{\otimes N}+\ket{\text{cluster}},
\end{equation}
where $\ket{\text{cluster}}$ is the 1d cluster state defined as $\prod_i CZ_{i,i+1}  \ket{+}^{\otimes N}$. $\ket{\Psi_{+c}}$ possesses a $\mathbb{Z}_2$ intrinsically non-onsite symmetry $\prod_i CZ_{i,i+1}$, which is spontaneously broken and leads to long-range correlations of non-invertible charged operators, as discussed in Ref.~\cite{zhang2024long}.\footnote{This symmetry is intrinsically non-onsite because it cannot be brought onsite via any finite-depth local unitary circuit, in the absence of ancilla degree of freedoms.} Such a symmetry-breaking pattern can arise in the two-dimensional ground-state subspace of a gapped local Hamiltonian that exhibits the coexistence between two distinct SPT orders protected by the $\mathbb{Z}_2\times \mathbb{Z}_2$ symmetries. These symmetries are generated by $\prod_{i\in \text{odd}} X_i$ and $\prod_{i\in \text{even}} X_i$, with the representative fixed-point states being $\ket{+}^{\otimes N}$ and $\ket{\text{cluster}}$. Specifically, Ref.~\cite{zhang2024long} constructs such a parent Hamiltonian: 

\begin{equation}
\begin{split}
      H&=  \sum_i  (1-X_i)  (1- Z_{i+1}X_{i+2}Z_{i+3})\\
      &+\sum_i(1- Z_{i}X_{i+1}Z_{i+2}) (1-X_{i+3}).
      \end{split}
\end{equation}
There are two exact ground states $\ket{+}^{\otimes N}$ and $\ket{\text{cluster}}$ even for finite system sizes, with a finite-energy gap above the ground-state subspace. These results have been analytically proven and numerically confirmed~\cite{zhang2024long}. 

A 2-round MF preparation of this state is proposed in Ref.~\cite{zhang2024long} using an ancillary GHZ state in the intermediate step. Ref.~\cite{stephen2024preparing} proposes a 1-round split-index MPS (SIMPS) fusion measurement protocol, which uses a non-transversal feedback correction. Here we provide another MF preparation protocol, which enables us to put this state in class $RIRX$--$RZRX$. We take advantage of the fact that this state is related to the example in Eq.~\eqref{eq:example_+G} as follows:
\beq
    \ket{\Psi_{+c}} = \text{KW}_o \Big(\prod_i CX_{2i-1,2i}\Big) \ket{\Psi_{+G}},
    \label{eq:KW_o_CX}
\eeq
where $\text{KW}_o$ is the Kramers-Wannier duality on the odd sublattice. To see this, we first notice that $\ket{\Psi_{+G}}$ is a equal weight superposition of two states, respectively stabilized by $X_{2i}$, $Z_{2i-1}Z_{2i+1}$, and $X_{2i-1}$, $Z_{2i}Z_{2i+2}$. After the transversal $CX$ and $\text{KW}_o$, the stabilizers are mapped to
\beq
    &X_{2i}\rightarrow X_{2i}\rightarrow X_{2i},\\
    &Z_{2i-1}Z_{2i+1}\rightarrow Z_{2i-1}Z_{2i+1}\rightarrow X_{2i+1},\\
    &X_{2i-1}\rightarrow X_{2i-1}X_{2i}\rightarrow Z_{2i-1}X_{2i}Z_{2i+1},\\
    &Z_{2i}Z_{2i+2}\rightarrow Z_{2i-1}Z_{2i}Z_{2i+1}Z_{2i+2}\rightarrow Z_{2i}X_{2i+1}Z_{2i+2},
\eeq
which suggests that $\text{KW}_o \big(\prod_i CX_{2i-1,2i}\big) \ket{\Psi_{+G}}$ is a equal weight superposition  of the $\ket{+}^{\otimes N}$ state, and the cluster state. The MPS tensor for $\ket{\Psi_{+c}}$ is given by:
\beq
    A_{+c} = \quad \begin{tikzpicture}[baseline=(current bounding box),scale=1.5]
        \draw (.1,0) --  (3.8,0);
        \draw (-.1,0) --  (-.8,0);
        \draw (-.8,1.2) -- (1.3,1.2);
        \draw (1.5,1.2) -- (1.7,1.2);
        \draw (1.9,1.2) -- (3.8,1.2);
        \draw[rounded corners=8pt]
        (-.5,0) -- (-.5,.8) -- (1.4,.8) -- (1.4,2.6) -- (3.3,2.6);
        \draw (3.5,2.6) -- (3.8,2.6);
        \fill (-.5,0) circle (1pt);
        \draw[rounded corners=8pt]
        (-.8,3) -- (2.9,3) -- (2.9,3.8);
        \draw[rounded corners=8pt]
        (0,-0.3) -- (0,.4) -- (1.8,.4) -- (1.8,2.2) -- (3.4,2.2) -- (3.4, 3.8);
        \filldraw[fill = white] (.3,.2) rectangle node {} (1.1,1.4);
        \node at (.7,1.2) {$C$};
        \node at (.7,.6) {SWAP};
        \node at (0,-.5) {$\ket{+}$};
        \filldraw[fill = myRed] (1.2,1.5) rectangle node {$CX$} (2,1.9);
        \filldraw[fill = myBlue] (2.1,2.4) rectangle node {} (2.5,3.2);
        \node at (2.3,3) {$C$};
        \node at (2.3,2.6) {$Z$};
        \filldraw[draw, circle, fill=myGray, minimum size=6pt](2.9,2.6) circle (6pt);
        \node at (2.9,2.6) {$H$};
    \end{tikzpicture}\ .
    \label{eq:A_+c}
\eeq
Both Pauli $Z$ and $X$ on the top virtual bond are pushable, corresponding to the pushing relation $RIRX$. After the projection, the associated MPS $\ket{B}$ allows the Pauli defects on the bottom two virtual bonds to be pushable. Similar to the situation in Eq.~\eqref{eq:+_GHZ}, the pushing relation for $\ket{B}$ is $RZRX$. As a result, we see that $\ket{\Psi_{+c}}$ is in class $RIRX$--$RZRX$, which is 2-round MF-preparable. 

We notice that the feedback correction for the fusion measurement preparation of $\ket{\Psi_{+c}}$ uses only on-site gates $X_{2i-1}X_{2i}$ and $\text{SWAP}_{2i-1,2i}$. Under the action of $\text{KW}_o (\prod_i CX_{2i-1,2i})$, these two gates are mapped to $Z_{2i-1}Z_{2i+1}$ and $CZ_{2i,2i-1}CZ_{2i,2i+1}$, which are not supported on a single site anymore. We can thus postpone this feedback correction after implementing the $\text{KW}_o (\prod_i CX_{2i-1,2i})$ in Eq.~\eqref{eq:KW_o_CX}, substituting these two on-site gates in the correction circuit by $Z_{2i-1}Z_{2i+1}$ and $CZ_{2i,2i-1}CZ_{2i,2i+1}$ gates. 

Therefore, the two rounds of measurements can be combined together, with a feedback correction using non-onsite gates. This example is a demonstration that the measurement rounds in MF circuits can be reduced when allowing non-transversal feedback correction.

\section{Discussion and outlook}\label{sec:outlook}

In this work, we have presented a framework for state preparation with local measurements and feedback correction. We have focused on transversal feedback corrections that push error defects associated with unwanted measurement outcomes. Within this framework, we can classify preparable states into distinct classes based on the sets of pushable defects and their pushing relations. It further gives a natural preparation circuit for each class and enables the study of their state structures. Moreover, this approach reveals a connection between certain non-invertible symmetries and the structure of certain classes of preparable states.  

For the classes where all Pauli defects are pushed to other Pauli defects, our preparation scheme gives rise to an MF circuit equivalent to the fusion measurement protocol. We show that these classes of states are related to a class of non-invertible dualities described by the Tambara-Yamagami categories. Furthermore, we have discussed other classes of states that are preparable using 1- and 2-round MF circuits beyond the scope of fusion measurement. For example, the states in the $RI$ class are described by a non-unitary evolution of some initial product states. The states in the $RZ$--$RIRX$ class are related to an MPO that generalizes the non-invertible cosine symmetries. Finally, we have demonstrated that states prepared using MF circuits composed of FDLU and non-invertible symmetries, e.g., KW and KT dualities, also fit naturally into our classification framework.

Our framework naturally motivates several interesting questions and future directions listed below.

\begin{enumerate}
    \item Focusing on the state preparation with one round of MF, in addition to the fusion measurement protocol, we have studied the class $RI$--$RIRX$, $RZ$--$\text{FDLU}$, $RIRI$--$RIRX$, etc. The states we studied either differ from the fusion measurement preparable state via an FDLU, or are preparable via implementing Kramers-Wannier or Kennedy-Tasaki dualities along with FDLU circuits. To classify all 1-round MF-preparable states and see whether they are always fusion measurement preparable with non-transversal feedback corrections, we need to study the general state structure in classes such as $RIRI$--$RIRX$--$RIRI$.

    \item In this work, we have only explored the classes of states with Pauli pushable defects. In general, as long as the pushable defects form a complete orthonormal basis, the target state can still be prepared using fusion measurement. For pushable defects that cannot be brought to Pauli operators by a gauge transformation, the number of required post-selections does not always remain finite as $N\rightarrow\infty$ for states with PBC. Little is known about the preparation of states with non-Pauli pushable defects, especially those forming an incomplete basis of defects. It would be interesting to explore examples of this kind.

    \item For target MPSs defined on OBC, we have established in Theorem \ref{theorem:preparation} that the preparation circuit for an OLMF-preparable target state can always be constructed under our scheme. This result leads to a no-go theorem: an MPS with OBC cannot be prepared with left-conditioned feedback if it does not admit a set of pushable defects. Despite this seemingly powerful result, we caution that the practical utility may be limited due to the following reason. A target state may admit many distinct MPS representations with OBC. Only when we show that any MPS representation with OBC does not allow for pushable defects, can we then prove that such a state cannot be prepared by left-conditioned feedback. This difficulty can be circumvented if there exists an analogue of the fundamental theorem for MPSs with OBC, so that the lack of a set of pushable defects for a given MPS representation immediately implies the same for all possible MPS representations. 

    \item A closely related notion is {\it short-range nonstabilizerness (magic)}, which characterizes whether a state can be related to $\ket{0}^{\otimes N}$ by a Clifford circuit followed by a constant-depth or shallow-depth local circuit~\cite{wei2025long,korbany2025long,li2026explicit}. Using the term in Ref.~\cite{parham2025quantum}, the short-range magic states can be prepared using $\mathsf{A}_1\mathsf{CQ}$, where $\mathsf Q$ denotes a finite-depth circuit with $k$-local gates and $\mathsf C$ denotes any Clifford gate which could be non-local. The number $1$ in the subscript denotes that there is only one round of Clifford. It has been shown that the state preparation power of $r$-round MF-circuits with $k$-local unitary gates is equivalent to the circuit with $r$ rounds of fan-out gate, or equivalent to $\mathsf{A}_r\mathsf{QC\cdots Q}$ with $r$ non-local Clifford gates~\cite{Buhrman2024statepreparation,parham2025quantum}. Therefore, our preparation scheme naturally offers a tool to explore the fan-out complexity or magic hierarchy of a state. Taking $r=1$, it is still an open question whether 1-round MF circuits with only geometrically local gates have equivalent state preparation power as $\mathsf{A}_1\mathsf{QCQ}$ circuits. It would be interesting to find a counter-example, i.e., a state prepared from $\mathsf{A}_1\mathsf{QCQ}$ but is not 1-round MF-preparable. 

    \item The MF-preparation of states in the spatial dimension greater than one is relatively less known. Some examples that can be prepared using fusion measurement protocols are studied~\cite{lu2022measurement,sahay2024finite,guo2026measurement}. It would be interesting to generalize our defect approaches to higher dimensions and construct MF-preparation circuits for tensor network states, especially for the ones that only have an incomplete set of pushable defects.

\end{enumerate}

\begin{acknowledgments} 
    We would like to thank Mikhail Litvinov, Zhi-Yuan Wei, David Stephen, Shuyu Zhang, Zijian Song, Rahul Sahay, Ruben Verresen, Henry Yuen, Natalie Parham, and Tyler Ellison for helpful discussions. T.-C.L acknowledges the support of the RQS postdoctoral fellowship through the National Science Foundation (QLCI grant OMA-2120757). This work was supported by the U.S. National Science Foundation under Grant No. NSF DMR-2316598 (YL and AM). 
\end{acknowledgments}

\bibliography{ref}

\appendix

\onecolumngrid
\newpage

\section{$RIRI$--yyyy class and associated state}
\label{app:RIRI-yyyy}
In this Appendix, we discuss state $\ket{A}$ when it is obtained from a state $\ket{D}$ by a depth-2 brick wall circuit. First, we will discuss the relation between state $\ket{D}$ and the associated state $\ket{B}$. Second, we will discuss how the MPS tensor $A$ has all unitary defects on the top virtual bond right-pushable.

Given the following MPS tensor $A$ :
\beq
    A = \quad \begin{tikzpicture}[baseline=(current bounding box),scale=1.5]
        \draw[rounded corners=8pt]
        (-.7,2.2) -- (.6,2.2) -- (.6,3);
        \draw[rounded corners=8pt]
        (-.7,1.2) -- (2.4,1.2) -- (2.4,2.2) -- (1,2.2) -- (1,3);
        \draw[rounded corners=8pt]
        (3.2,1.2) -- (2.8,1.2) -- (2.8,2.2) -- (3.2,2.2);
        \filldraw[fill = white] (2.2,1.6) rectangle node {$U$} (3,2);
        \filldraw[fill = white] (.4,2.4) rectangle node {$V$} (1.2,2.8);
        \filldraw[fill = white] (2.2,1) rectangle node {$D$} (3,1.4);
    \end{tikzpicture}\quad ,
\eeq
all the unitary defects in the top left virtual bond can be right-pushed to $I$. Therefore, the state $\ket{A}$ is in class $RIRI$--$yyyy$, where $yyyy$ denotes the class of pushing relation for the associated state $\ket{B}$ defined from Eq.~\eqref{eq:B_definition}. The transfer matrix of $\ket{B}$ is given as in Eq.~\eqref{eq:T_B}:
\beq
    T_B = \quad \begin{tikzpicture}[baseline=(current bounding box),scale=1.5]
        \draw[rounded corners=8pt]
        (-.7,2.2) -- (.6,2.2) -- (.6,3);
        \draw[rounded corners=8pt]
        (-.7,1.2) -- (2.4,1.2) -- (2.4,2.2) -- (1,2.2) -- (1,3);
        \draw[rounded corners=8pt]
        (4,1.2) -- (2.8,1.2) -- (2.8,2.2) -- (3.4,2.2) -- (3.4,3.6) -- (2.8,3.6) -- (2.8,4.6) -- (4,4.6);
        \filldraw[fill = white] (2.2,1.6) rectangle node {$U$} (3,2);
        \filldraw[fill = white] (.4,2.4) rectangle node {$V$} (1.2,2.8);
        \filldraw[fill = white] (2.2,1) rectangle node {$D$} (3,1.4);

        \draw[rounded corners=8pt]
        (-.7,3.6) -- (.6,3.6) -- (.6,2.8);
        \draw[rounded corners=8pt]
        (-.7,4.6) -- (2.4,4.6) -- (2.4,3.6) -- (1,3.6) -- (1,2.8);
        \filldraw[fill = white] (2.2,3.8) rectangle node {$\bar{U}$} (3,4.2);
        \filldraw[fill = white] (.4,3) rectangle node {$\bar{V}$} (1.2,3.4);
        \filldraw[fill = white] (2.2,4.4) rectangle node {$\bar{D}$} (3,4.8);

        \draw[rounded corners=8pt]
        (4,3.6) -- (3.6,3.6) -- (3.6,2.2) -- (4,2.2);
    \end{tikzpicture}\quad = \quad \begin{tikzpicture}[baseline=(current bounding box),scale=1.5]
        \draw[rounded corners=8pt]
        (1.2,2.2) -- (2,2.2) -- (2,3.6) -- (1.2,3.6);
        \draw[rounded corners=8pt]
        (1.2,1.2) -- (2.4,1.2) -- (2.4,4.6) -- (1.2,4.6);
        \draw[rounded corners=8pt]
        (4,1.2) -- (2.8,1.2) -- (2.8,3.6) -- (2.8,4.6) -- (4,4.6);
        \filldraw[fill = white] (2.2,1) rectangle node {$D$} (3,1.4);
        \filldraw[fill = white] (2.2,4.4) rectangle node {$\bar{D}$} (3,4.8);

        \draw[rounded corners=8pt]
        (4,3.6) -- (3.2,3.6) -- (3.2,2.2) -- (4,2.2);
    \end{tikzpicture} .
\eeq
In the above, we have used the following identity 
\beq
    \begin{tikzpicture}[baseline=(current bounding box),scale=1.5]
        \draw[rounded corners=8pt]
        (-.4,.8) -- (0,.8) -- (0,0) -- (-.4,0);
        \draw[rounded corners=8pt]
        (.6,.8) -- (.2,.8) -- (.2,0) -- (.6,0);
    \end{tikzpicture}\quad = \frac{1}{D}\sum_{p,q=0}^{D-1} \quad
    \begin{tikzpicture}[baseline=(current bounding box),scale=1.5]
        \draw[rounded corners=8pt]
        (-.7,0) -- (.7,0);
        \draw[rounded corners=8pt]
        (-.7,.8) -- (.7,.8);
        \filldraw[draw, circle, fill=myBlue, minimum size=6pt](-.3,.8) circle (6pt);
        \node at (-0.3,.8) {$\bar{Z}^{p}$};
        \filldraw[draw, circle, fill=myRed, minimum size=6pt](.3,.8) circle (6pt);
        \node at (.3,.8) {$\bar{X}^{q}$};
        \filldraw[draw, circle, fill=myBlue, minimum size=6pt](-.3,0) circle (6pt);
        \node at (-0.3,0) {$Z^{p}$};
        \filldraw[draw, circle, fill=myRed, minimum size=6pt](.3,0) circle (6pt);
        \node at (.3,0) {$X^{q}$};
    \end{tikzpicture}\ .
\eeq

Therefore, the representative tensor $B$ is simply given by
\beq
\begin{tikzpicture}[baseline=(current bounding box),scale=1.5]
        \draw[rounded corners=8pt]
        (1.2,1.2) -- (2.4,1.2) -- (2.4,2.2);
        \draw[rounded corners=8pt]
        (4,1.2) -- (2.8,1.2) -- (2.8,2.2);
        \draw[rounded corners=8pt]
        (1.2,1.4) -- (2.2,1.4) -- (2.2,2.2);
        \draw[rounded corners=8pt]
        (4,1.4) -- (3,1.4) -- (3,2.2);
        \filldraw[fill = white] (2,1) rectangle node {$B$} (3.2,1.6);
    \end{tikzpicture}\quad  = \quad \begin{tikzpicture}[baseline=(current bounding box),scale=1.5]
        \draw[rounded corners=8pt]
        (1.2,1.6) -- (1.8,1.6) -- (1.8,2.2);
        \draw[rounded corners=8pt]
        (4,1.6) -- (3.4,1.6) -- (3.4,2.2);
        \draw[rounded corners=8pt]
        (1.2,1.2) -- (2.4,1.2) -- (2.4,2.2);
        \draw[rounded corners=8pt]
        (4,1.2) -- (2.8,1.2) -- (2.8,2.2);
        \filldraw[fill = white] (2,1) rectangle node {$D$} (3.2,1.6);
    \end{tikzpicture}\quad .
\eeq
In other words, the associated state $\ket{B}$ for state $\ket{A}$ in class $RIRI$--$yyyy$ is equivalent to the state generated by the tensor $D$, obtained from a projection onto the complementary virtual space, up to the stacking of another product state.

If we assume that the state $\ket{D}$ is preparable using a fusion measurement protocol, i.e., tensor $D$ has a complete set of pushable defects $\{V_t\}$, then we claim that the state $\ket{A}$ can be prepared using a fusion measurement protocol with non-transversal feedback corrections. To see this, we first note that any unitary defect in the top virtual bond is right-pushable to identity due to the following equation,
\beq
    \begin{tikzpicture}[baseline=(current bounding box),scale=1.5]
        \draw[rounded corners=8pt]
        (-.7,2.2) -- (.6,2.2) -- (.6,3);
        \draw[rounded corners=8pt]
        (-.7,1.2) -- (2.4,1.2) -- (2.4,2.2) -- (1,2.2) -- (1,3);
        \draw[rounded corners=8pt]
        (3.2,1.2) -- (2.8,1.2) -- (2.8,2.2) -- (3.2,2.2);
        \filldraw[fill = white] (2.2,1.6) rectangle node {$U$} (3,2);
        \filldraw[fill = white] (.4,2.4) rectangle node {$V$} (1.2,2.8);
        \filldraw[fill = white] (2.2,1) rectangle node {$D$} (3,1.4);
        \filldraw[draw, circle, fill=white, minimum size=6pt](0,2.2) circle (6pt);
        \node at (0,2.2) {$W$};
        \node at (3.6,1.6) {$=$};
        \draw[rounded corners=8pt]
        (-.7+4.7,2.2) -- (.6+4.7,2.2) -- (.6+4.7,3.6);
        \draw[rounded corners=8pt]
        (-.7+4.7,1.2) -- (2.4+4.7,1.2) -- (2.4+4.7,2.2) -- (1+4.7,2.2) -- (1+4.7,3.6);
        \draw[rounded corners=8pt]
        (3.2+4.7,1.2) -- (2.8+4.7,1.2) -- (2.8+4.7,2.2) -- (3.2+4.7,2.2);
        \filldraw[fill = white] (2.2+4.7,1.6) rectangle node {$U$} (3+4.7,2);
        \filldraw[fill = white] (.4+4.7,2.4) rectangle node {$V$} (1.2+4.7,2.8);
        \filldraw[fill = white] (.3+4.7,3) rectangle node {$VW^T V^{\dagger}$} (1.3+4.7,3.4);
        \filldraw[fill = white] (2.2+4.7,1) rectangle node {$D$} (3+4.7,1.4);
    \end{tikzpicture}\quad .
\eeq
Furthermore, a unitary $U_t$ on the physical bond of tensor $D$ that pushes $V_t$ on the virtual bond of $D$ is conjugated by the brick wall circuit into a two-site unitary $U'_t$:
\beq
    \begin{tikzpicture}[baseline=(current bounding box),scale=1.5]
        \draw[rounded corners=8pt]
        (.1,2.2) -- (.6,2.2) -- (.6,3);
        \draw[rounded corners=8pt]
        (2.4,.8) -- (2.4,2.2) -- (1,2.2) -- (1,3);
        \draw[rounded corners=8pt]
        (2.8,.8) -- (2.8,2.2) -- (3.6,2.2) -- (3.6,3);
        \draw[rounded corners=8pt]
        (4.4,2.2) -- (4,2.2) -- (4,3);
        \filldraw[fill = white] (2.2,1.6) rectangle node {$U$} (3,2);
        \filldraw[fill = white] (.4,2.4) rectangle node {$V$} (1.2,2.8);
        \filldraw[fill = white] (3.4,2.4) rectangle node {$V$} (4.2,2.8);
        \filldraw[fill = white] (2.2,1) rectangle node {$U_t$} (3,1.4);
    \end{tikzpicture}\quad =\quad
    \begin{tikzpicture}[baseline=(current bounding box),scale=1.5]
        \draw[rounded corners=8pt]
        (.1+5.2,2.2) -- (.6+5.2,2.2) -- (.6+5.2,3.6);
        \draw[rounded corners=8pt]
        (2.4+5.2,1.4) -- (2.4+5.2,2.2) -- (1+5.2,2.2) -- (1+5.2,3.6);
        \draw[rounded corners=8pt]
        (2.8+5.2,1.4) -- (2.8+5.2,2.2) -- (3.6+5.2,2.2) -- (3.6+5.2,3.6);
        \draw[rounded corners=8pt]
        (4.4+5.2,2.2) -- (4+5.2,2.2) -- (4+5.2,3.6);
        \filldraw[fill = white] (2.2+5.2,1.6) rectangle node {$U$} (3+5.2,2);
        \filldraw[fill = white] (.4+5.2,2.4) rectangle node {$V$} (1.2+5.2,2.8);
        \filldraw[fill = white] (3.4+5.2,2.4) rectangle node {$V$} (4.2+5.2,2.8);
        \filldraw[fill = white] (.4+5.2,3) rectangle node {$U'_t$} (4.2+5.2,3.4);
    \end{tikzpicture}\quad .
\eeq
Therefore, the state $\ket{D}$ is fusion measurement preparable, as long as we use unitary operators $\{VW^TV^{\dagger}\}$ and $\{U'_t\}$ for the feedback corrections,  which are generally non-transversal.

\section{Success probability}
\label{app:success}
In this Appendix, we show that the success probability of preparing a target state $\ket{A}$ with PBC (periodic boundary condition)  from the associated state $\ket{B}$ is constant, as long as the pushable defects are Pauli and are also pushed to Pauli defects.

Recall that the local measurements are on the ancillary physical bonds of $B$ tensors in the computational basis. Without loss of generality, we assume that at each site there are four measurement outcomes $t=0,1,2,3$ respectively corresponding to the $V_t=I,X,Y,Z$ as follows:
\begin{equation}
\begin{aligned}
\begin{tikzpicture}[baseline=(current bounding box),scale=1.5]
        \draw (2,.3) -- (2,.7);
        \filldraw[fill = white] (-.3,-.3) rectangle node {$B$} (.3,.3);
        \draw (0.7,0) --  (0.3,0);
        \draw (-0.3,0) --  (-0.7,0);
        \node at (1,0) {$=$};
        \draw (-.2,.3) -- (-.2,.7);
        \draw (+.2,.3) -- (+.2,.7);
        \node at (+.2,.9) {$t$};
        \filldraw[fill = white] (1.7,-.3) rectangle node {$A$} (2.3,.3);
        \draw (3.5,0) --  (2.3,0);
        \draw (1.7,0) --  (1.3,0);
        \filldraw[draw, circle, fill=white, minimum size=6pt](2.9,0) circle (6pt);
        \node at (2.9,0) {$V_t$};
        \node at (3.6,0) {.};
\end{tikzpicture}
\end{aligned}  
\end{equation}  

First, we claim that the probability of outcome $t_i$ at site $i$ is independent of the outcome $t_j$ at site $j$. Let us take $j=i+1$ for demonstration. After measuring site $i$ and obtaining outcome $t_i$ and before measuring site $i+1$, the state locally at these two sites is the following:
\begin{equation}
\begin{aligned}
\begin{tikzpicture}[baseline=(current bounding box),scale=1.5]
        \draw (2,.3) -- (2,.7);
        \filldraw[fill = white] (1.7,-.3) rectangle node {$A$} (2.3,.3);
        \draw (4.5,0) --  (2.3,0);
        \draw (1.7,0) --  (1.3,0);
        \filldraw[draw, circle, fill=white, minimum size=6pt](2.9,0) circle (6pt);
        \node at (2.9,0) {$V_{t_i}$};
        \filldraw[fill = white] (3.5,-.3) rectangle node {$B$} (4.1,.3);
        \draw (3.6,.3) -- (3.6,.7);
        \draw (4,.3) -- (4,.7);
        \node at (4.6,0) {.};
\end{tikzpicture}
\end{aligned}  
\end{equation}
Assuming defect $V_{t_i}$ to be right-pushable to $V'_{t_i}$ by a physical unitary $U_{t_i}$, we have
\begin{equation}
\begin{aligned}
\begin{tikzpicture}[baseline=(current bounding box),scale=1.5]
        \draw (2,.3) -- (2,.7);
        \filldraw[fill = white] (1.7,-.3) rectangle node {$A$} (2.3,.3);
        \draw (4.5,0) --  (2.3,0);
        \draw (1.7,0) --  (1.3,0);
        \filldraw[draw, circle, fill=white, minimum size=6pt](2.9,0) circle (6pt);
        \node at (2.9,0) {$V_{t_i}$};
        \filldraw[fill = white] (3.5,-.3) rectangle node {$B$} (4.1,.3);
        \draw (3.6,.3) -- (3.6,.7);
        \draw (4,.3) -- (4,.7);
        \node at (4.8,0) {$=$};
        \draw (5.1,0) -- (8.3,0);
        \filldraw[fill = white] (5.5,-.3) rectangle node {$A$} (6.1,.3);
        \draw (5.8,.3) -- (5.8,1);
        \filldraw[fill = white] (6.4,-.3) rectangle node {$B$} (7.2,.3);
        \draw (6.5,.3) -- (6.5,1);
        \draw (7.1,.3) -- (7.1,1);
        \filldraw[draw, circle, fill=white, minimum size=6pt](7.7,0) circle (6pt);
        \node at (7.7,0) {$V'_{t_i}$};
        \filldraw[draw, circle, fill=white, minimum size=6pt](6.5,.65) circle (6pt);
        \node at (6.5,.65) {$U_{t_i}$};
        \filldraw[draw, circle, fill=white, minimum size=6pt](7.1,.65) circle (6pt);
        \node at (7.1,.65) {$W_{t_i}$};
        \node at (8.4,0) {,};
\end{tikzpicture}
\end{aligned}  
\end{equation}
where $W_{t_i} = \sum_t c_t\ket{t}\bra{t}$ and the factors $c_t$ are given by the commutation between $V_{t_i}'$ and $V_t$. Now we measure the right-most physical bond at site $i+1$ on the computational basis. Since $W_{t_i}$ gate is a phase gate, it commutes with the measurement in the computational basis. Consequently, the outcome $t_i$ does not affect the measurement at site $i+1$, i.e., $p_{t_{i}}$ and $p_{t_{j}}$ are two independent probabilistic distributions. Furthermore, as we consider translation-invariant states, the distributions at different sites are identical, which we denote as $p(t)$.

Then the success probability to prepare $\ket{A}$ corresponds to the probability for the $N$ samples $\{t_i\}$ to yield the defect-less outcome after pushing to a single virtual bond. This probability satisfies $P_N \to \text{const}$ as $N \to \infty$. For example, $P_N \to \frac{1}{4}$ for $D = 2$ with pushing relations $X \mapsto X$, $Z \mapsto Z$ when all four outcomes occur.

\section{Coarse-graining pushing relations}
\label{app:coarse_grain}
In this Appendix, we discuss the pushing relations when all the generalized Pauli defects in the $D$-dimensional virtual space are pushable to other Pauli defects. As mentioned in Sec.~\ref{subsec:right_pushable}, the generalized Pauli-$X$ and $Z$ matrices generate the Weyl-Heisenberg group modulo $D$. The generalized Pauli defects form an Abelian $G=\mathbb{Z}_D\oplus \mathbb{Z}_D$ group since they are defined up to a phase. We use a pair $(a,b)=``Z^a X^b"$ for $a,b=0,\cdots, D-1$ to denote $G$ group elements. 

\subsubsection{Right-pushing}
Let us first consider the cases when all the generalized Pauli defects are right-pushable.
A right-pushing relation for these defects defines an endomorphism $\Phi$ generated by $X\mapsto P_X$ and $Z\mapsto P_Z$:
\beq
    \Phi: (1,0) \mapsto (a_x, b_x),\quad (0,1)\mapsto (a_z, b_z).
\eeq

When iterating this endomorphism, the kernels form a sequence satisfying $\operatorname{ker}(\Phi^n)\subseteq \operatorname{ker}(\Phi^{n+1})$. After a finite number of iterations, this sequence eventually stabilizes to a subgroup of $G$ denoted by $K_{\Phi}$. By definition, any element in $K_{\Phi}$ will be pushed to $(0,0)$, any element not in $K_{\Phi}$ will never be pushed to $(0,0)$, i.e., the identity defect. When $K_{\Phi} = G$, there is a finite number $p$ such that $\Phi^p$ maps all elements to the identity. Correspondingly, after a blocking of $p$ sites, the transfer matrix has pushing relations $Z\mapsto I$ and $X\mapsto I$. When the $K_{\Phi} = \{e\}$, the map $\Phi$ is an isomorphism of $G$. Therefore, there is another number $p'$ such that $\Phi^{p'}$ is the identity map on $G$. Correspondingly, after a blocking of $p'$ sites, the transfer matrix has pushing relations $Z\mapsto Z$ and $X\mapsto X$. When $K_{\Phi} = \mathbb{Z}_D = \langle (1,0) \rangle$, $\Phi$ is an isomorphism of $\mathbb{Z}_D = \langle (0,1) \rangle$. There is a number $p''$ such that $\Phi^{p''}$ maps $(a,b)$ to $(0,b)$. Correspondingly, after a blocking of $p''$ sites, the transfer matrix has pushing relations $Z\mapsto I$ and $X\mapsto X$.

In general, we claim that $G = K_{\Phi}\oplus I_{\Phi}$ for some subgroup $I_{\Phi}$. To show this, we pick an integer $n$ such that $K_{\Phi}=\operatorname{ker}(\Phi^n) = \operatorname{ker}(\Phi^{n+1})$. Then we have $|\operatorname{Im}(\Phi^n)|=|\operatorname{Im}(\Phi^{n+1})|$, and therefore $\operatorname{Im}(\Phi^n)=\operatorname{Im}(\Phi^{n+1})$ because the latter is by definition contained in the former. We denote it by $I_{\Phi}$. The map $\Phi^n$ defines a bijection on $I_{\Phi}$. On the one hand, $I_{\Phi}\cap K_{\Phi} = \{(0,0)\}$ because if there is any other element in this overlap, then the kernel of $\Phi^n$ is not stabilized. On the other hand, any element $(a,b)\in G$ is given by the sum of an element in $I_{\Phi}$ and another element in $K_{\Phi}$. We can construct this sum as follows. Take $\Phi^{2n}((a,b)) = (a',b')\in I_{\Phi}$, then there exists another element $(a'',b'')\in I_{\Phi}$ such that $\Phi^n(a'',b'') = (a',b')$ since $\Phi^n$ is a bijection in $I_{\Phi}$. Then the decomposition is given by $(a,b) = (a'',b'')+ (a-a'',b-b'')$, where $(a-a'',b-b'')$ is in the kernel of $\Phi^n$. From the above results, we conclude that $G = K_{\Phi}\oplus I_{\Phi}$.

Now we know that $\Phi^n$ maps the elements in $K_{\Phi}$ to identity, and is an isomorphism on $I_{\Phi}$. Since $I_{\Phi}$ is a finite group, there exists a finite integer $m$, such that $\Phi^{mn}$ maps the elements in $K_{\Phi}$ to identity, and is an identity map on $I_{\Phi}$. In the remainder of our argument, we will use the standard decomposition of an Abelian group $\mathbb{Z}_D$ into a direct sum of cyclic groups
\beq
    \mathbb{Z}_D = \bigoplus_i \mathbb{Z}_{p_i} =\bigoplus_i \langle n_i \rangle,
    \label{eq:D_decompose}
\eeq
where different $p_i$ are some powers of different prime numbers. For example, when $D=12$ we can decompose $\mathbb{Z}_{12}$ as a direct sum of $\mathbb{Z}_4$ and $\mathbb{Z}_3$, with generators $n_1=3$ and $n_2=4$. Accordingly, we have the following decomposition:
\beq
    G = \mathbb{Z}_D\oplus \mathbb{Z}_D = \bigoplus_{i} \Big(\mathbb{Z}_{p_i}\oplus \mathbb{Z}_{p_i}\Big) = \bigoplus_{i} \Big(\langle (n_i, 0)\oplus\langle(0,n_i)\rangle\Big) .
    \label{eq:G_decompose}
\eeq
In the language of MPS, we are decomposing the $D$-dimensional virtual subspace as in Eq.~\eqref{eq:D_decompose}, such that each subspace is of dimension $p_i$. The subgroup $\mathbb{Z}_{p_i}\oplus\mathbb{Z}_{p_i}$ in Eq.~\eqref{eq:G_decompose} is formed by the Pauli $Z$ and $X$ defects in the $i$-th subspace.

Both $K_{\Phi}$ and $I_{\Phi}$ can also be decomposed into direct sums of cyclic subgroups of prime power orders $\{p_i\}$. When the two cyclic subgroups $\mathbb{Z}_{p_i}$ are both in $K_{\Phi}$ or $I_{\Phi}$, all the Pauli defects in the $i$-th subspace are right-pushed to $I$ or themselves, respectively. Since $G = K_{\Phi}\oplus I_{\Phi}$, 
we also allow a generic decomposition of $\mathbb{Z}_{p_i}\oplus\mathbb{Z}_{p_i}$ for an index $i$, such that one generator is in $K_{\Phi}$ and the other generator is in $I_{\Phi}$. In this case, we can always perform a gauge transformation on the $i$-th virtual subspace, to map the Pauli defect corresponding to the first generator to Pauli $Z$. Then $Z$ is right-pushed to $I$, and consequently $X$ must be right-pushed to itself in this case.

Therefore, we conclude that given a general right-pushing relation for all Pauli defects in $D$-dimensional virtual space, we can decompose the virtual space into three subspaces such that after blocking a finite number of sites and appropriate gauge transformation: in the first subspace we have $Z\mapsto I$ and $X\mapsto I$; in the second subspace we have $Z\mapsto I$ and $X\mapsto X$; in the third subspace we have $Z\mapsto Z$ and $X\mapsto X$. 

\subsubsection{Left and right-pushing}
Now we consider the pushing relations when all Pauli defects are either left or right-pushable to other Pauli defects. We denote by $L \subset G$ the set of Pauli defects that are left-pushable, and denote by $R\subset G$ the set of Pauli defects that are right-pushable. We require that $L\cup R = G$. 

According to Definition \ref{def:pushable}, the left-pushing relation defines an endomorphism $\Phi^L$ on $L$, and the right-pushing relation defines an endomorphism $\Phi^R$ on $R$. There is a stabilized kernel $K_L\subseteq L$ such that all defects in $K_L$ are eventually left-pushed to identity, and there is a stabilized kernel $K_R\subseteq R$ such that all defects in $K_R$ are eventually right-pushed to identity. Following a similar argument as before, there are also stabilized images $I_L$ and $I_R$ such that
\beq
    L = K_L \oplus I_L,\quad R= K_R \oplus I_R.
\eeq
There is a finite number $n$ such that after a blocking of $n$ sites, the defects in $I_L$ are left-pushed to themselves, and the defects in $I_R$ are right-pushed to themselves. By their definitions, the images are identical, i.e. $I_L=I_R$.

Once again, we decompose the virtual space into subspaces of dimension $p_i$, where $p_i$ is a power of some prime. Accordingly, the group of generalized Pauli defects has the decomposition as in Eq.~\eqref{eq:G_decompose}. For each subspace labeled by index $i$, the subgroup $\mathbb{Z}_{p_i}\oplus\mathbb{Z}_{p_i}$ has two generators. When both generators are in $R$, all generalized Pauli defects in this subspace are right-pushable, which is discussed in the previous subsection. The cases when both generators are in $L$ are simply the spatial inversion of the right-pushing cases. The only case when there are intrinsically left- and right-pushing relations is therefore when one generator is in $K_L$ and the other generator is in $K_R$. In this case, we can perform a gauge transformation on this $i$-th subspace such that $Z\mapsto I$ and $I\mapsfrom X$ after coarse-graining.

\section{Implementing KW and KT}
\label{app:implement}

As mentioned in Sec.~\ref{subsec:left_right_pushable}, both the generalized Kramers-Wannier duality (KW) and Kennedy-Tasaki duality (KT) can be implemented using a single round of MF as they are Clifford sequential circuits~\cite{chen2023sequential,zhang2024characterizing}. In the following, we will present their MF implementations. 

We first note the following identity up to overall rescaling factors,
\beq
    \begin{tikzpicture}[baseline=(current bounding box),scale=1.5]
        \draw (-.8,0.4) -- (.8,.4);
        \draw (0,0.8) -- (0,0.4);
        \fill (0,0.4) circle (1pt);
        \node at (0,1) {$\bra{+}$};
        \node at (1,0.4) {$=$};
        \draw (2-.8,0.4) -- (2.8,.4);
        \node at (3,0.4) {$=$};
        \draw (4-.8,0.4) -- (4.8,.4);
        \draw (4,0) -- (4,0.4);
        \fill (4,0.4) circle (1pt);
        \node at (4,-.2) {$\ket{+}$};
    \end{tikzpicture}\ .
\eeq
When we view the vertical upward direction as the time direction, a $\bra{+}$ on top of a bond means a projection onto $\ket{+}$, a $\ket{+}$ below a bond means an initial state at $\ket{+}$. Using this identity, we can rewrite the KW MPO tensor in Eq.~\eqref{eq:KW} as follows:
\beq
        \begin{tikzpicture}[baseline=(current bounding box),scale=1.5]
        \draw[rounded corners=8pt]
        (-.4,.8) -- (1.1,.8) -- (1.1,1.3);
        \draw[rounded corners=8pt]
        (0,-.1) -- (0,.4) -- (1.7,.4);
        \filldraw[fill = myBlue] (.3,.2) rectangle node {} (.8,1);
        \node at (.55,.8) {$C$};
        \node at (.55,.4) {$Z$};
        \filldraw[draw, circle, fill=myGray, minimum size=6pt](1.2,.4) circle (6pt);
        \node at (1.2,.4) {$H^{\dagger}$};
        \end{tikzpicture}\quad =\quad
        \begin{tikzpicture}[baseline=(current bounding box),scale=1.5]
        \draw (-.4,.8) -- (-.1,.8);
        \draw (0,-.1) -- (0,1.3);
        \draw (.1,.8) -- (1.1,.8);
        \fill (1.1,.8) circle (1pt);
        \fill (0,.4) circle (1pt);
        \draw (1.1,-.1) -- (1.1,.3);
        \draw (1.1,.5) -- (1.1,1.3);
        \draw (0,.4) -- (2,.4);
        \filldraw[fill = myBlue] (.3,.2) rectangle node {} (.8,1);
        \node at (.55,.8) {$C$};
        \node at (.55,.4) {$Z$};
        \filldraw[draw, circle, fill=myGray, minimum size=6pt](1.6,.4) circle (6pt);
        \node at (1.6,.4) {$H^{\dagger}$};
        \node at (0,1.5) {$\bra{+}$};
        \node at (1.1,-0.3) {$\ket{+}$};
        \end{tikzpicture}\quad=\quad
        \begin{tikzpicture}[baseline=(current bounding box),scale=1.5]
        \draw (-.4,.6) -- (-.1,.6);
        \draw (0,-.1) -- (0,1.3);
        \draw (.1,.6) -- (1.1,.6);
        \fill (1.1,.6) circle (1pt);
        \fill (0,.4) circle (1pt);
        \draw (1.1,-.1) -- (1.1,.3);
        \draw (1.1,.5) -- (1.1,1.3);
        \draw (0,.4) -- (2,.4);
        \filldraw[fill = myBlue] (-.2,.8) rectangle node {} (1.3,1.2);
        \node at (0,1) {$C$};
        \node at (1.1,1) {$Z$};
        \filldraw[draw, circle, fill=myGray, minimum size=6pt](1.6,.4) circle (6pt);
        \node at (1.6,.4) {$H^{\dagger}$};
        \node at (0,1.5) {$\bra{+}$};
        \node at (1.1,-0.3) {$\ket{+}$};
        \end{tikzpicture}\quad=\quad
        \begin{tikzpicture}[baseline=(current bounding box),scale=1.5]
        \draw (-.4,.4) -- (0,.4);
        \draw (0,-.1) -- (0,1.3);
        \fill (1.1,.4) circle (1pt);
        \fill (0,.4) circle (1pt);
        \draw (1.1,-0.1) -- (1.1,1.3);
        \draw (1.1,.4) -- (2,.4);
        \filldraw[fill = myBlue] (-.2,.7) rectangle node {} (1.3,1.1);
        \node at (0,.9) {$C$};
        \node at (1.1,.9) {$Z$};
        \filldraw[draw, circle, fill=myGray, minimum size=6pt](1.6,.4) circle (6pt);
        \node at (1.6,.4) {$H^{\dagger}$};
        \node at (1.1,1.5) {$\bra{+}$};
        \node at (0,-0.3) {$\ket{+}$};
        \end{tikzpicture}\ .
\eeq

We then combine the above identity with the following identity:
\beq
    \begin{tikzpicture}[baseline=(current bounding box),scale=1.5]
        \draw (.5,0) -- (-.5,0);
        \draw (-0.5,-.3) -- (-0.5,.3);
        \draw (0.5,-.3) -- (0.5,.3);
        \fill (.5,0) circle (1pt);
        \fill (-.5,0) circle (1pt);
        \filldraw[draw, circle, fill=myGray, minimum size=6pt](0,0) circle (6pt);
        \node at (0,0) {$H^{\dagger}$};
    \end{tikzpicture}\quad = \quad 
    \begin{tikzpicture}[baseline=(current bounding box),scale=1.5]
        \draw (-0.3,-.3) -- (-0.3,.3);
        \draw (0.3,-.3) -- (0.3,.3);
        \filldraw[fill = myBlue] (-.4,-.2) rectangle node {$CZ$} (.4,.2);
    \end{tikzpicture}\,
\eeq
It is straightforward to see that the generalized KW can be implemented by introducing an ancilla in the $\ket{+}$ state for each site $i$, applying $\prod_i CZ_{i, i'} CZ_{i',i+1}$, and then projecting the original qudit in the $X=+1$ basis. We use the primed indices to denote the ancillary qudits. When the measurement outcome is imperfect, i.e. $X\neq+1$, the resulting tensor differs from the KW MPO by a Pauli $X$ defect, which is right-pushable via applying $X$ on the physical bond. 

Similarly, we have the following identity for the KT MPO tensor:
\beq
        \begin{tikzpicture}[baseline=(current bounding box),scale=1.5]
        \draw (-.4,0.8) -- (1,.8);
        \draw (1.2,0.8) -- (2.7,.8);
        \draw (2.9,0.8) -- (3.8,.8);
        \draw[rounded corners=8pt]
        (0,-.1) -- (0,.4) -- (1.1,.4) -- (1.1, 1.3);
        \draw[rounded corners=8pt]
        (1.7,-.1) -- (1.7,.4) -- (2.8,.4) -- (2.8, 1.3);
        \filldraw[fill = myRed] (.3,.2) rectangle node {} (.8,1);
        \node at (.55,.8) {$C$};
        \node at (.55,.4) {$X$};
        \filldraw[fill = myRed] (2,.2) rectangle node {} (2.5,1);
        \node at (2.25,.8) {$C$};
        \node at (2.25,.4) {$X$};
        \filldraw[draw, circle, fill=myGray, minimum size=6pt](1.6,.8) circle (6pt);
        \node at (1.6,.8) {$H$};
        \filldraw[draw, circle, fill=myGray, minimum size=6pt](3.3,.8) circle (6pt);
        \node at (3.3,.8) {$H^{\dagger}$};
        \end{tikzpicture}\quad=\quad
        \begin{tikzpicture}[baseline=(current bounding box),scale=1.5]
        \draw (-0.4,0.8) -- (-.1,.8);
        \draw (.1,0.8) -- (1,.8);
        \draw (1.2,0.8) -- (1.6,.8);
        \draw (1.8,0.8) -- (2,.8);
        \draw (2.2,0.8) -- (3.1,.8);
        \draw (3.3,0.8) -- (4.2,.8);
        \draw (0,-.1) -- (0,1.3);
        \draw (1.1,-.1) -- (1.1, 1.3);
        \draw (2.1,-.1) -- (2.1, 1.3);
        \draw (3.2,-.1) -- (3.2, 1.3);
        \fill (0,0.4) circle (1pt);
        \fill (1.1,0.4) circle (1pt);
        \draw (0,.4) -- (1.1,.4);
        \fill (2.1,0.4) circle (1pt);
        \fill (3.2,0.4) circle (1pt);
        \draw (2.1,.4) -- (3.2,.4);
        \filldraw[fill = myRed] (.3,.2) rectangle node {} (.8,1);
        \node at (.55,.8) {$C$};
        \node at (.55,.4) {$X$};
        \filldraw[fill = myRed] (2.4,.2) rectangle node {} (2.9,1);
        \node at (2.65,.8) {$C$};
        \node at (2.65,.4) {$X$};
        \filldraw[draw, circle, fill=myGray, minimum size=6pt](1.6,.8) circle (6pt);
        \node at (1.6,.8) {$H$};
        \filldraw[draw, circle, fill=myGray, minimum size=6pt](3.7,.8) circle (6pt);
        \node at (3.7,.8) {$H^{\dagger}$};
        \node at (1.1,-.3) {$\ket{+}$};
        \node at (3.2,-.3) {$\ket{+}$};
        \node at (0,1.5) {$\bra{+}$};
        \node at (2.1,1.5) {$\bra{+}$};
        \end{tikzpicture}\quad .
\eeq
Therefore, the generalized KT can be implemented by introducing two ancillas for each site in the $\ket{+}$ state, applying $\prod_i C(Z_i Z_{i'}Z^{\dagger}_{i+1} Z^{\dagger}_{i'+1},Z_{\tilde{i}}Z_{\tilde{i}'})$, and then projecting the original qudits in the $X=+1$ basis. We use the tilde index to denote the second qudit at each site, the primed index to denote the ancillary qudit, and $C(P_1, P_2)$ is the controlled gate in terms of Pauli operators $P_1$ and $P_2$. Imperfect measurement outcomes correspond to the pushable Pauli defects.

\section{$RILI$ class state structure}
\label{app:RILI}
We use $RILI$ to denote that $Z$ is right-pushed to $I$ and $X$ is left-pushed to $I$: $Z\mapsto I, I\mapsfrom X$. The transfer matrix that supports this pushing relation satisfies:
\beq
    \big(\bar{Z}\otimes Z\big) T_A = T_A = T_A \big(\bar{X}\otimes X\big).
\eeq
This constrains the parameters in Eq.~\eqref{eq:transfer_matrix_bell} to $b=c=0$. The general form of the transfer matrix is thus:
\beq
    T_A&=\sum_{a,d=0}^{d-1}\lambda_{a,d}\quad \begin{tikzpicture}[baseline=(current bounding box),scale=1.5]
        \draw[rounded corners=8pt]
        (-.8,.8) -- (0,.8) -- (0,0) -- (-.8,0);
        \draw[rounded corners=8pt]
        (1,.8) -- (.2,.8) -- (.2,0) -- (1,0);
        \filldraw[draw, circle, fill=myRed, minimum size=6pt](0.6,0) circle (6pt);
        \node at (0.6,0) {$X^{d}$};
        \filldraw[draw, circle, fill=myBlue, minimum size=6pt](-0.4,0.8) circle (6pt);
        \node at (-0.4,0.8) {$Z^a$};
        \end{tikzpicture}\quad =\sum_{a,d=0}^{d-1}\lambda_{a,d}\quad \begin{tikzpicture}[baseline=(current bounding box),scale=1.5]
        \draw[rounded corners=8pt]
        (-.8,.8) -- (-.4,.8) -- (-.4,0) -- (-.8,0);
        \draw[rounded corners=8pt]
        (1,.8) -- (.2,.8) -- (.2,0) -- (1,0);
        \filldraw[draw, circle, fill=myGray, minimum size=6pt](0.6,0) circle (6pt);
        \node at (0.6,0) {$H$};
        \filldraw[draw, circle, fill=myGray, minimum size=6pt](0.6,0.8) circle (6pt);
        \node at (0.6,0.8) {$\bar{H}$};
        \filldraw[draw, circle, fill=myBlue, minimum size=6pt](-0.4,0.4) circle (6pt);
        \node at (-0.4,0.4) {$Z^a$};
        \filldraw[draw, circle, fill=myBlue, minimum size=6pt](0.2,0.4) circle (6pt);
        \node at (0.2,0.4) {$Z^d$};
        \end{tikzpicture} \quad = \quad \begin{tikzpicture}[baseline=(current bounding box),scale=1.5]
        \draw (-.8,.8) -- (-.2,.8);
        \draw (.2,.8) -- (1,.8);
        \draw (-.8,0) -- (-.2,0);
        \draw (.2,0) -- (1,0);
        \draw (-0.2,-.3) -- (-0.2,1.1);
        \draw (0.2,-.3) -- (0.2,1.1);
        \fill (.2,0) circle (1pt);
        \fill (-.2,0) circle (1pt);
        \fill (-.2,.8) circle (1pt);
        \fill (.2,.8) circle (1pt);
        \filldraw[draw, circle, fill=myGray, minimum size=6pt](0.6,.8) circle (6pt);
        \node at (0.6,.8) {$\bar{H}$};
        \filldraw[draw, circle, fill=myGray, minimum size=6pt](0.6,0) circle (6pt);
        \node at (0.6,0) {$H$};
        \node at (0,-.55) {$\ket{\psi_s}$};
        \node at (0,1.35) {$\bra{\psi_s}$};
        \end{tikzpicture}\ ,
\eeq
where $\ket{\psi_s}=\sum_{p,q} \sqrt{s_{p,q}}\ket{p}\otimes \ket{q}$ with $s_{p,q}\equiv \sum_{a,d=0}^{d-1}\lambda_{a,d} e^{\frac{2\pi i}{d}(pa+qd)}$. Using the identity
\beq
    \begin{tikzpicture}[baseline=(current bounding box),scale=1.5]
        \draw (.5,0) -- (-.5,0);
        \draw (-0.5,-.3) -- (-0.5,.3);
        \draw (0.5,-.3) -- (0.5,.3);
        \fill (.5,0) circle (1pt);
        \fill (-.5,0) circle (1pt);
        \filldraw[draw, circle, fill=myGray, minimum size=6pt](0,0) circle (6pt);
        \node at (0,0) {$H$};
    \end{tikzpicture}\quad = \quad 
    \begin{tikzpicture}[baseline=(current bounding box),scale=1.5]
        \draw (-0.3,-.3) -- (-0.3,.3);
        \draw (0.3,-.3) -- (0.3,.3);
        \filldraw[fill = myBlue] (-.4,-.2) rectangle node {$CZ^{\dagger}$} (.4,.2);
    \end{tikzpicture}\ ,
\eeq
resulting from the definition of Hadamard gate in Eq.~\eqref{eq:Hadamard} and the definition of controlled-$Z$ gate $CZ=\sum_{p,q}\frac{1}{\sqrt{d}}e^{\frac{2\pi i pq}{d}}\ket{p,q}\bra{p,q}$, we can write a representative MPS as
\beq
    \ket{A} = \bigg(\prod_{i} CZ^{\dagger}_{\tilde{i},i+1}\bigg) \ket{\psi_s}^{\otimes N},
\eeq
which is a product state with a transversal (non-overlapping) unitary circuit.

\section{$RZ$--$RILI$ class state structure}
\label{app:RZ-RILI}

We use $RZ$--$RILI$ to denote the class of states that satisfy $Z\mapsto Z$ for $A$ and $Z\mapsto I$, $I\mapsfrom X$ for $B$. The general form of the transfer matrix is 
\beq
        T_A = \quad 
        &\lambda_0\ \begin{tikzpicture}[baseline=(current bounding box),scale=1.5]
        \draw[rounded corners=8pt]
        (-.8,.8) -- (0,.8) -- (0,0) -- (-.8,0);
        \draw[rounded corners=8pt]
        (1,.8) -- (.2,.8) -- (.2,0) -- (1,0);
        \end{tikzpicture}\ + \lambda_1 \
        \begin{tikzpicture}[baseline=(current bounding box),scale=1.5]
        \filldraw[opacity=0, minimum size=6pt](-.4,0) circle (6pt);
        \draw[rounded corners=8pt]
        (-.8,.8) -- (0,.8) -- (0,0) -- (-.8,0);
        \draw[rounded corners=8pt]
        (1,.8) -- (.2,.8) -- (.2,0) -- (1,0);
        \filldraw[draw, circle, fill=myBlue, minimum size=6pt](-0.4,0.8) circle (6pt);
        \node at (-0.4,0.8) {$Z$};
        \end{tikzpicture}\ +\lambda_4\ \begin{tikzpicture}[baseline=(current bounding box),scale=1.5]
        \filldraw[opacity=0, minimum size=6pt](-.4,0.8) circle (6pt);
        \draw[rounded corners=8pt]
        (-.8,.8) -- (0,.8) -- (0,0) -- (-.8,0);
        \draw[rounded corners=8pt]
        (1,.8) -- (.2,.8) -- (.2,0) -- (1,0);
        \filldraw[draw, circle, fill=myRed, minimum size=6pt](0.6,0) circle (6pt);
        \node at (0.6,0) {$X$};
        \filldraw[draw, circle, fill=myRed, minimum size=6pt](-.4,0) circle (6pt);
        \node at (-0.4,0) {$X$};
        \end{tikzpicture}\\[2ex] 
         +\ & i\lambda_5 \
        \begin{tikzpicture}[baseline=(current bounding box),scale=1.5]
        \draw[rounded corners=8pt]
        (-.8,.8) -- (0,.8) -- (0,0) -- (-.8,0);
        \draw[rounded corners=8pt]
        (1,.8) -- (.2,.8) -- (.2,0) -- (1,0);
        \filldraw[draw, circle, fill=myBlue, minimum size=6pt](-0.4,0.8) circle (6pt);
        \node at (-0.4,0.8) {$Z$};
        \filldraw[draw, circle, fill=myRed, minimum size=6pt](0.6,0) circle (6pt);
        \node at (0.6,0) {$X$};
        \filldraw[draw, circle, fill=myRed, minimum size=6pt](-.4,0) circle (6pt);
        \node at (-0.4,0) {$X$};
        \end{tikzpicture}\ + i\lambda_6 \
        \begin{tikzpicture}[baseline=(current bounding box),scale=1.5]
        \draw[rounded corners=8pt]
        (-.8,.8) -- (0,.8) -- (0,0) -- (-.8,0);
        \draw[rounded corners=8pt]
        (1,.8) -- (.2,.8) -- (.2,0) -- (1,0);
        \filldraw[draw, circle, fill=myBlue, minimum size=6pt](0.6,.8) circle (6pt);
        \node at (0.6,.8) {$Z$};
        \filldraw[draw, circle, fill=myRed, minimum size=6pt](0.6,0) circle (6pt);
        \node at (0.6,0) {$X$};
        \filldraw[draw, circle, fill=myRed, minimum size=6pt](-.4,0) circle (6pt);
        \node at (-0.4,0) {$X$};
        \end{tikzpicture}\ + \lambda_7 \
        \begin{tikzpicture}[baseline=(current bounding box),scale=1.5]
        \draw[rounded corners=8pt]
        (-.8,.8) -- (0,.8) -- (0,0) -- (-.8,0);
        \draw[rounded corners=8pt]
        (1,.8) -- (.2,.8) -- (.2,0) -- (1,0);
        \filldraw[draw, circle, fill=myBlue, minimum size=6pt](0.6,.8) circle (6pt);
        \node at (0.6,.8) {$Z$};
        \filldraw[draw, circle, fill=myBlue, minimum size=6pt](-0.4,0.8) circle (6pt);
        \node at (-0.4,0.8) {$Z$};
        \filldraw[draw, circle, fill=myRed, minimum size=6pt](0.6,0) circle (6pt);
        \node at (0.6,0) {$X$};
        \filldraw[draw, circle, fill=myRed, minimum size=6pt](-0.4,0) circle (6pt);
        \node at (-0.4,0) {$X$};
        \end{tikzpicture}\ .
\eeq
Note that this transfer matrix is related to another transfer matrix $T_{A'}$ in class $RI$--$RIRX$ as follows:
\beq
    \begin{tikzpicture}[baseline=(current bounding box),scale=1.5]
        \filldraw[fill = white] (-.3,-.2) rectangle node {$T_{A}$} (.3,1);
        \draw (0.7,0) --  (0.3,0);
        \draw (-0.3,0) --  (-0.7,0);
        \draw (0.7,.8) --  (0.3,.8);
        \draw (-0.3,.8) --  (-0.7,.8);
    \end{tikzpicture}\ =\ \begin{tikzpicture}[baseline=(current bounding box),scale=1.5]
        \filldraw[draw, circle, fill=white, minimum size=6pt](-0.6,0) circle (.1);
        \draw[rounded corners=8pt]
        (-.6,.1) -- (-.6,-.4) -- (0.6,-.4) -- (.6,0);
        \fill (.6,0) circle (1pt);

        \filldraw[draw, circle, fill=white, minimum size=6pt](-0.6,.8) circle (.1);
        \draw[rounded corners=8pt]
        (-.6,.7) -- (-.6,1.2) -- (0.6,1.2) -- (.6,0.8);
        \fill (.6,0.8) circle (1pt);
        
        \filldraw[fill = white] (-.3,-.2) rectangle node {$T_{A'}$} (.3,1);
        \draw (0.9,0) --  (0.3,0);
        \draw (-0.3,0) --  (-0.9,0);
        \draw (0.9,.8) --  (0.3,.8);
        \draw (-0.3,.8) --  (-0.9,.8);
    \end{tikzpicture}\  .
\eeq
The above mapping will be discussed in more detail later around Eq.~\eqref{eq:T_A_RI_app}. The general state structure of $\ket{A'}$ in class $RI$--$RIRX$ is derived in Eq.~\eqref{eq:RI-RIRX}. Therefore, $\ket{A}$ in class $RZ$--$RILI$ should have the same structure with an extra controlled-$X$ between the right and left virtual bonds,
\beq
    \begin{tikzpicture}[baseline=(current bounding box),scale=1.2]
        \draw (0,.3) --  (0,1.3);
        \filldraw[fill = white] (-.3,-.3) 
        rectangle node {$A$} (.3,.3);
        \draw (.5,.8) -- (.5,1.3);
        \node at (.5,.6) {$\ket{+}$};
        \draw (0.9,0) --  (0.3,0);
        \draw (-0.3,0) --  (-0.9,0);
        \node at (1.2,0) {$\cdots$};
        \node at (-1.2,0) {$\cdots$};
    \end{tikzpicture}\ =\ \begin{tikzpicture}[baseline=(current bounding box),scale=1.2]
        \filldraw[draw, circle, fill=white, minimum size=6pt](-0.6,0) circle (.1);
        \draw[rounded corners=8pt]
        (-.6,.1) -- (-.6,-.4) -- (0.8,-.4) -- (.8,0);
        \fill (.8,0) circle (1pt);

        \draw (0,.3) --  (0,1.3);
        \filldraw[fill = white] (-.3,-.3) rectangle node {$A'$} (.3,.3);
        \draw (.5,.8) -- (.5,1.3);
        \node at (.5,.6) {$\ket{+}$};
        \draw (1,0) --  (0.3,0);
        \draw (-0.3,0) --  (-0.9,0);
        \node at (-1.2,0) {$\cdots$};
        \node at (1.3,0) {$\cdots$};
    \end{tikzpicture}\ =\ \begin{tikzpicture}[baseline=(current bounding box),scale=1.2]
        \filldraw[draw, circle, fill=white, minimum size=6pt](-0.6,0) circle (.1);
        \draw[rounded corners=8pt]
        (-.6,.1) -- (-.6,-.4) -- (1,-.4) -- (1,0);
        \fill (1,0) circle (1pt);

        \filldraw[draw, circle, fill=white, minimum size=6pt](1.3,0) circle (.1);
        \draw[rounded corners=8pt]
        (1.3,.1) -- (1.3,-.4) -- (2.9,-.4) -- (2.9,0);
        \fill (2.9,0) circle (1pt);

        \draw (-0.9,0) --  (3.1,0);
        \draw (0,.3) --  (0,.5);
        \draw (0,1.1) --  (0,1.3);
        \draw (0.5,.3) --  (0.5,1.3);
        \filldraw[fill = white] (-.3,-.3) rectangle node {$B'$} (.8,.3);
        \draw (1.9,.3) --  (1.9,1.3);
        \draw (2.4,.3) --  (2.4,.5);
        \draw (2.4,1.1) --  (2.4,1.3);
        \filldraw[fill = white] (1.6,-.3) rectangle node {$B'$} (2.7,.3);

        \draw (0.2,.5) --  (0.2,1.1);
        \draw (-.3,1.1) --  (0.2,1.1);
        \draw (-.3,.5) --  (0.2,.5);
        \node at (0,.8) {$U_Z$};
        \filldraw[fill = white] (.3,.5) rectangle node {} (2.1,1.1);
        \node at (.5,.8) {$C$};
        \node at (1.9,.8) {$U_Z$};
        \draw (2.2,.5) --  (2.2,1.1);
        \draw (2.2,1.1) --  (2.7,1.1);
        \draw (2.2,.5) --  (2.7,.5);
        \node at (2.4,.8) {$C$};
        \node at (-1.2,0) {$\cdots$};
        \node at (3.4,0) {$\cdots$};
    \end{tikzpicture} \ .
\eeq
Combining the following identity 
\beq
    \begin{tikzpicture}[baseline=(current bounding box),scale=1.5]
        \filldraw[draw, circle, fill=white, minimum size=6pt](-0.4,.8) circle (.1);
        \draw[rounded corners=8pt]
        (-.4,.9) -- (-.4,0) -- (1.2,0) -- (1.2,.4);
        \fill (1.2,.4) circle (1pt);
        \draw[rounded corners=8pt]
        (-.8,.8) -- (1,.8) -- (1,1.3);
        \draw[rounded corners=8pt]
        (0,.1) -- (0,.4) -- (1.6,.4);
        \draw (0,-.3) -- (0,-.1);
        \filldraw[fill = myRed] (.3,.2) rectangle node {} (.8,1);
        \node at (.55,.4) {$X$};
        \node at (.55,.8) {$C$};
    \end{tikzpicture}\ = \ \begin{tikzpicture}[baseline=(current bounding box),scale=1.5]
        \draw[rounded corners=8pt]
        (-.8,.8) -- (0,.8) -- (0,-.3);
        \draw[rounded corners=8pt] (1,-.1) --(0.4,-.1) -- (.4,1.3);
        \filldraw[fill = myRed] (-.2,0.1) rectangle node {} (.6,.6);
        \node at (0,.35) {$C$};
        \node at (.4,.35) {$X$};
    \end{tikzpicture}
\eeq
with Eq.~\eqref{eq:RI-RIRX}, we obtain
\begin{equation}
\begin{aligned}
\begin{tikzpicture}[baseline=(current bounding box),scale=1.2]
        \node at (-.4,0) {$\cdots$};
        \draw (1,.3) -- (1,1.2);
        \draw (2.4,.3) -- (2.4,1.2);
        \draw (.7,0) --  (-.1,0);
        \filldraw[fill = white] (.7,-.3) rectangle node {$A$} (1.3,.3);
        \draw (1.3,0) --  (2.1,0);
        \filldraw[fill = white] (2.1,-.3) rectangle node {$A$} (2.7,.3);
        \draw (2.7,0) --  (3.1,0);
        \node at (3.4,0) {$\cdots$};
        \draw (.1,.8) -- (.1,1.2);
        \node at (.1,.6) {$\ket{+}$};
        \draw (.1+1.4,.8) -- (.1+1.4,1.2);
        \node at (.1+1.4,.6) {$\ket{+}$};
\end{tikzpicture}\
=\ \begin{tikzpicture}[baseline=(current bounding box),scale=1.2]
        \draw (-1.3,.3) -- (-1.3,1.2);

        \draw[rounded corners=8pt] (-1.3,-.2) -- (-1.3,-.5) -- (-2,-.5);
        \draw[rounded corners=8pt] (-2, -.8) -- (-1.2, -.8) -- (-1.2,-3.7);
        \draw[rounded corners=8pt]
        (-.5,-.2) -- (-.5,-2) -- (.2,-2) -- (.2,-3.7);
        \filldraw[fill = myRed] (-1.4,-1.5) rectangle node {} (-.3,-1);
        \draw[rounded corners=8pt] (-.4, -1.7) -- (.1, -1.7) -- (.1,-.2);
        \draw[rounded corners=8pt] (-.8,-3.7) -- (-.8,-1.7) -- (-.6, -1.7);
        \draw (.1,-.6) -- (.1,-.2);
        \node at (-1,-3.9) {$\ket{\psi_{M}}$};
        \node at (-1.2,-1.25) {$C$};
        \node at (-.5,-1.25) {$X$};

        \draw[rounded corners=8pt]
        (.9,-.2) -- (.9,-3.3) -- (1.5,-3.3);
        \filldraw[fill = myRed] (0,-2.8) rectangle node {} (1.1,-2.3);
        \draw (1, -3) -- (1.5, -3);
        \draw[rounded corners=8pt] (.6,-3.7) -- (.6,-3) -- (.8, -3);
        \draw (.1,-.6) -- (.1,-.2);
        \node at (.4,-3.9) {$\ket{\psi_{M}}$};
        \node at (.2,-2.55) {$C$};
        \node at (.9,-2.55) {$X$};

        \node at (-2.4,-.65) {$\cdots$};
        \node at (1.9,-3.15) {$\cdots$};
        
        \draw (-.5,.3) -- (-.5,1.2);
        \draw (.1,.3) -- (.1,1.2);
        \draw (.9,.3) -- (.9,1.2);
        \filldraw[fill = white] (-.1-1.4,.5) rectangle node {} (1.1-1.4,1);
        \node at (.1-1.4,.75) {$C$};
        \node at (.9-1.4,.75) {$U_Z$};
        \filldraw[fill = white] (-.1,.5) rectangle node {} (1.1,1);
        \node at (.1,.75) {$C$};
        \node at (.9,.75) {$U_Z$};
        
        \draw (-1.1,-.2) --  (-1.1,.3);
        \draw (-1.1,-.2) --  (-1.8,.-.2);
        \draw (-1.1,.3) --  (-1.8,.3);
        \node at (-1.6,.05) {$U$};
        \filldraw[fill = white] (-.7,-.2) rectangle node {$U$} (.3,.3);
        \draw (.7,-.2) --  (.7,.3);
        \draw (.7,-.2) --  (1.4,.-.2);
        \draw (.7,.3) --  (1.4,.3);
        \node at (1.2,.05) {$U$};
\end{tikzpicture}\ .
\end{aligned}  
\end{equation}  

\section{$RZ$ class MPS tensor}
\label{app:RZ}
The general form of a transfer matrix for $D=2$ with the $RZ$ pushing relation is
\beq
        T_A = \quad 
        &\lambda_0\ \begin{tikzpicture}[baseline=(current bounding box),scale=1.5]
        \draw[rounded corners=8pt]
        (-.8,.8) -- (0,.8) -- (0,0) -- (-.8,0);
        \draw[rounded corners=8pt]
        (1,.8) -- (.2,.8) -- (.2,0) -- (1,0);
        \end{tikzpicture}\ + \lambda_1 \
        \begin{tikzpicture}[baseline=(current bounding box),scale=1.5]
        \filldraw[opacity=0, minimum size=6pt](-.4,0) circle (6pt);
        \draw[rounded corners=8pt]
        (-.8,.8) -- (0,.8) -- (0,0) -- (-.8,0);
        \draw[rounded corners=8pt]
        (1,.8) -- (.2,.8) -- (.2,0) -- (1,0);
        \filldraw[draw, circle, fill=myBlue, minimum size=6pt](-0.4,0.8) circle (6pt);
        \node at (-0.4,0.8) {$Z$};
        \end{tikzpicture}\ + \lambda_2 \
        \begin{tikzpicture}[baseline=(current bounding box),scale=1.5]
        \filldraw[opacity=0, minimum size=6pt](-.4,0) circle (6pt);
        \draw[rounded corners=8pt]
        (-.8,.8) -- (0,.8) -- (0,0) -- (-.8,0);
        \draw[rounded corners=8pt]
        (1,.8) -- (.2,.8) -- (.2,0) -- (1,0);
        \filldraw[draw, circle, fill=myBlue, minimum size=6pt](0.6,.8) circle (6pt);
        \node at (0.6,.8) {$Z$};
        \end{tikzpicture}\ + \lambda_3 \
        \begin{tikzpicture}[baseline=(current bounding box),scale=1.5]
        \filldraw[opacity=0, minimum size=6pt](-.4,0) circle (6pt);
        \draw[rounded corners=8pt]
        (-.8,.8) -- (0,.8) -- (0,0) -- (-.8,0);
        \draw[rounded corners=8pt]
        (1,.8) -- (.2,.8) -- (.2,0) -- (1,0);
        \filldraw[draw, circle, fill=myBlue, minimum size=6pt](0.6,.8) circle (6pt);
        \node at (0.6,.8) {$Z$};
        \filldraw[draw, circle, fill=myBlue, minimum size=6pt](-0.4,0.8) circle (6pt);
        \node at (-0.4,0.8) {$Z$};
        \end{tikzpicture}\\[2ex]
        +\ &\lambda_4\ \begin{tikzpicture}[baseline=(current bounding box),scale=1.5]
        \filldraw[opacity=0, minimum size=6pt](-.4,0.8) circle (6pt);
        \draw[rounded corners=8pt]
        (-.8,.8) -- (0,.8) -- (0,0) -- (-.8,0);
        \draw[rounded corners=8pt]
        (1,.8) -- (.2,.8) -- (.2,0) -- (1,0);
        \filldraw[draw, circle, fill=myRed, minimum size=6pt](0.6,0) circle (6pt);
        \node at (0.6,0) {$X$};
        \filldraw[draw, circle, fill=myRed, minimum size=6pt](-.4,0) circle (6pt);
        \node at (-0.4,0) {$X$};
        \end{tikzpicture}\ + i\lambda_5 \
        \begin{tikzpicture}[baseline=(current bounding box),scale=1.5]
        \draw[rounded corners=8pt]
        (-.8,.8) -- (0,.8) -- (0,0) -- (-.8,0);
        \draw[rounded corners=8pt]
        (1,.8) -- (.2,.8) -- (.2,0) -- (1,0);
        \filldraw[draw, circle, fill=myBlue, minimum size=6pt](-0.4,0.8) circle (6pt);
        \node at (-0.4,0.8) {$Z$};
        \filldraw[draw, circle, fill=myRed, minimum size=6pt](0.6,0) circle (6pt);
        \node at (0.6,0) {$X$};
        \filldraw[draw, circle, fill=myRed, minimum size=6pt](-.4,0) circle (6pt);
        \node at (-0.4,0) {$X$};
        \end{tikzpicture}\ + i\lambda_6 \
        \begin{tikzpicture}[baseline=(current bounding box),scale=1.5]
        \draw[rounded corners=8pt]
        (-.8,.8) -- (0,.8) -- (0,0) -- (-.8,0);
        \draw[rounded corners=8pt]
        (1,.8) -- (.2,.8) -- (.2,0) -- (1,0);
        \filldraw[draw, circle, fill=myBlue, minimum size=6pt](0.6,.8) circle (6pt);
        \node at (0.6,.8) {$Z$};
        \filldraw[draw, circle, fill=myRed, minimum size=6pt](0.6,0) circle (6pt);
        \node at (0.6,0) {$X$};
        \filldraw[draw, circle, fill=myRed, minimum size=6pt](-.4,0) circle (6pt);
        \node at (-0.4,0) {$X$};
        \end{tikzpicture}\ + \lambda_7 \
        \begin{tikzpicture}[baseline=(current bounding box),scale=1.5]
        \draw[rounded corners=8pt]
        (-.8,.8) -- (0,.8) -- (0,0) -- (-.8,0);
        \draw[rounded corners=8pt]
        (1,.8) -- (.2,.8) -- (.2,0) -- (1,0);
        \filldraw[draw, circle, fill=myBlue, minimum size=6pt](0.6,.8) circle (6pt);
        \node at (0.6,.8) {$Z$};
        \filldraw[draw, circle, fill=myBlue, minimum size=6pt](-0.4,0.8) circle (6pt);
        \node at (-0.4,0.8) {$Z$};
        \filldraw[draw, circle, fill=myRed, minimum size=6pt](0.6,0) circle (6pt);
        \node at (0.6,0) {$X$};
        \filldraw[draw, circle, fill=myRed, minimum size=6pt](-0.4,0) circle (6pt);
        \node at (-0.4,0) {$X$};
        \end{tikzpicture}\ ,
        \label{eq:T_A_RZ_app}
\eeq
with eight real parameters: $\lambda_i$. Using the following identity
\beq
    \begin{tikzpicture}[baseline=(current bounding box),scale=1.5]
        \draw[rounded corners=8pt]
        (-.4,.8) -- (0,.8) -- (0,0) -- (-.4,0);
        \draw[rounded corners=8pt]
        (.6,.8) -- (.2,.8) -- (.2,0) -- (.6,0);
    \end{tikzpicture}\quad =\sum_{p,q=0,1} \frac{1}{2} \quad
    \begin{tikzpicture}[baseline=(current bounding box),scale=1.5]
        \draw[rounded corners=8pt]
        (-.7,0) -- (.7,0);
        \draw[rounded corners=8pt]
        (-.7,.8) -- (.7,.8);
        \filldraw[draw, circle, fill=myBlue, minimum size=6pt](-.3,.8) circle (6pt);
        \node at (-0.3,.8) {$Z^{p}$};
        \filldraw[draw, circle, fill=myRed, minimum size=6pt](.3,.8) circle (6pt);
        \node at (.3,.8) {$X^{q}$};
        \filldraw[draw, circle, fill=myBlue, minimum size=6pt](-.3,0) circle (6pt);
        \node at (-0.3,0) {$Z^{p}$};
        \filldraw[draw, circle, fill=myRed, minimum size=6pt](.3,0) circle (6pt);
        \node at (.3,0) {$X^{q}$};
    \end{tikzpicture}\ ,
\eeq
the general transfer matrix can be rewritten as
\beq
    & \quad\quad\quad\quad\quad\quad\quad\quad T_A = \sum_{p,q,p',q'=0,1}S_{p',q';p,q}  \quad
    \begin{tikzpicture}[baseline=(current bounding box),scale=1.5]
        \draw[rounded corners=8pt]
        (-.7,0) -- (.7,0);
        \draw[rounded corners=8pt]
        (-.7,.8) -- (.7,.8);
        \filldraw[draw, circle, fill=myBlue, minimum size=6pt](-.3,.8) circle (6pt);
        \node at (-0.3,.8) {$Z^{p'}$};
        \filldraw[draw, circle, fill=myRed, minimum size=6pt](.3,.8) circle (6pt);
        \node at (.3,.8) {$X^{q'}$};
        \filldraw[draw, circle, fill=myBlue, minimum size=6pt](-.3,0) circle (6pt);
        \node at (-0.3,0) {$Z^{p}$};
        \filldraw[draw, circle, fill=myRed, minimum size=6pt](.3,0) circle (6pt);
        \node at (.3,0) {$X^{q}$};
    \end{tikzpicture}\ , 
\eeq
where $S_{p',q';p,q}$ are the elements of the $4\times 4$ matrix $S$ in the basis $\{(p,q) = (0,0), (1,0), (0,1), (1,1)\}  $:

\begin{equation}
        S = \frac{1}{2}\begin{pmatrix}
        \lambda_0+\lambda_3+\lambda_4+\lambda_7 & \lambda_1+\lambda_2-i\lambda_5-i\lambda_6 & 0 & 0\\
        \lambda_1+\lambda_2+i\lambda_5+i\lambda_6 & \lambda_0+\lambda_3-\lambda_4-\lambda_7 & 0 & 0\\
        0 & 0 & \lambda_0-\lambda_3+\lambda_4-\lambda_7 & \lambda_1-\lambda_2-i\lambda_5+i\lambda_6\\
        0 & 0 & \lambda_1-\lambda_2+i\lambda_5-i\lambda_6 & \lambda_0-\lambda_3-\lambda_4+\lambda_7
    \end{pmatrix}.
\end{equation}

Diagonalizing this $S$ matrix, we label the four eigenvalues and corresponding eigenstates by a pair $(m,n)$:
\beq
    s_{(0,0)} &= \frac{1}{2}(\lambda_0+\lambda_3-\sqrt{(\lambda_1+\lambda_2)^2+(\lambda_5+\lambda_6)^2+(\lambda_4+\lambda_7)^2 }),\\
    s_{(0,1)} &= \frac{1}{2}(\lambda_0+\lambda_3+\sqrt{(\lambda_1+\lambda_2)^2+(\lambda_5+\lambda_6)^2+(\lambda_4+\lambda_7)^2 }),\\
    s_{(1,0)} &= \frac{1}{2}(\lambda_0-\lambda_3-\sqrt{(\lambda_1-\lambda_2)^2+(\lambda_5-\lambda_6)^2+(\lambda_4-\lambda_7)^2 }),\\
    s_{(1,1)} &= \frac{1}{2}(\lambda_0-\lambda_3+\sqrt{(\lambda_1-\lambda_2)^2+(\lambda_5-\lambda_6)^2+(\lambda_4-\lambda_7)^2 }),\\
    v_{(0,0)} &= (\frac{\lambda_4+\lambda_7-\sqrt{(\lambda_1+\lambda_2)^2+(\lambda_5+\lambda_6)^2+(\lambda_4+\lambda_7)^2 }}{(\lambda_1+\lambda_2)+i(\lambda_5+\lambda_6)},1,0,0)^T,\\
    v_{(0,1)} &= (\frac{\lambda_4+\lambda_7+\sqrt{(\lambda_1+\lambda_2)^2+(\lambda_5+\lambda_6)^2+(\lambda_4+\lambda_7)^2 }}{(\lambda_1+\lambda_2)+i(\lambda_5+\lambda_6)},1,0,0)^T,\\
    v_{(1,0)} &= (0,0, \frac{\lambda_4-\lambda_7-\sqrt{(\lambda_1-\lambda_2)^2+(\lambda_5-\lambda_6)^2+(\lambda_4-\lambda_7)^2 }}{(\lambda_1-\lambda_2)+i(\lambda_5-\lambda_6)},1)^T,\\
    v_{(1,1)} &= (0,0, \frac{\lambda_4-\lambda_7+\sqrt{(\lambda_1-\lambda_2)^2+(\lambda_5-\lambda_6)^2+(\lambda_4-\lambda_7)^2 }}{(\lambda_1-\lambda_2)+i(\lambda_5-\lambda_6)},1)^T.
\eeq

The transfer matrix can therefore be rewritten as
\beq
    T_A = \sum_{m,n=0,1}s_{m,n}  \quad
    \begin{tikzpicture}[baseline=(current bounding box),scale=1.5]
        \draw[rounded corners=8pt]
        (-.7,0) -- (.7,0);
        \draw[rounded corners=8pt]
        (-.7,.8) -- (.7,.8);
        \filldraw[draw, circle, fill=white, minimum size=6pt](0,.8) circle (10pt);
        \node at (0,.8) {$\bar{V}_{(m,n)}$};
        \filldraw[draw, circle, fill=white, minimum size=6pt](0,0) circle (10pt);
        \node at (0,0) {$V_{(m,n)}$};
    \end{tikzpicture}\ ,
\eeq
where $V_{(m,n)}\equiv v^{\text{norm} *}_{(m,n)}\cdot (I, Z, X, ZX)$ with normalized eigenvectors:
\beq
    V_{(0,0)} &= \cos \theta_{0}\cdot  I + i e^{i \beta_{0}}\sin \theta_{0}\cdot  Z,\quad
    V_{(0,1)} = Z(\cos \theta_{0}\cdot  I + i e^{-i \beta_{0}}\sin \theta_{0}\cdot Z),\\
    V_{(1,0)} &= X(\cos\theta_{1}\cdot  I - i e^{i \beta_{1}}\sin \theta_{1}\cdot  Z),\quad
    V_{(1,1)} = XZ(\cos\theta_{1}\cdot  I - i e^{-i \beta_{1}}\sin \theta_{1}\cdot  Z).
\eeq
The angels $\theta$'s and $\beta$'s are given from $\lambda$'s:
\beq
    \theta_{0} &= - \arctan \frac{\sqrt{(\lambda_1+\lambda_2)^2+(\lambda_5+\lambda_6)^2}}{\lambda_4+\lambda_7-\sqrt{(\lambda_1+\lambda_2)^2+(\lambda_5+\lambda_6)^2+(\lambda_4+\lambda_7)^2 }},\quad \beta_{0} = \arctan \frac{\lambda_1+\lambda_2}{\lambda_5+\lambda_6},\\
    \theta_{1} &= -\arctan\frac{\sqrt{(\lambda_1-\lambda_2)^2+(\lambda_5-\lambda_6)^2}}{\lambda_4-\lambda_7-\sqrt{(\lambda_1-\lambda_2)^2+(\lambda_5-\lambda_6)^2+(\lambda_4-\lambda_7)^2 }},\quad \beta_{1} = \arctan \frac{\lambda_1-\lambda_2}{\lambda_5-\lambda_6}.
\eeq

Therefore, we can write a representative MPS tensor as 
\beq
        A=\sum_{m,n=0,1}\sqrt{s_{m,n}}\ \begin{tikzpicture}[baseline=(current bounding box),scale=1.5]
        \draw (-.8,0.4) -- (.7,.4);
        \draw (.9,0.4) -- (1.1,.4);     
        \draw (1.3,0.4) -- (1.9,.4);
        \draw (.8,-0.4) -- (.8,1.5);
        \draw (1.2,-0.4) -- (1.2,1.5);
        \node at (.8,-.6) {$\ket{m}$};
        \node at (1.2,-.6) {$\ket{n}$};
        \filldraw[draw, circle, fill=white, minimum size=6pt](0,0.4) circle (10pt);
        \node at (0,0.4) {$V_{(m,n)}$};
        \end{tikzpicture}\ = \
        \begin{tikzpicture}[baseline=(current bounding box),scale=1.5]
        \draw (-.8,0.4) -- (-.1,.4);
        \draw (.1,0.4) -- (1.1,.4);
        \draw (1.3,0.4) -- (3,.4);
        \draw[rounded corners=8pt]
        (0,-.4) -- (0,.8) -- (1,.8) -- (1, 1.2) -- (2.3, 1.2) -- (2.3, 1.5);
        \draw[rounded corners=8pt]
        (.3,-.4) -- (0.3, 0) -- (1.2,0) -- (1.2,.8) -- (2.6,.8) -- (2.6, 1.5);
        \filldraw[fill = myRed] (.3,.2) rectangle node {} (.8,1);
        \node at (.55,.8) {$C$};
        \node at (.55,.4) {$X$};
        \filldraw[fill = myBlue] (1.5,.2) rectangle node {$\mathcal{E}$} (2,1.4);
        \node at (0.15,-.6) {$\ket{\psi_{s}}$};
        \end{tikzpicture}
        \label{eq:A_RZ}
\eeq  
with $\ket{\psi_{s}}= \sum_{m,n=0,1}\sqrt{s_{m,n}}\ket{m}\otimes \ket{n}$, and the three-qubit non-unitary gate $\mathcal{E}$ is defined as
\beq
    \mathcal{E} &= \sum_{m,n=0,1} \ket{m}\bra{m}\otimes\ket{n}\bra{n}\otimes Z^n \big(\cos \theta_{m} I + i (-1)^m e^{i(-1)^n\beta_{m}}\sin\theta_{m} Z\big).
    \label{eq:E}
\eeq

In particular, for class $RZ$--$RIRX$ discussed in Sec.~\ref{subsec:RZ-RIRX}, $\lambda_1=\lambda_2=0$, hence $\beta$'s are all zero and $\mathcal{E}$ becomes unitary. A general family of MPS tensors $A'$ describes this class:
\beq
        A'=\ \begin{tikzpicture}[baseline=(current bounding box),scale=1.5]
        \draw (-.3,0.4) -- (-.1,.4);
        \draw (.1,0.4) -- (2.4,.4);
        \draw (2.6,0.4) -- (4.8,.4);
        \draw[rounded corners=8pt]
        (0,-.4) -- (0,.8) -- (2.4,.8) -- (2.4, 1.2) -- (4.4, 1.2) -- (4.4, 1.5);
        \draw[rounded corners=8pt]
        (.3,-.4) -- (0.3, 0) -- (2.5,0) -- (2.5,.8) -- (4.7,.8) -- (4.7, 1.5);
        \filldraw[fill = myRed] (.3,.2) rectangle node {} (.8,1); 
        \node at (.55,.8) {$C$};
        \node at (.55,.4) {$X$};
        \filldraw[draw, circle, fill=myBlue, minimum size=6pt](1.2,0.4) circle (6pt);
        \node at (1.2,0.4) {$R^z_{\alpha_{1}}$};
        \filldraw[fill = myBlue] (1.6,.2) rectangle node {} (2.2,1); 
        \node at (1.9,.8) {$C$};
        \node at (1.9,.4) {$R^z_{\alpha_{2}}$};
        \filldraw[fill = myBlue] (2.8,.2) rectangle node {} (3.4,1); 
        \node at (3.1,.8) {$C$};
        \node at (3.1,.4) {$R^z_{\alpha_{3}}$};
        \filldraw[fill = myBlue] (3.6,.2) rectangle node {} (4.2,1.4); 
        \node at (3.9,1.2) {$C$};
        \node at (3.9,.8) {$C$};
        \node at (3.9,.4) {$R^z_{\alpha_{4}}$};
        \node at (0.15,-.6) {$\ket{\psi_{r}}$};
        \end{tikzpicture}
\eeq  
where $\ket{\psi_r}=\sum_{m,n}r_{m,n}\ket{m}\otimes\ket{n}$, and the (controlled-)rotation gates are defined as
\beq
    &R^z_{\alpha} = e^{i\frac{\alpha Z}{2}},\quad CR^z_{\alpha} = e^{i\frac{\alpha (1-Z_1)Z_2}{4}},\quad CCR^z_{\alpha} = e^{i\frac{\alpha (1-Z_1)(1-Z_2)Z_3}{8}}.
\eeq
Although the above MPS tensor seems to have more parameters than Eq.~\eqref{eq:A_RZ} whose gate $\mathcal{E}$ is Eq.~\eqref{eq:E} with $\beta_m\equiv 0$, it has the same transfer matrix as that of class $RZ$--$RIRX$, see Sec.~\ref{subsec:RZ-RIRX}.

\section{Generalized cosine symmetries}
\label{app:cosine}
In this Appendix, we discuss the generalized cosine symmetries. They are related to the MPO that prepares states of class $RZ$--$RIRX$ as in Eq.~\eqref{eq:A_RZ-RZRX} with $\alpha_2 = \alpha_4 = 0$:
\beq
        \begin{tikzpicture}[baseline=(current bounding box),scale=1.5]
        \draw (-.3,0.4) -- (-.1,.4);
        \draw (.1,0.4) -- (1.5,.4);
        \draw (1.7,0.4) -- (3.5,.4);
        \draw[rounded corners=8pt]
        (0,-.4) -- (0,.8) -- (1.4,.8) -- (1.4, 1.2) -- (3, 1.2) -- (3, 1.5);
        \draw[rounded corners=8pt]
        (.3,-.4) -- (0.3, 0) -- (1.6,0) -- (1.6,.8) -- (3.3,.8) -- (3.3, 1.5);
        \filldraw[fill = myRed] (.3,.2) rectangle node {} (.8,1); 
        \node at (.55,.8) {$C$};
        \node at (.55,.4) {$X$};
        \filldraw[draw, circle, fill=myBlue, minimum size=6pt](1.2,0.4) circle (6pt);
        \node at (1.2,0.4) {$R^z_{\alpha}$};
        \filldraw[fill = myBlue] (1.8,.2) rectangle node {} (2.4,1); 
        \node at (2.1,.8) {$C$};
        \node at (2.1,.4) {$R^z_{\alpha'}$};
        \end{tikzpicture}\quad \cong\quad 
        \begin{tikzpicture}[baseline=(current bounding box),scale=1.5]
        \draw (-.3,0.4) -- (-.1,.4);
        \draw (.1,0.4) -- (1.5,.4);
        \draw (1.7,0.4) -- (3.5,.4);
        \draw[rounded corners=8pt]
        (0,-.4) -- (0,.8) -- (1.4,.8) -- (1.4, 1.2) -- (3, 1.2) -- (3, 1.5);
        \draw[rounded corners=8pt]
        (.3,-.4) -- (0.3, 0) -- (1.6,0) -- (1.6,.8) -- (3.3,.8) -- (3.3, 1.5);
        \filldraw[fill = myBlue] (.3,.2) rectangle node {} (.8,1); 
        \node at (.55,.8) {$C$};
        \node at (.55,.4) {$Z$};
        \filldraw[draw, circle, fill=myRed, minimum size=6pt](1.2,0.4) circle (6pt);
        \node at (1.2,0.4) {$R^x_{\alpha}$};
        \filldraw[fill = myRed] (1.8,.2) rectangle node {} (2.4,1); 
        \node at (2.1,.8) {$C$};
        \node at (2.1,.4) {$R^x_{\alpha'}$};
        \end{tikzpicture}
        \label{eq:MPO_alpha13}
\eeq 
where the two sides are related by a gauge transformation of the Hadamard gate. We define the MPO tensor in Eq.~\eqref{eq:MPO_alpha13} conjugated by a Hadamard gate on the left physical bond as $L_{\alpha,\alpha'}$. Using the above result, we have 
\beq
    L_{\alpha,\alpha'}\cong \quad 
    \begin{tikzpicture}[baseline=(current bounding box),scale=1.5]
        \draw (-.6,0.4) -- (-.1,.4);
        \draw (.1,0.4) -- (1.5,.4);
        \draw (1.7,0.4) -- (3.5,.4);
        \draw[rounded corners=8pt]
        (0,-.6) -- (0,.8) -- (1.4,.8) -- (1.4, 1.2) -- (3, 1.2) -- (3, 1.5);
        \draw[rounded corners=8pt]
        (.3,-.6) -- (0.3, 0) -- (1.6,0) -- (1.6,.8) -- (3.3,.8) -- (3.3, 1.5);
        \filldraw[fill = myRed] (.3,.2) rectangle node {} (.8,1); 
        \node at (.55,.8) {$X$};
        \node at (.55,.4) {$C$};
        \filldraw[draw, circle, fill=myRed, minimum size=6pt](1.2,0.4) circle (6pt);
        \node at (1.2,0.4) {$R^x_{\alpha}$};
        \filldraw[fill = myRed] (1.8,.2) rectangle node {} (2.4,1); 
        \node at (2.1,.8) {$C$};
        \node at (2.1,.4) {$R^x_{\alpha'}$};
        \filldraw[draw, circle, fill=white, minimum size=6pt](1.1,-.4) circle (.1);
        \draw (1,-.4) -- (3.5,-.4);
        \draw (-.6,-.4) -- (0,-.4);
        \draw (1.1,-.5) -- (1.1,-.3);
        \fill (0,-.4) circle (1pt);
        \node at (3.6,.4) {.};
        \end{tikzpicture}
\eeq
Moreover, using the following tensor identity:
\beq
     \begin{tikzpicture}[baseline=(current bounding box),scale=1.5]
     \filldraw[opacity=0, minimum size=6pt](0,-.7) circle (1pt);
        \draw (-.6,.8) -- (1,.8);
        \draw (1.3,.8)-- (1.8,.8);
        \draw[rounded corners=8pt]
        (0,-.4) -- (0,.4) -- (1.1,.4) -- (1.1,1.3);
        \filldraw[fill = myRed] (.3,.2) rectangle node {} (.8,1);
        \node at (.55,.8) {$C$};
        \node at (.55,.4) {$X$};
        \filldraw[draw, circle, fill=white, minimum size=6pt](1.1,-.2) circle (.1);
        \draw (1,-.2) -- (1.8,-.2);
        \draw (-.6,-.2) -- (0,-.2);
        \draw (1.1,-.3) -- (1.1,-.1);
        \fill (0,-.2) circle (1pt);
        \end{tikzpicture}\quad = \quad\begin{tikzpicture}[baseline=(current bounding box),scale=1.5]
     \filldraw[opacity=0, minimum size=6pt](0,-1.3) circle (1pt);
        \draw (-.6,.8) -- (1,.8);
        \draw (1.3,.8)-- (1.8,.8);
        \draw[rounded corners=8pt]
        (0,-.8) -- (0,.4) -- (1.1,.4) -- (1.1,1.6);
        \filldraw[fill = myBlue] (.3,.2) rectangle node {} (.8,1);
        \node at (.55,.8) {$C$};
        \node at (.55,.4) {$Z$};
        \filldraw[draw, circle, fill=myGray, minimum size=6pt](1.1,1.2) circle (6pt);
        \node at (1.1,1.2) {$H$};
        \filldraw[draw, circle, fill=myGray, minimum size=6pt](0,-.1) circle (6pt);
        \node at (0,-.1) {$H$};
        \filldraw[draw, circle, fill=white, minimum size=6pt](1.1,-.5) circle (.1);
        \draw (1,-.5) -- (1.8,-.5);
        \draw (-.6,-.5) -- (0,-.5);
        \draw (1.1,-.6) -- (1.1,-.4);
        \fill (0,-.5) circle (1pt);
        \end{tikzpicture}\quad = \quad
        \begin{tikzpicture}[baseline=(current bounding box),scale=1.5]
        \draw[rounded corners=8pt]
        (-1,.8) -- (1.1,.8) -- (1.1,1.3);
        \draw[rounded corners=8pt]
        (-1,-.8) -- (1.1,-.8) -- (1.1,-1.3);
        \draw[rounded corners=8pt]
        (2.4,-.4) -- (0,-.4) -- (0,.4) -- (2.4,.4);
        \filldraw[fill = myBlue] (.3,.2) rectangle node {} (.8,1);
        \node at (.55,.8) {$C$};
        \node at (.55,.4) {$Z$};
        \filldraw[draw, circle, fill=myGray, minimum size=6pt](1.2,.4) circle (6pt);
        \node at (1.2,.4) {$H$};
        \filldraw[fill = myBlue] (.3,-1) rectangle node {} (.8,-.2);
        \node at (.55,-.8) {$C$};
        \node at (.55,-.4) {$Z$};
        \filldraw[draw, circle, fill=myGray, minimum size=6pt](1.2,-.4) circle (6pt);
        \node at (1.2,-.4) {$H$};
        \filldraw[fill = myRed] (1.6,-.6) rectangle node {} (2.1,.6);
        \node at (1.85,.4) {$X$};
        \node at (1.85,-.4) {$C$};
        \filldraw[fill = myRed] (-.7,-1) rectangle node {} (-.2,1);
        \node at (-.45,.8) {$X$};
        \node at (-.45,-.8) {$C$};
        \end{tikzpicture}
\eeq
we can write the MPO tensor $L_{\alpha,\alpha'}$ as 
\beq
    L_{\alpha,\alpha'}\cong \quad 
        \begin{tikzpicture}[baseline=(current bounding box),scale=1.5]
        \draw[rounded corners=8pt]
        (-1.2,1.2) -- (1.1,1.2) -- (1.1,1.6);
        \draw[rounded corners=8pt]
        (-1.2,-1.2) -- (1.1,-1.2) -- (1.1,-1.6);
        \draw[rounded corners=8pt]
        (1.7,-.8) -- (0,-.8) -- (0,.8) -- (1.7,.8);
        \draw (1.9,-.8) -- (3,-.8);
        \draw (1.9,.8) -- (3,.8);
        \draw (1.8,-1.6) -- (1.8,1.6);
        \filldraw[fill = myBlue] (.3,.6) rectangle node {} (.8,1.4);
        \node at (.55,1.2) {$C$};
        \node at (.55,.8) {$Z$};
        \filldraw[draw, circle, fill=myGray, minimum size=6pt](1.2,.8) circle (6pt);
        \node at (1.2,.8) {$H$};
        \filldraw[fill = myBlue] (.3,-1.4) rectangle node {} (.8,-.6);
        \node at (.55,-1.2) {$C$};
        \node at (.55,-.8) {$Z$};
        \filldraw[draw, circle, fill=myGray, minimum size=6pt](1.2,-.8) circle (6pt);
        \node at (1.2,-.8) {$H$};
        \filldraw[draw, circle, fill=myBlue, minimum size=6pt](0,-.4) circle (6pt);
        \node at (0,-.4) {$R^z_{\alpha}$};
        \filldraw[fill = myBlue] (-.2,0) rectangle node {} (2,.5); 
        \node at (1.8,.25) {$C$};
        \node at (0,.25) {$R^z_{\alpha'}$};
        \filldraw[fill = myRed] (2.2,-1) rectangle node {} (2.7,1);
        \node at (2.45,.8) {$X$};
        \node at (2.45,-.8) {$C$};
        \filldraw[fill = myRed] (-.9,-1.4) rectangle node {} (-.4,1.4);
        \node at (-.65,1.2) {$X$};
        \node at (-.65,-1.2) {$C$};
        \end{tikzpicture}
        \quad \cong \quad
    \begin{tikzpicture}[baseline=(current bounding box),scale=1.5]
        \draw[rounded corners=8pt]
        (-.8,1.2) -- (1.1,1.2) -- (1.1,1.6);
        \draw[rounded corners=8pt]
        (-.8,-1.2) -- (1.1,-1.2) -- (1.1,-1.6);
        \draw[rounded corners=8pt]
        (1.7,-.8) -- (0,-.8) -- (0,.8) -- (1.7,.8);
        \draw (1.9,-.8) -- (2.4,-.8);
        \draw (1.9,.8) -- (2.4,.8);
        \draw (1.8,-1.6) -- (1.8,1.6);
        \filldraw[fill = myBlue] (.3,.6) rectangle node {} (.8,1.4);
        \node at (.55,1.2) {$C$};
        \node at (.55,.8) {$Z$};
        \filldraw[draw, circle, fill=myGray, minimum size=6pt](1.2,.8) circle (6pt);
        \node at (1.2,.8) {$H$};
        \filldraw[fill = myBlue] (.3,-1.4) rectangle node {} (.8,-.6);
        \node at (.55,-1.2) {$C$};
        \node at (.55,-.8) {$Z$};
        \filldraw[draw, circle, fill=myGray, minimum size=6pt](1.2,-.8) circle (6pt);
        \node at (1.2,-.8) {$H$};
        \filldraw[draw, circle, fill=myBlue, minimum size=6pt](0,-.4) circle (6pt);
        \node at (0,-.4) {$R^z_{\alpha}$};
        \filldraw[fill = myBlue] (-.2,0) rectangle node {} (2,.5); 
        \node at (1.8,.25) {$C$};
        \node at (0,.25) {$R^z_{\alpha'}$};
        \end{tikzpicture}
\eeq

Therefore, the MPO is given by a composition of KW, single-qubit rotation (and controlled-rotation), and $\text{KW}^{\dagger}$,
\beq
    L_{\alpha, \alpha'} = \text{KW}^{\dagger} \big(\prod_i R^z_{i, \alpha} CR^z_{\tilde{i} i, \alpha'} \big) \text{KW},
\eeq
where $R^z_{i, \alpha_1}$ is a single-qubit rotation on the first qubit of site $i$, and $CR^z_{\tilde{i} i, \alpha_3}$ is the controlled-rotation from the second to the first qubit of site $i$. As a result, the fusion rule of this MPO is given by
\beq
    L_{\alpha, \alpha'}L_{\alpha'', \alpha'''} &= \text{KW}^{\dagger} \big(\prod_i R^z_{i, \alpha} CR^z_{\tilde{i} i, \alpha'} \big) \text{KW} \cdot \text{KW}^{\dagger} \big(\prod_i R^z_{i, \alpha''} CR^z_{\tilde{i} i, \alpha'''} \big) \text{KW} \\
    &= \text{KW}^{\dagger} \big(\prod_i R^z_{i, \alpha} CR^z_{\tilde{i} i, \alpha'} \big) (1+\eta) \cdot\big(\prod_i R^z_{i, \alpha''} CR^z_{\tilde{i} i, \alpha'''} \big) \text{KW}\\
    &= L_{\alpha+\alpha'', \alpha'+\alpha'''} + L_{\alpha-\alpha'', \alpha'-\alpha'''}.
\eeq

\section{$RI$ class MPS tensor}
\label{app:RI}
The general form of a transfer matrix that has $RI$ pushing relation is
\beq
    T_A=\sum_{a,b,c=0}^{D-1}\lambda_{a,b,c}\quad \begin{tikzpicture}[baseline=(current bounding box),scale=1.5]
        \draw[rounded corners=8pt]
        (-.8,.8) -- (0,.8) -- (0,0) -- (-.8,0);
        \draw[rounded corners=8pt]
        (1,.8) -- (.2,.8) -- (.2,0) -- (1,0);
        \filldraw[draw, circle, fill=myRed, minimum size=6pt](0.6,0) circle (6pt);
        \node at (0.6,0) {$X^{\bar{b}}$};
        \filldraw[draw, circle, fill=myBlue, minimum size=6pt](0.6,.8) circle (6pt);
        \node at (0.6,.8) {$Z^{c}$};
        \filldraw[draw, circle, fill=myBlue, minimum size=6pt](-0.4,0.8) circle (6pt);
        \node at (-0.4,0.8) {$Z^a$};
        \end{tikzpicture}\quad = \sum_{a,b,c=0}^{D-1}\lambda'_{a,b,c}\quad \begin{tikzpicture}[baseline=(current bounding box),scale=1.5]
        \draw[rounded corners=8pt]
        (-.8,1.6) -- (-0.2,1.6) -- (-0.2,-0.8) -- (-.8,-.8);
        \draw[rounded corners=8pt]
        (1,1.6) -- (.4,1.6) -- (.4,-0.8) -- (1,-.8);
        \filldraw[draw, circle, fill=myRed, minimum size=6pt](0.4,0.1) circle (6pt);
        \node at (0.4,0.1) {$X^{\bar{b}}$};
        \filldraw[draw, circle, fill=myRed, minimum size=6pt](-0.2,0.1) circle (6pt);
        \node at (-0.2,0.1) {$X^{\bar{b}}$};
        \filldraw[draw, circle, fill=myBlue, minimum size=6pt](0.4,.7) circle (6pt);
        \node at (0.4,.7) {$Z^{c}$};
        \filldraw[draw, circle, fill=myBlue, minimum size=6pt](-0.2,0.7) circle (6pt);
        \node at (-0.2,0.7) {$Z^a$};
        \filldraw[fill = myRed] (-.4,-.6) rectangle node {} (.6,-.2);
        \node at (-.2,-.4) {$X^{\dagger}$};
        \node at (.4,-.4) {$C$};
        \filldraw[fill = myRed] (-.4,1) rectangle node {} (.6,1.4);
        \node at (-.2,1.2) {$X$};
        \node at (.4,1.2) {$C$};
        \end{tikzpicture}\ .
\eeq
To be more explicit, for the general bond dimension $D=2$ transfer matrix with $RI$ pushing relation,
\beq
        T_A = \quad 
        &\lambda_0\ \begin{tikzpicture}[baseline=(current bounding box),scale=1.5]
        \draw[rounded corners=8pt]
        (-.8,.8) -- (0,.8) -- (0,0) -- (-.8,0);
        \draw[rounded corners=8pt]
        (1,.8) -- (.2,.8) -- (.2,0) -- (1,0);
        \end{tikzpicture}\ + \lambda_1 \
        \begin{tikzpicture}[baseline=(current bounding box),scale=1.5]
        \filldraw[opacity=0, minimum size=6pt](-.4,0) circle (6pt);
        \draw[rounded corners=8pt]
        (-.8,.8) -- (0,.8) -- (0,0) -- (-.8,0);
        \draw[rounded corners=8pt]
        (1,.8) -- (.2,.8) -- (.2,0) -- (1,0);
        \filldraw[draw, circle, fill=myBlue, minimum size=6pt](-0.4,0.8) circle (6pt);
        \node at (-0.4,0.8) {$Z$};
        \end{tikzpicture}\ + \lambda_2 \
        \begin{tikzpicture}[baseline=(current bounding box),scale=1.5]
        \filldraw[opacity=0, minimum size=6pt](-.4,0) circle (6pt);
        \draw[rounded corners=8pt]
        (-.8,.8) -- (0,.8) -- (0,0) -- (-.8,0);
        \draw[rounded corners=8pt]
        (1,.8) -- (.2,.8) -- (.2,0) -- (1,0);
        \filldraw[draw, circle, fill=myBlue, minimum size=6pt](0.6,.8) circle (6pt);
        \node at (0.6,.8) {$Z$};
        \end{tikzpicture}\ + \lambda_3 \
        \begin{tikzpicture}[baseline=(current bounding box),scale=1.5]
        \filldraw[opacity=0, minimum size=6pt](-.4,0) circle (6pt);
        \draw[rounded corners=8pt]
        (-.8,.8) -- (0,.8) -- (0,0) -- (-.8,0);
        \draw[rounded corners=8pt]
        (1,.8) -- (.2,.8) -- (.2,0) -- (1,0);
        \filldraw[draw, circle, fill=myBlue, minimum size=6pt](0.6,.8) circle (6pt);
        \node at (0.6,.8) {$Z$};
        \filldraw[draw, circle, fill=myBlue, minimum size=6pt](-0.4,0.8) circle (6pt);
        \node at (-0.4,0.8) {$Z$};
        \end{tikzpicture}\\[2ex]
        +\ &\lambda_4\ \begin{tikzpicture}[baseline=(current bounding box),scale=1.5]
        \filldraw[opacity=0, minimum size=6pt](-.4,0.8) circle (6pt);
        \draw[rounded corners=8pt]
        (-.8,.8) -- (0,.8) -- (0,0) -- (-.8,0);
        \draw[rounded corners=8pt]
        (1,.8) -- (.2,.8) -- (.2,0) -- (1,0);
        \filldraw[draw, circle, fill=myRed, minimum size=6pt](0.6,0) circle (6pt);
        \node at (0.6,0) {$X$};
        \end{tikzpicture}\ + \lambda_5 \
        \begin{tikzpicture}[baseline=(current bounding box),scale=1.5]
        \draw[rounded corners=8pt]
        (-.8,.8) -- (0,.8) -- (0,0) -- (-.8,0);
        \draw[rounded corners=8pt]
        (1,.8) -- (.2,.8) -- (.2,0) -- (1,0);
        \filldraw[draw, circle, fill=myBlue, minimum size=6pt](-0.4,0.8) circle (6pt);
        \node at (-0.4,0.8) {$Z$};
        \filldraw[draw, circle, fill=myRed, minimum size=6pt](0.6,0) circle (6pt);
        \node at (0.6,0) {$X$};
        \end{tikzpicture}\ + i\lambda_6 \
        \begin{tikzpicture}[baseline=(current bounding box),scale=1.5]
        \draw[rounded corners=8pt]
        (-.8,.8) -- (0,.8) -- (0,0) -- (-.8,0);
        \draw[rounded corners=8pt]
        (1,.8) -- (.2,.8) -- (.2,0) -- (1,0);
        \filldraw[draw, circle, fill=myBlue, minimum size=6pt](0.6,.8) circle (6pt);
        \node at (0.6,.8) {$Z$};
        \filldraw[draw, circle, fill=myRed, minimum size=6pt](0.6,0) circle (6pt);
        \node at (0.6,0) {$X$};
        \end{tikzpicture}\ + i\lambda_7 \
        \begin{tikzpicture}[baseline=(current bounding box),scale=1.5]
        \draw[rounded corners=8pt]
        (-.8,.8) -- (0,.8) -- (0,0) -- (-.8,0);
        \draw[rounded corners=8pt]
        (1,.8) -- (.2,.8) -- (.2,0) -- (1,0);
        \filldraw[draw, circle, fill=myBlue, minimum size=6pt](0.6,.8) circle (6pt);
        \node at (0.6,.8) {$Z$};
        \filldraw[draw, circle, fill=myBlue, minimum size=6pt](-0.4,0.8) circle (6pt);
        \node at (-0.4,0.8) {$Z$};
        \filldraw[draw, circle, fill=myRed, minimum size=6pt](0.6,0) circle (6pt);
        \node at (0.6,0) {$X$};
        \end{tikzpicture}\ ,
        \label{eq:T_A_RI_app}
\eeq
Using a conjugation of $XC$ gate with the control qubit on the right and the target qubit on the left, the form of the Pauli operators in the middle is the same as the $RZ$ cases in Eq.~\eqref{eq:T_A_RZ_app}, with $\lambda_1\leftrightarrow\lambda_3$ and $\lambda_5\leftrightarrow\lambda_7$. Therefore, we can use the result we have for $RZ$ cases (see Appendix \ref{app:RZ}) to write a representative $A$ tensor for general parameters: 
\beq
        A=\quad \begin{tikzpicture}[baseline=(current bounding box),scale=1.5]
        \filldraw[draw, circle, fill=white, minimum size=6pt](-0.4,0.4) circle (.1);
        \draw (-.8,0.4) -- (-.1,.4);
        \draw (.1,0.4) -- (1.1,.4);
        \draw (1.3,0.4) -- (3,.4);
        \draw[rounded corners=8pt]
        (0,-.4) -- (0,.8) -- (1,.8) -- (1, 1.2) -- (2.3, 1.2) -- (2.3, 1.5);
        \draw[rounded corners=8pt]
        (.3,-.4) -- (0.3, 0) -- (1.2,0) -- (1.2,.8) -- (2.6,.8) -- (2.6, 1.5);
        \draw[rounded corners=8pt]
        (-.4,.5) -- (-.4, -.2) -- (-0.1,-.2);
        \draw[rounded corners=8pt]
        (.4,-0.2) -- (2.4, -0.2) -- (2.4,0.4);
        \fill (2.4,0.4) circle (1pt);
        \draw (0.1,-.2) -- (0.2,-.2);
        \filldraw[fill = myRed] (.3,.2) rectangle node {} (.8,1);
        \node at (.55,.4) {$X$};
        \node at (.55,.8) {$C$};
        \filldraw[fill = myBlue] (1.5,.2) rectangle node {$\mathcal{E}'$} (2,1.4);
        \node at (0.15,-.6) {$\ket{\psi'_{s}}$};
        \end{tikzpicture}
\eeq  
where we use the composition of a delta tensor and a plus tensor to denote the controlled-$X$ gate. The two-qubit state and the three-qubit non-unitary gate are defined as
\beq
    \ket{\psi'_s} = \sum_{m,n} \sqrt{s'_{m,n}} \ket{m}\otimes\ket{n},\quad \mathcal{E}' = \sum_{m,n=0,1} \ket{m}\bra{m}\otimes\ket{n}\bra{n}\otimes Z^n \big(\cos \theta'_{m} I + i(-1)^m e^{i(-1)^n\beta'_{m}}\sin\theta'_{m} Z\big).
\eeq
The parameters are depend on the  $\lambda$'s as follows:
\beq
    \theta'_{0} &= \arctan \frac{\sqrt{(\lambda_3+\lambda_2)^2+(\lambda_7+\lambda_6)^2}}{\lambda_4+\lambda_5-\sqrt{(\lambda_3+\lambda_2)^2+(\lambda_7+\lambda_6)^2+(\lambda_4+\lambda_5)^2 }},\quad \beta'_{0} = \arctan \frac{\lambda_3+\lambda_2}{\lambda_7+\lambda_6},\\
    \theta'_{1} &= \arctan\frac{\sqrt{(\lambda_3-\lambda_2)^2+(\lambda_7-\lambda_6)^2}}{\lambda_4-\lambda_5-\sqrt{(\lambda_3-\lambda_2)^2+(\lambda_7-\lambda_6)^2+(\lambda_4-\lambda_5)^2 }},\quad \beta'_{1} = \arctan \frac{\lambda_3-\lambda_2}{\lambda_7-\lambda_6}.
\eeq
and 
\beq
    s'_{(0,0)} &= \frac{1}{2}(\lambda_0+\lambda_1-\sqrt{(\lambda_3+\lambda_2)^2+(\lambda_7+\lambda_6)^2+(\lambda_4+\lambda_5)^2 }),\\
    s'_{(0,1)} &= \frac{1}{2}(\lambda_0+\lambda_1+\sqrt{(\lambda_3+\lambda_2)^2+(\lambda_7+\lambda_6)^2+(\lambda_4+\lambda_5)^2 }),\\
    s'_{(1,0)} &= \frac{1}{2}(\lambda_0-\lambda_1-\sqrt{(\lambda_3-\lambda_2)^2+(\lambda_7-\lambda_6)^2+(\lambda_4-\lambda_5)^2 }),\\
    s'_{(1,1)} &= \frac{1}{2}(\lambda_0-\lambda_1+\sqrt{(\lambda_3-\lambda_2)^2+(\lambda_7-\lambda_6)^2+(\lambda_4-\lambda_5)^2 }).
\eeq
In Sec.~\ref{subsec:example_RI-RIRX}, we discuss a subclass of states by restricting the parameters $\lambda_1=\lambda_2=\lambda_6=\lambda_7=0$. In this case, the non-unitary three-qubit gate $\mathcal{E}'$ can be simplified to 
\beq
    \mathcal{E}' &= \sum_{m,n=0,1} \ket{m}\bra{m}\otimes\ket{n}\bra{n}\otimes Z^n \big(\cos \theta'_{m} I + (-1)^{m+n}\sin\theta'_{m} Z\big)\\
    &=\sum_{m,n=0,1} \ket{m}\bra{m}\otimes\ket{n}\bra{n}\otimes\frac{1}{\sqrt{\cosh{2\gamma_m}}} Z^n \big(\cosh \gamma_{m} I + (-1)^{m+n}\sinh\gamma_{m} Z\big) \\
    &\propto \sum_{m,n=0,1} \ket{m}\bra{m}\otimes\ket{n}\bra{n}\otimes\Big(\frac{\cosh{2\gamma_0}}{\cosh{2\gamma_1}}\Big)^{\frac{m}{2}} Z^n e^{(-1)^{m+n}\gamma_{m}Z}.
\eeq
where $\tanh{\gamma_m}\equiv\tan{\theta'_m}$. In the second equality above, we have used the fact that
\beq
    \frac{\cosh{\gamma_m}}{\sqrt{\cosh{2\gamma_m}}}=\eta_\theta\cos{\theta_m},\ \frac{\sinh{\gamma_m}}{\sqrt{\cosh{2\gamma_m}}}=\eta_\theta\sin{\theta_m},
\eeq
where $\eta_\theta = 1$ when $\theta_m\in (-\frac{\pi}{2},\frac{\pi}{2})$ and $\eta_\theta = -1$ otherwise.
Labeling the first, second, and third qubits respectively as $m,n,t$, up to a rescaling, this gate becomes
\beq
    \mathcal{E}' = CZ_{n,t} e^{\kappa Z_m+\gamma' Z_n Z_t+\gamma Z_m Z_n Z_t},
\eeq
where 
\beq
e^{ -4\kappa}=\frac{\cosh{2\gamma_0}}{\cosh{2\gamma_1}},\quad \gamma=\frac{\gamma_0+\gamma_1}{2},\quad \gamma'=\frac{\gamma_0-\gamma_1}{2}.
\eeq

Using the identity (up to an overall rescaling factor)
\beq
    \begin{tikzpicture}[baseline=(current bounding box),scale=1.5]
        \draw[opacity=0, rounded corners=8pt]
        (-.4,.3) -- (-.4,0.8) -- (0.4,0.8) -- (.4,0.4);
        \filldraw[draw, circle, fill=white, minimum size=6pt](-0.4,0.4) circle (.1);
        \draw (-.8,0.4) -- (.8,.4);
        \draw[rounded corners=8pt]
        (-.4,.5) -- (-.4,0) -- (0.4,0) -- (.4,0.4);
        \fill (.4,0.4) circle (1pt);
    \end{tikzpicture} \quad = \quad  
    \begin{tikzpicture}[baseline=(current bounding box),scale=1.5]
        \draw (-.8,0.4) -- (-.4,.4);
        \draw (.4,0.4) -- (.8,.4);
        \node at (-.2,.4) {$\ket{0}$};
        \node at (.2,.4) {$\bra{+}$};
    \end{tikzpicture}\ ,
\eeq
we can write the above $A$ tensor as
\beq
        A&=\quad \begin{tikzpicture}[baseline=(current bounding box),scale=1.5]
        \draw (-.6,0.4) -- (-.1,.4);
        \draw[rounded corners=8pt] (.1,0.4) -- (1,.4) -- (1,-0.4);
        \draw[rounded corners=8pt] (1.7,-0.4) -- (1.7,.4) -- (3.4,.4);
        \draw[rounded corners=8pt]
        (0,-.4) -- (0,.8) -- (1,.8) -- (1, 1.2) -- (2.7, 1.2) -- (2.7, 1.5);
        \draw[rounded corners=8pt]
        (.3,-.4) -- (.3,0) -- (0.9,0);
        \draw[rounded corners=8pt]
        (1.1, 0) -- (1.5,0) -- (1.5,.8) -- (3,.8) -- (3, 1.5);
        \filldraw[fill = myRed] (.3,.2) rectangle node {} (.8,1);
        \node at (.55,.4) {$X$};
        \node at (.55,.8) {$C$};
        \filldraw[fill = myBlue] (1.9,.2) rectangle node {$\mathcal{E}'$} (2.4,1.4);
        \node at (0.2,-.6) {$\ket{\psi'_{s}}$};
        \node at (1,-.6) {$\ket{0}$};
        \node at (1.7,-.6) {$\ket{+}$};
        \end{tikzpicture}\quad = \quad \begin{tikzpicture}[baseline=(current bounding box),scale=1.5]
        \draw[rounded corners=8pt] (-0.6,0.8) -- (0,.8) -- (0,-0.4);
        \draw[rounded corners=8pt] (1.5,-0.4) -- (1.5,.4) -- (3.2,.4);
        \draw[rounded corners=8pt]
        (0.3,-.4) -- (.3,1.2) -- (1, 1.2) -- (2.5, 1.2) -- (2.5, 1.5);
        \draw[rounded corners=8pt]
        (.6,-.4) -- (.6, 0) -- (1.1,0) -- (1.1,.8) -- (2.8,.8) -- (2.8, 1.5);
        \filldraw[fill = myRed] (-.2,.2) rectangle node {$XC$} (.5,.6);
        \filldraw[fill = myBlue] (1.7,.2) rectangle node {$\mathcal{E}'$} (2.2,1.4);
        \node at (0.45,-.6) {$\ket{\psi'_{s}}$};
        \node at (0,-.6) {$\ket{0}$};
        \node at (1.5,-.6) {$\ket{+}$};
        \end{tikzpicture}\\[2ex]
        &=\quad \begin{tikzpicture}[baseline=(current bounding box),scale=1.5]
        \draw[rounded corners=8pt] (1.5,-0.4) -- (1.5,.4) -- (3.2,.4);
        \draw[rounded corners=8pt]
        (-.3,1.2) -- (1, 1.2) -- (2.5, 1.2) -- (2.5, 1.5);
        \draw (.3,1.2) -- (.3,-.4);
        \fill (.3,1.2) circle (1pt);
        \draw[rounded corners=8pt]
        (.6,-.4) -- (.6, 0) -- (1.1,0) -- (1.1,.8) -- (2.8,.8) -- (2.8, 1.5);
        \filldraw[fill = myBlue] (1.7,.2) rectangle node {$\mathcal{E}'$} (2.2,1.4);
        \node at (0.45,-.6) {$\ket{\psi'_{s}}$};
        \node at (1.5,-.6) {$\ket{+}$};
        \end{tikzpicture}
\eeq  
We can do further simplifications by noticing that $\mathcal{E}'$ being a phase gate can be pushed through the delta tensor, to another leg. We sketch the MPS of the above tensor $A$ as follows
\beq
    \ket{A} &= \cdots \quad \begin{tikzpicture}[baseline=(current bounding box),scale=1.5]
        \draw[rounded corners=8pt]
        (-.3,1.2) -- (2.5, 1.2) -- (2.5, 1.5);
        \draw (.3,1.2) -- (.3,-.4);
        \fill (.3,1.2) circle (1pt);
        \draw[rounded corners=8pt]
        (.6,-.4) -- (.6, 0) -- (1.1,0) -- (1.1,.8) -- (2.8,.8) -- (2.8, 1.5);
        \node at (0.45,-.6) {$\ket{\psi'_{s}}$};
        \node at (1.5,-.6) {$\ket{+}$};
        \draw[rounded corners=8pt] (1.5,-0.4) -- (1.5,.4) -- (3.1,.4) -- (3.1,1.2) -- (5.8,1.2) -- (5.8,1.5);
        \draw (3.6,1.2) -- (3.6,-.4);
        \fill (3.6,1.2) circle (1pt);
        \draw[rounded corners=8pt]
        (3.9,-.4) -- (3.9, 0) -- (4.4,0) -- (4.4,.8) -- (6.1,.8) -- (6.1, 1.5);
        \node at (3.75,-.6) {$\ket{\psi'_{s}}$};
        \node at (4.8,-.6) {$\ket{+}$};
        \draw[rounded corners=8pt] (4.8,-0.4) -- (4.8,.4) -- (6.5,.4);
        \filldraw[fill = myBlue] (1.7,.2) rectangle node {$\mathcal{E}'$} (2.2,1.4);
        \filldraw[fill = myBlue] (5,.2) rectangle node {$\mathcal{E}'$} (5.5,1.4);
        \end{tikzpicture}
        \quad \cdots \\[2ex]
        &= \cdots \quad \begin{tikzpicture}[baseline=(current bounding box),scale=1.5]
        \draw[rounded corners=8pt]
        (-.3,1.2) -- (2.5, 1.2) -- (2.5, 1.5);
        \draw[rounded corners=8pt] (-.3,0) -- (.3,0) -- (.3,-.4);
        \draw[rounded corners=8pt]
        (.6,-.4) -- (.6, 0) -- (1.1,0) -- (1.1,.8) -- (2.8,.8) -- (2.8, 1.5);
        \node at (0.45,-.6) {$\ket{\psi'_{s}}$};
        \node at (1.5,-.6) {$\ket{+}$};
        \draw[rounded corners=8pt] (1.5,-0.4) -- (1.5,.4) -- (3.1,.4) -- (3.1,1.2) -- (5.8,1.2) -- (5.8,1.5);
        \draw[rounded corners=8pt] (1.5,0) -- (3.6,0) -- (3.6,-.4);
        \fill (1.5,0) circle (1pt);
        \draw[rounded corners=8pt]
        (3.9,-.4) -- (3.9, 0) -- (4.4,0) -- (4.4,.8) -- (6.1,.8) -- (6.1, 1.5);
        \node at (3.75,-.6) {$\ket{\psi'_{s}}$};
        \node at (4.8,-.6) {$\ket{+}$};
        \draw[rounded corners=8pt] (4.8,-0.4) -- (4.8,.4) -- (6.5,.4);
        \filldraw[fill = myBlue] (1.7,.2) rectangle node {$\mathcal{E}'$} (2.2,1.4);
        \filldraw[fill = myBlue] (5,.2) rectangle node {$\mathcal{E}'$} (5.5,1.4);
        \fill (4.8,0) circle (1pt);
        \draw (4.8,0) -- (6.5,0);
        \end{tikzpicture}
        \quad \cdots \\[2ex]
        &= \cdots \quad \begin{tikzpicture}[baseline=(current bounding box),scale=1.5]
        \draw[rounded corners=8pt]
        (-.3,1.2) -- (2.5, 1.2) -- (2.5, 1.5);
        \draw[rounded corners=8pt]
        (-.3, 0) -- (1.1,0) -- (1.1,.8) -- (2.8,.8) -- (2.8, 1.5);
        \draw[rounded corners=8pt] (1.5,-0.4) -- (1.5,.4) -- (3.1,.4) -- (3.1,1.2) -- (5.8,1.2) -- (5.8,1.5);
        \draw[rounded corners=8pt]
        (1.8,-.4) -- (1.8, 0) -- (4.4,0) -- (4.4,.8) -- (6.1,.8) -- (6.1, 1.5);
        \node at (1.65,-.6) {$\ket{\psi'_{s}}$};
        \node at (4.95,-.6) {$\ket{\psi'_{s}}$};
        \draw[rounded corners=8pt] (4.8,-0.4) -- (4.8,.4) -- (6.5,.4);
        \filldraw[fill = myBlue] (1.7,.2) rectangle node {$\mathcal{E}'$} (2.2,1.4);
        \filldraw[fill = myBlue] (5,.2) rectangle node {$\mathcal{E}'$} (5.5,1.4);
        \draw[rounded corners=8pt] (5.1,-0.4) -- (5.1,0) -- (6.5,0);
        \end{tikzpicture}
        \quad \cdots 
        \label{eq:RI_non-unitary}
\eeq
Therefore, the state $\ket{A}$ with the $RI$ pushing relation is always prepared from some product state using a circuit of commuting three-qubit non-unitary $\mathcal{E}'$ gates.

\end{document}